\documentclass[letterpaper,11pt]{article}
\usepackage{style}
\usepackage[utf8]{inputenc}

\usepackage{dsfont}
\usepackage[sort,numbers]{natbib}
\usepackage[margin=1in,footskip=0.25in]{geometry}

\title{Exact Label Recovery in Euclidean Random Graphs \footnote{Part of this work appeared in the Symposium on Discrete Algorithms (SODA) 2024 \cite{Gaudio2024}.}}
\author{Julia Gaudio\thanks{(\url{julia.gaudio@northwestern.edu}) Department of Industrial Engineering and Management Sciences, Northwestern
University} \and Charlie Guan\thanks{(\url{charlie.guan@northwestern.edu}) Department of Industrial Engineering and Management Sciences, Northwestern
University}
\and Xiaochun Niu\thanks{(\url{xiaochunniu2024@u.northwestern.edu}) Department of Industrial Engineering and Management Sciences, Northwestern
University} \and Ermin Wei\thanks{(\url{ermin.wei@northwestern.edu}) Department of Electrical and Computer Engineering and Department of Industrial Engineering and Management Sciences, Northwestern
University}}

\date{}

\begin{document}
\maketitle

\begin{abstract}
In this paper, we propose a family of label recovery problems on weighted Euclidean random graphs. The vertices of a graph are embedded in $\mathbb{R}^d$ according to a Poisson point process, and are assigned to a discrete community label. Our goal is to infer the vertex labels, given edge weights whose distributions depend on the vertex labels as well as their geometric positions. Our general model provides a geometric extension of popular graph and matrix problems, including submatrix localization and $\mathbb{Z}_2$-synchronization, and includes the Geometric Stochastic Block Model (proposed by Sankararaman and Baccelli) as a special case. We study the fundamental limits of exact recovery of the vertex labels. Under a mild distinctness of distributions assumption, we determine the information-theoretic threshold for exact label recovery, in terms of a Chernoff--Hellinger divergence criterion. Impossibility of recovery below the threshold is proven by a unified analysis using a Cram\'er lower bound. Achievability above the threshold is proven via an efficient two-phase algorithm, where the first phase computes an almost-exact labeling through a local propagation scheme, while the second phase refines the labels. The information-theoretic threshold is dictated by the performance of the so-called \emph{genie estimator}, which decodes the label of a single vertex given all the other labels. This shows that our proposed models exhibit \emph{the local-to-global amplification} phenomenon.
\end{abstract}

\section{Introduction}
Inference problems on graphs and matrices are well-studied in the statistics, machine learning, and information-theoretic literature. Some prominent examples include the Stochastic Block Model (SBM), group synchronization, and the spiked Wigner model. In these models, graph vertices (or matrix indices) are associated with a hidden, discrete community label which we wish to estimate, and we are given observations for each pair of vertices. 
These models are related by a common property: the observations are independent, conditioned on the community labels, and the distribution of an observation on vertices $(u,v)$ is determined by the labels of $u$ and $v$.

Specifically, consider the SBM, which was introduced by Holland,
Laskey, and Leinhardt \citep{Holland1983}. The SBM is a probabilistic model that generates graphs with community structure, where edges are generated independently conditioned on community labels. Since its introduction, the SBM has been intensively studied, with many exciting developments over the past decade. Many community recovery problems are now well-understood; for example, the fundamental limits of the exact recovery problem are known, and there is a corresponding efficient algorithm that achieves those limits \citep{Abbe2015}. For an overview of theoretical developments and open questions, see the survey of Abbe \citep{Abbe2017}. 

While the SBM is a powerful model, its simplicity fails to capture \emph{transitive} behavior that occurs in real-world networks. In particular, social networks typically contain many triangles; a given pair of people are more likely to be friends if they have a friend in common \citep{rapoport1953spread}. The SBM does not capture this transitive behavior, since edges are formed independently, conditioned on the community assignments. To address this shortcoming,  Sankararaman and Baccelli \citep{Sankararaman2018} introduced a spatial random graph model, which we refer to as the \emph{Geometric Stochastic Block Model} (GSBM).

\subsection{The GSBM}
 In the GSBM, vertices are generated according to a Poisson point process in a bounded region of $\mathbb{R}^d$. Each vertex is randomly assigned one of two community labels, with equal probability. A given pair of vertices $(u,v)$ is connected by an edge with a probability that depends on  both the community labels of $u$ and $v$ as well as their distance. Edges are formed independently, conditioned on the community assignments and locations. The geometric embedding thus governs the transitive edge behavior. The goal is to determine the communities of the vertices, observing the edges and the locations. In a follow-up work, Abbe, Baccelli, and Sankararaman \citep{Abbe2021} studied both exact recovery in logarithmic-degree graphs and partial recovery in sparse graphs. Their work established a phase transition for both partial and exact recovery, in terms of the Poisson intensity parameter $\lambda$. The critical value of $\lambda$ was identified in some special cases of the sparse model, but a precise characterization of the information-theoretic threshold for exact recovery in the logarithmic degree regime was left open.
  
Our work resolves this gap, by identifying the information-theoretic threshold for exact recovery in the logarithmic degree regime (and confirming a conjecture of Abbe et al \citep{Abbe2021}). Additionally, we propose a polynomial-time algorithm achieving the information-theoretic threshold. The algorithm consists of two phases: the first phase produces a preliminary almost-exact labeling through a local label propagation scheme, while the second phase refines the initial labels to achieve exact recovery. At a high level, the algorithm bears some similarity to prior works on the SBM using a two-phase approach \citep{Abbe2015,Mossel2015}. Our work therefore shows that just like the SBM, the GSBM exhibits the so-called \emph{local-to-global amplification} phenomenon \citep{Abbe2017}. This means that exact recovery is achievable whenever the probability of misclassifying an individual vertex by the so-called \emph{genie} estimator, which labels a single vertex given the labels of the remaining $n-1$ vertices, is $o(1/n)$. However, the GSBM is qualitatively very different from the SBM, and it is not apparent at the outset that it should exhibit local-to-global amplification. In particular, the GSBM is not a low-rank model, suggesting that approaches such as spectral methods \citep{Abbe2021} and semidefinite programming \citep{Hajek2016}, which exploit the low-rank structure of the SBM, may fail in the GSBM. In order to achieve almost exact recovery in the GSBM, we instead use the density of local subgraphs to propagate labels. Our propagation scheme allows us to achieve almost exact recovery, and also ensures that no local region has too many misclassified vertices. The dispersion of errors is crucial to showing that labels can be correctly refined in the second phase.
  
Notably, our algorithm runs in linear time (where the input size is the number of edges). This is in contrast with the SBM, for which no statistically optimal linear-time algorithm for exact recovery has been proposed. To our knowledge, the best-known runtime for the SBM in the logarithmic degree regime is achieved by the spectral algorithm of Abbe et al \citep{Abbe2020}, which runs in $O(n \log^2 n)$ time, while the number of edges is $\Theta(n \log n)$. More recent work of Cohen--Addad et al \citep{Cohen2022} proposed a linear-time algorithm for the SBM, but the algorithm was not shown to achieve the information-theoretic threshold for exact recovery. Intuitively, the strong local interactions in the GSBM enable more efficient algorithms than what seems to be possible in the SBM. 

  Finally, we study the robustness of our algorithm to monotone corruptions, in the case of two communities. As in prior works on semirandom SBMs \citep{Feige2001,Hajek2016,Moitra2016}, we allow an adversary to add intra-community edges and delete inter-community edges. While these changes appear to be helpful, it is known that some algorithms are not robust to such monotone changes in the standard SBM \citep{Feige2001}. We show that our algorithm continues to be statistically optimal, in the presence of a monotone adversary.

\subsection{General Inference Problems on Geometric Graphs}
  In fact, our algorithm can be adapted to capture a wide variety of inference problems beyond the SBM. For example, consider the problem of $\mathbb{Z}_2$ synchronization, itself a simplification of group synchronization. In the standard (non-geometric) version of the problem, each vertex has a label $x^{\star}(u) \in \{\pm 1\}$, and we observe $x^{\star}(u) \cdot x^{\star}(v)+ Z_{uv}$, where $Z_{uv} = Z_{vu} \sim \mathcal{N}(0, \sigma^2)$ is independent noise. Synchronization problems arise in applications to cryo-electron microscopy \citep{Shkolnisky2012}, sensor, clock, and camera calibration \citep{Cucuringu2012,Tron2009,Giridhar2006}, and robotics \citep{Singer2011,Rosen2020}. We introduce a geometric version of $\mathbb{Z}_2$ synchronization, in which observations are only available for vertices $(u,v)$ that are within a prescribed distance, to model synchronization problems with distance limitations, such as in large-scale networks with physical sensors.

  For another example, consider submatrix localization (also referred to as the \emph{spiked Wigner model} or \emph{biclustering}). In the standard version of the problem, we observe a symmetric matrix $A \in \mathbb{R}^{n \times n}$. There is an unknown subset $S^{\star} \subset [n]$, and the matrix entries are sampled $A_{ij} \sim \mathcal{N}(\mu,1)$ if both $i,j \in S^{\star}$, and $A_{ij} \sim \mathcal{N}(0,1)$ otherwise. Submatrix localization has found applications in bioinformatics \citep{ChengChurch2000, shabalin2009}, text mining \citep{Dhillon2001}, and customer segmentation \citep{Hofmann1999}. We propose a geometric variant in which only some observations are available, dictated again by a geometric embedding in $\mathbb{R}^d$, which can be interpreted as a feature embedding. Interpreted this way, this new problem models submatrix localization tasks with limited observations, for which observations are only available between pairs of matrix indices with similar features.

  \xn{All of these geometric inference problems, namely the GSBM, $\mathbb{Z}_2$-synchronization, and submatrix localization, share a common feature: each vertex is associated with both a geometric embedding and an independent discrete community label. Observations for pairs of vertices which are sufficiently close are sampled according to their community labels. To unify and study this family of problems, we introduce the generalized \emph{Geometric Hidden Community Model} (GHCM), which encompasses these three geometric inference problems as special cases. See Definition \ref{def:ghcm} for a formal definition.}
  
  We identify the information-theoretic threshold for the general problem, which is stated in terms of a Chernoff--Hellinger divergence \citep{Abbe2015} of certain labeled mixture distributions, one for each label. We show that all of these inference problems can be handled by a single algorithmic framework, using a local propagation procedure to arrive at an almost exact labeling, followed by a refinement step to arrive at an exact labeling. The information-theoretic threshold is intimately tied with the success of the refinement phase; our proof of the refinement phase amounts to analyzing the genie estimator, which estimates the label of a single vertex given the labels of all remaining vertices. 

  To show the negative side of the information-theoretic threshold (namely, that no algorithm succeeds below the purported threshold), we employ the framework of Abbe et al \citep{Abbe2021}. Essentially, it is enough to show that a given vertex is misclassified by the genie estimator with probability $\omega(1/n)$. In turn, the probability of misclassification can be lower-bounded in a unified way by employing a Cram\'er lower bound. The bound is tight enough to yield a sharp information-theoretic threshold for exact recovery. 

\textcolor{black}{However, there is one important caveat in our results: we assume a certain distinctness of the distributions that govern the pairwise observations between vertices. This assumption is crucial for the success of our algorithm. Unfortunately, there are many interesting special cases of our geometric inference problem which fail to satisfy the distinctness assumption, such as the GSBM with $k \geq 3$ communities, where all intra-community edge probabilities are equal to $p$ and all inter-community edge probabilities are equal to $q$. Notably, this special case of the SBM (without underlying geometry) also precludes spectral algorithms; see \cite{Abbe2020} for a discussion of this difficulty.}

\subsection{Further Related Work}
The GHCM contributes to several lines of research, which we highlight here.
\paragraph*{Inference problems without a geometric structure}
Inference problems without a geometric structure have been extensively studied. In particular, the SBM has attracted significant attention in the probability, statistics, information theory, and machine learning literature \citep{Blondel2008,Decelle2011,Mossel2015,Abbe2015,Hajek2016}; see also \cite{Abbe2017} for a comprehensive survey. While our work focuses on the exact recovery problem and hence the logarithmic degree regime, the SBM has also been studied in other regimes, with the goal of achieving either almost exact or partial recovery \citep{Yun2014,Mossel2015,Abbe2016KS,Mossel2023}. 

There is also an extensive literature on community recovery problems with Gaussian observations. The $\mathbb{Z}_2$-synchronization problem was studied by \citep{Bandeira2017} and \citep{Javanmard2016}, particularly the ability of semidefinite programming in recovering the hidden labels. We note that the synchronization problem has also been studied in spatial graphs; Abbe et al 
\citep{Abbe2018} proposed a version of the more general group synchronization problem, in which vertices reside on a grid. Finally, submatrix localization has been studied from both a theoretical and applied perspective (e.g., \cite{Hartigan1972,shabalin2009,Ma2015,Chen2016,Cai2017}).

\paragraph*{Community recovery in geometric graphs} Our work contributes to the growing literature on community recovery in random geometric graphs, beginning with latent space models proposed in the network science and sociology literature (see, for example, \cite{Handcock2007,Hoff2002}). There have been several models for community detection in geometric graphs. \textcolor{black}{There are two branches in the literature: those where locations and community assignments are independent (e.g., \cite{Avrachenkov2021, Avrachenkov2024, Galhotra2018}) and those where locations and community assignments are dependent (e.g., \cite{Abbe2020b, Li2023, peche2020robustness}). Our model contributes to the former.} The most similar to the one we study is the Geometric Kernel Block Model (GKBM), proposed by Avrachenkov et al \cite{Avrachenkov2024} in response to the preliminary version of this work \citep{Gaudio2024}. The GKBM considers the two-community GSBM only in the one-dimensional setting, but models the geometric dependence of edge probabilities by any function, as opposed to the indicator function. They modified the algorithm of \cite{Gaudio2024} on the exact recovery of the two-community GSBM to show exact recovery is possible above the information-theoretic threshold of GKBM. Another similar model to the one we study is the Soft Geometric Block Model (Soft GBM), proposed by Avrachenkov et al \citep{Avrachenkov2021}. The main difference between the Soft GBM and the GSBM is that the positions of the vertices are unknown. Avrachenkov et al \citep{Avrachenkov2021} proposed a spectral algorithm for almost exact recovery, clustering communities using a higher-order eigenvector of the adjacency matrix. Using a refinement procedure similar to ours, Avrachenkov et al 
 \cite{Avrachenkov2021} also achieved exact recovery, though only in the denser linear average degree regime.  
  
  Another related model is the Geometric Block Model (GBM), proposed by Galhotra et al \cite{Galhotra2018} with follow-up work including \cite{Galhotra2022} and 
  \cite{Chien2020}. In the GBM, community assignments are generated independently, and latent vertex positions are generated uniformly at random on the unit sphere. Edges are then formed according to parameters $\{\beta_{i,j}\}$, where a pair of vertices $u,v$ in communities $i,j$ with locations $Z_u, Z_v$ is connected if $\langle Z_u, Z_v \rangle \leq \beta_{i,j}$.

  Community recovery in geometric graphs is also related to the so-called \emph{Censored} SBM, in which only some edge observations are known \citep{Abbe2014,Dhara2022,Hajek2016b}. Typically, one considers independent measurements, where we discover the edge status of any given pair $(u,v)$ independently, with some probability $p$. Thus, the \emph{measurement graph} is Erd\H{o}s--R\'enyi. In contrast, Chen et al \cite{Chen2016b} proposed measurement graphs with local structure: line graphs, grid graphs, rings, and small-world graphs, thus introducing local structure into the observed censored graph.
  
  In the previously mentioned models, the vertex positions do not depend on the community assignments. In contrast, Abbe et al \cite{Abbe2020b} proposed the Gaussian Mixture Block Model (GMBM), where (latent) vertex positions are determined according to a mixture of Gaussians, one for each community. Edges are formed between all pairs of vertices whose distance falls below a threshold. A similar model was recently studied by Li and Schramm \cite{Li2023} in the high-dimensional setting. Additionally, P\'eche and Perchet \cite{peche2020robustness} studied a geometric perturbation of the SBM, where vertices are generated according to a mixture of Gaussians, and the probability of connecting a pair of vertices is given by the sum of the SBM parameter and a function of the latent positions. 
  
  In addition, some works \citep{araya2019latent,eldan2022community} consider the task of recovering the geometric representation (locations) of the vertices in random geometric graphs as a form of community detection. Their setting differs significantly from ours. We refer to the survey by Duchemiin and De Castro \cite{duchemin2023random} for an overview of the recent developments in non-parametric inference in random geometric graphs. 

  \paragraph*{Geometric random graphs} 
  
  Another line of work studies geometric random graphs without community structure. A prominent example is the Random Geometric Graph (RGG)  \citep{Dall2002,Penrose2003,Erba2020}, in which vertices are distributed uniformly in a region, and edges connect vertices which are within a prescribed distance of each other. A related model is the Soft RGG \citep{Penrose2016,Bubeck2016,Liu2023}, where vertices are connected by an edge with probability $p$, independently, if they are within a prescribed distance of each other. 
  While the input in our model has community structure, our ability to propagate labels depends on the connectivity of an auxiliary graph, which connects a pair of vertices whenever they are ``visible'' to each other (to be defined precisely in Section \ref{sec:connectivity}). Conditioned on the total number of vertices, this (vertex) visibility graph has the same distribution as a RGG. In fact, we show that the visibility graph constructed from sufficiently occupied \emph{regions} is also connected, thus enabling a propagation scheme. 


  \subsection{Notation and Organization} We write $[n]=\{1,\cdots, n\}$. We use Bachmann--Landau notation with respect to the parameter $n$; For example, $o(1)$ means $o_n(1)$. $\text{Bern}$ denotes the Bernoulli distribution. $\text{Bin}$ denotes the binomial distribution. 
  The constant $\nu_d$ is the volume of a unit Euclidean ball in $d$ dimensions. 

Let $D_+(\cdot, \cdot)$ be the \emph{Chernoff-Hellinger (CH) divergence} between two distributions. As introduced by \cite{Abbe2015}, the CH divergence between two discrete distributions with probability mass functions (PMF) $p$ and $q$ on $\mathcal{X}$ is given by 
\begin{align}
    D_+(p,q) = \sup_{t \in [0,1]} \sum_{x \in \mathcal{X}} tp(x) + (1-t)q(x) - p(x)^t q(x)^{1-t} = 1 - \inf_{t \in [0,1]} \sum_{x \in \mathcal{X}} p(x)^t q(x)^{1-t}. \nonumber
\end{align}
Here, the second equality comes from the fact that $\sum_{x \in \mathcal{X}} tp(x) + (1-t)q(x) = 1$ for all $t \in \mathbb{R}$.
We extend this definition to continuous distributions with densities $p$ and $q$, where the CH divergence is defined as 
\$D_+(p,q) &= \sup_{t \in [0,1]} \int_{x \in \cX} tp(x) + (1-t)q(x) - p(x)^t q(x)^{1-t} dx \\
&= 1 - \inf_{t \in [0,1]} \int_{x \in \cX} p(x)^t q(x)^{1-t} dx.\$
Moreover, if $p=(p_1, \hdots, p_k), q=(q_1, \hdots, q_k)$ are vectors of distributions associated with a vector of prior probabilities $\pi=(\pi_1, \hdots, \pi_k)$, we define the CH divergence between $p$ and $q$ in the discrete case to be
\#\label{def:ch-discrete}
    D_+(p,q) &= \sup_{t \in [0,1]} \sum_{i=1}^k \pi_i \left(\sum_{x \in \mathcal{X}} tp_i(x) + (1-t)q_i(x) - p(x)^t q_i(x)^{1-t}\right) \notag\\ &= 1 - \inf_{t \in [0,1]} \sum_{i=1}^k \pi_i \sum_{x \in \mathcal{X}} p_i(x)^t q_i(x)^{1-t},  
\#
and in the continuous case to be
\#\label{def:ch-continuous}
    D_+(p,q) &= \sup_{t \in [0,1]} \sum_{i=1}^k \pi_i \left(\int_{x} tp_i(x) + (1-t)q_i(x) - p_i(x)^t q_i(x)^{1-t} dx\right) \notag\\ &= 1 - \inf_{t \in [0,1]} \sum_{i=1}^k \pi_i \int_{x \in \mathcal{X}} p_i(x)^t q_i(x)^{1-t} dx.  
\#

\paragraph*{Organization} The rest of the paper is organized as follows. Section \ref{sec:models} formulates the Geometric Hidden Community Model and provides our main result and its application to specific network inference problems. Section \ref{sec:algorithm} provides a two-phase algorithm for achieving exact recovery. Section \ref{sec:proof-outline} provides an overview of proof strategies and intuition, with full proofs deferred to the appendix. We discuss future directions in Section \ref{sec:discussion}. Appendix \ref{sec:impossibility} proves the impossibility result in exact recovery. In Appendix \ref{sec:connectivity}, we explore the connectivity of the ``visibility graph,'' a key property in proving success of our algorithm. Appendices \ref{sec:almost_exact_recovery} and \ref{sec:exact_recovery} prove our algorithm achieves almost exact recovery in Phase I and exact recovery in Phase II, respectively. Appendix \ref{sec:monotone} shows our algorithm is robust to monotone adversarial corruptions in the two-community GSBM case. Finally, Appendix \ref{sec:gaussian_proofs} shows that our algorithm achieves almost exact recovery in Phase I and exact recovery in Phase II for the Gaussian case. 


\section{Models and Main Results}
\label{sec:models}
This section presents our main results. First, we propose the Geometric Hidden Community Model (Section \ref{sect:ghcm}) and study its corresponding information-theoretic threshold for exact recovery (Section \ref{sect:limits}). We summarize the lower and upper bounds for exact recovery in Table \ref{tab:impossibility-achievability}. Then, we examine the model's applications to specific spatial graph inference problems, including the geometric stochastic block model (Corollaries \ref{cor:gsbm_2_com} and \ref{cor:gsbm}), the geometric $\mathbb{Z}_2$ synchronization problem (Corollary \ref{cor:gsync}), and the monotone adversarial version of the geometric stochastic block model (Theorem \ref{thm:monotone}).

\subsection{Geometric Hidden Community Model}\label{sect:ghcm}
We formulate a new spatial graph model, the \emph{Geometric Hidden Community Model (GHCM)}, in the logarithmic-degree regime, which generates vertices on a torus and samples observations between vertices only if they are sufficiently close.
\begin{definition}[Geometric Hidden Community Model]\label{def:ghcm}
Let $n\in\mathbb{N}$ be a scaling parameter. Fix $\lambda > 0$ as the intensity parameter, $d \in \mathbb{N}$ as the dimension, and $k\in\N$ as the number of communities. Denote $Z\subset\Z$ as the subset of integers with cardinality $|Z|=k$ to represent the set of discrete community labels; and $\pi\in\R^{k}$ as the prior probabilities on the $k$ communities. For each $i,j \in Z$, let $P_{ij}$\footnote{Or equivalently, sometimes we use $P_{i,j}$.} be a distribution that models either a discrete random variable with
probability mass function (PMF) or a continuous random variable with probability density function (PDF) on $\mathbb{R}$,\footnote{
We avoid subtleties that arise from allowing general measures.} and $P_{ij} \stackrel{d}{=} P_{ji}$. Let $p_{ij}$ be the PMF/PDF of $P_{ij}$. 

A graph $G$ is sampled from $\text{GHCM}(\lambda, n, \pi, P, d)$, with observations $\{Y_{uv}\}\subset\R$ over the undirected edges, according to the following steps:
\begin{enumerate}
    \item The locations of vertices are generated according to a homogeneous Poisson point process\footnote{The definition and construction of a homogeneous Poisson point process are provided in Definition \ref{def:PPP}.} with intensity $\lambda$ in the region $\cS_{d,n} := [-n^{1/d}/2, n^{1/d}/2 ]^d \subset \mathbb{R}^d$. Let $V\subset\cS_{d,n}$ denote the vertex set. \label{sample:step-1}
    \item Community labels are assigned independently by the probability vector $\pi$. The ground truth label of vertex $u \in V$ is given by $x^\star(u) \in Z$, with $\mathbb{P}(x^\star(u) = i) = \pi_i$ for $i\in Z$.
    \item Conditioned on locations and community labels, pairwise observations $Y_{uv}$ are sampled independently. For $u, v \in V$ and $u \neq v$, we have $Y_{uv} \sim P_{x^{\star}(u), x^{\star}(v)}$ if $\|u-v\|\le (\log n)^{1/d}$; otherwise $Y_{uv} = 0$. These observations are symmetric, with $Y_{uv} = Y_{vu}$ for all $u,v$.
\end{enumerate}
Here $\|u-v\|$ denotes the toroidal metric with $\|\cdot\|_2$ denoting the standard Euclidean metric:
\$
\|u-v\| = \big\|\min\{|u_1-v_1|, n^{1/d}-|u_1-v_1|\}, \cdots, \min\{|u_d-\nu_d|, n^{1/d}-|u_d-\nu_d|\} \big\|_2.
\$
\end{definition}

\textcolor{black}{ Observe that in any region of
unit volume, the number of vertices is distributed as $\text{Poisson}(\lambda)$; hence, there exist $\lambda$ vertices in
expectation in any unit-volume region. The region $\mathcal{S}_{d, n}$ has volume $n$ so that the total number of vertices 
is $|V| \sim \text{Poisson}(\lambda n)$ with $E[|V|] = \lambda n$. The visibility radius of $(\log n)^{1/d}$ ensures the expected number of vertices within the visibility radius of each vertex is $\lambda \nu_d\log n$, so that the vertices are of logarithmic degree.}
An observation $Y_{uv}$ is only sampled for a pair of vertices $\{u,v\}$ if their locations are within a distance of $(\log n)^{1/d}$; in that case, we say they are \emph{mutually visible}, denoted by $u\sim v$. The sampled $Y_{uv}$ then depends on the community labels of $u$ and $v$.

The GHCM generalizes several spatial network inference problems, including:
\begin{itemize}
    \item \emph{Geometric Stochastic Block Model (GSBM)}. By specifying $P_{ij} = \text{Bern}(a_{ij})$ for $i,j\in Z$ with $a_{ij} = a_{ji}$, the model reduces to the GSBM \citep{Sankararaman2018,Abbe2021}. The observation $Y_{uv}$ represents an edge or non-edge between vertices $u$ and $v$. In particular, \cite{Sankararaman2018} introduced a spatial random graph model, which is equivalent to the GHCM with $k=2, Z=\{\pm 1\}, \pi_{-1} = \pi_1 = 1/2, P_{ij}=\text{Bern}(p\mathds{1}\{i=j\}+q\mathds{1}\{i\neq j\})$ for $i \in \{\pm 1\}$. We refer to their model as the two-community symmetric GSBM.
    \item \emph{Geometric $\Z_2$ synchronization}. The GHCM reduces to geometric $\Z_2$ synchronization when $k=2$, $Z=\{\pm 1\}$, $P_{i,i} = \cN(1, \sigma^2)$, and $P_{i,-i} = \cN(-1, \sigma^2)$ for $i \in \{\pm 1\}$. Equivalently, the observation $Y_{uv}$ is sampled as $Y_{uv} = x^\star(u) \cdot x^\star(v) + Z_{uv}$,  where $Z_{uv} \sim \mathcal{N}(0, \sigma^2)$ is independent Gaussian noise. The geometric structure restricts observations to only vertex pairs within a limited distance, which models synchronization problems with distance constraints, such as large-scale networks of sensors.
    \item  \emph{Geometric submatrix localization}. The GHCM corresponds to geometric submatrix localization when $k=2$, $Z=\{\pm 1\}$, $P_{1,1} = \cN(\mu,1)$, and $P_{-1,-1} = P_{-1,1} = P_{1,-1} = \cN(0,1)$. 
    Equivalently, $Y_{uv}$ is sampled as $Y_{uv} = \mu\mathds{1}\{x^\star(u)=1\} \cdot \mathds{1}\{x^\star(v)=1\} + Z_{uv}$  with independent Gaussian noise $Z_{uv} \sim \cN(0, 1)$. The geometric structure can be interpreted as a feature embedding in $\mathbb{R}^d$, yielding submatrix localization problems where observations are only available between data points nearby sharing similar features.
\end{itemize}

\textcolor{black}{We now illustrate that the GHCM exhibits much stronger edge transitive behavior compared to the standard SBM, thereby more faithfully capturing the transitivity observed in real-world networks. For the GHCM, let $\lambda$ be the Poisson intensity parameter, and let $a, b$ denote the intra- and inter-community edge probabilities for vertices which are mutually visible, namely located at a distance of at most $(\log n)^{1/d}$. For the SBM, suppose that the intra- and inter-community edge probabilities are respectively ${a \lambda \nu_d \log n}/{n}$ and ${b\lambda \nu_d \log n}/{n}$, where $\nu_d$ is the volume of a $d$-dimensional unit ball. With these parametrizations, both models are such that a given vertex has $a \lambda \nu_d \log n/2$ neighbors in the same community, and $b \lambda \nu_d \log n/2$ neighbors in the opposite community, in expectation. However, the two models have very different transitive behavior. Specifically, we can calculate the conditional probability that two vertices $(u,v)$ are connected by an edge, given that both are connected to another vertex $w$. In the SBM, this conditional probability is only $\Theta(\log n/n)$ (and has the same asymptotics if we additionally condition on the community assignments of $u$, $v$, and/or $w$). On the other hand, in the GHCM, the conditional probability that $u,v$ are connected 
given that both are connected to 
$w$ is $\Theta(1)$. To see this, observe that for $u,v$ to be connected to $w$, these vertices must be within distance $(\log n)^{1/d}$ of $w$. Given that this placement occurs, there is a constant probability that $u,v$ themselves are within distance $(\log n)^{1/d}$ of each other, and a further constant conditional probability that they are connected by an edge. Given the dramatically different ``transitive closure'' probabilities, we conclude that the GHCM models the desired transitive behavior.}

\subsection{Fundamental Limits for Exact Recovery}\label{sect:limits}
Our goal for community recovery is to estimate the ground truth labeling $x^\star$, up to some level of accuracy, by observing $\{Y_{uv}\}$ and $G$, including the geometric locations of the vertices. However, when symmetries are present, a labeling is only correct up to a permutation. For example, in the two-community symmetric GSBM, one can only estimate the correct labeling up to a global flip. To account for such symmetries, we first define the notion of permissible relabeling.
\begin{definition}[Permissible relabeling] \label{def:permissible} A permutation $\omega\colon Z\to Z$ is a \emph{permissible relabeling} if $\pi_i = \pi_{\omega(i)}$ for any $i\in Z$ and $P_{ij} = P_{\omega(i), \omega(j)}$ for any $i,j\in Z$. Let $\Omega_{\pi, P}$ be the set of permissible relabelings.
\end{definition}
Given an estimator $\widetilde{x}$, we define the agreement of $\widetilde{x}$ and $x^\star$ as
\$
A(\widetilde{x}, x^\star) = \frac{1}{|V|} \max_{\omega\in\Omega_{\pi, P}} \sum_{u\in V}\mathds{1}\{\widetilde x(u) = \textcolor{black}{\omega( x^\star(u))}\}.
\$ 
\xn{Throughout the paper, we interchangeably use the notation $\omega \circ x^\star(u) = \omega( x^\star(u))$ to denote the composition of the true label $x^\star(u)$ under a permissible relabeling $\omega$. }
Next, we define different levels of accuracy including \emph{exact recovery}, our ultimate goal, as follows.
  \begin{itemize}
      \item \emph{Exact recovery:} $\lim\limits_{n \to \infty} \pr(A(\widetilde{x}, x^\star)=1) = 1$,
      \item \emph{Almost exact recovery:} $\lim\limits_{n \to \infty} \pr(A(\widetilde{x}, x^\star)\ge1-\epsilon) = 1$, for all $\epsilon>0$,
      \item \emph{Partial recovery:} $\lim\limits_{n \to \infty} \pr(A(\widetilde{x}, x^\star)\ge\alpha) = 1$, for some $\alpha \in (0,1)$. 
  \end{itemize}
\textcolor{black}{Recall that $|V| \sim \text{Poisson}(\lambda n)$ and $n$ is a scaling parameter. Taking the limit as $n \to \infty$ is equivalent to letting the graph size approach infinity.}
An exact recovery estimator must recover all labels up to a permissible relabeling with probability tending to $1$ as the graph size goes to infinity. An almost exact recovery estimator recovers all labels but a vanishing fraction of vertices, and a partial recovery estimator recovers the labelings of a constant fraction of vertices as the graph size goes to infinity.

 In the following, we present the impossibility and the achievability results for exact recovery in the GHCM. There exists an information-theoretic threshold below which no estimator succeeds in exact recovery. Above this threshold, exact recovery is efficiently possible with some assumptions. First, we present the negative result. Recall the definition of CH divergence $D_+(\cdot, \cdot)$ between two vectors of distributions defined in \eqref{def:ch-discrete} and \eqref{def:ch-continuous}. 
 \xn{Let $\theta_i = (p_{i1}, \cdots, p_{ik})$ be the vector of distributions $P_{i\cdot}$ associated with the prior $\pi$.}

\begin{theorem}[Impossibility]\label{theorem:impossibility}  
Any estimator fails to achieve exact recovery for $G\sim \text{GHCM}(\lambda, n, \pi, P, d)$ if
\begin{enumerate}
    \item[(a)] $\lambda \nu_d \min_{i \neq j}  D_+(\theta_i, \theta_j) < 1$; or
    \item[(b)] $d=1$, $\lambda<1$, and $|\Omega_{\pi, P}|\ge 2$.
\end{enumerate}
\end{theorem}

Theorem \ref{theorem:impossibility}(a) identifies the information-theoretic threshold as
\begin{align}
    \label{eqn:threshold}
    \lambda \nu_d \min_{i \neq j}  D_+(\theta_i, \theta_j) = 1.
\end{align}
The parameters of the GHCM lie below the threshold \eqref{eqn:threshold} if either $\lambda$ or $\min_{i \neq j}  D_+(\theta_i, \theta_j)$ is sufficiently small. A smaller value of the intensity parameter $\lambda$ results in fewer vertices in the torus, while a smaller value of $\min_{i \neq j}  D_+(\theta_i, \theta_j)$ indicates that there are two communities $i$ and $j$ such that their distributions $P_{i\cdot}$ and $P_{j\cdot}$ are too similar. Consequently, Theorem \ref{theorem:impossibility}(a) implies that exact recovery is impossible if we do not observe enough data or if there exist two communities with indistinguishable pairwise observations.
Moreover, Theorem \ref{theorem:impossibility}(a) implies exact recovery becomes more difficult in higher dimensions because $\nu_d$, the volume of the Euclidean ball in $d$ dimensions, tends to zero as $d$ tends to infinity. We show the exact recovery problem can be reduced to a hypothesis testing problem between pairs of vectors of distributions, with each community associated with one vector. The CH divergence quantifies the error rate in this hypothesis test. 
Theorem \ref{theorem:impossibility}(b) identifies additional impossible regimes in terms of $\lambda$ in the one-dimensional setting with multiple permissible relabelings. 

We conjecture that when the model parameters are above the threshold \eqref{eqn:threshold}, exact recovery is possible. The following conjecture, together with Theorem \ref{theorem:impossibility}, would imply \eqref{eqn:threshold} is the information-theoretic threshold of the GHCM.
\begin{conjecture}\label{conjecture}
For $G\sim \text{GHCM}(\lambda, n, \pi, P, d)$, the MAP estimator achieves exact recovery 
when $\lambda \nu_d \min_{i \neq j}  D_+(\theta_i, \theta_j) > 1$ and one of the following holds: (1) $d\ge2$; or (2) $d = 1$, $\lambda>1$; or (3) $d=1$, $|\Omega_{\pi, P}|=1$.
\end{conjecture}

Our linear-time algorithm (Algorithm \ref{alg:almost-exact}, Section \ref{sec:algorithm}) succeeds in exact recovery whenever the model parameters are above the information-theoretic threshold \eqref{eqn:threshold}. However, the algorithm requires the following assumptions:
\begin{assumption}[Bounded log-likelihood difference]\label{ass:bounded-ratio}
There exists $\eta>0$ such that $\log(p_{ij}(y)/p_{ab}(y))<\eta$ for any $i,j,a,b\in Z$ and any $y\in \R$.    
\end{assumption}
  
\begin{assumption}[Distinctness]\label{ass:distinguishable}
  For any $i\in Z$ and any $r\neq s \in Z$, $P_{ir} \neq P_{is}$.
\end{assumption}

Assumption \ref{ass:distinguishable} enables the differentiation of any two communities, $r$ and $s$, based on vertices from any community $i$. This assumption is satisfied in the two-community GSBM 
and the geometric $\Z_2$ synchronization. However, there are notable exceptions where this assumption does not hold, such as the GSBM with $k \geq 3$ communities where all intra- and inter-community probabilities are $a$ and $b$, respectively; or geometric submatrix localization. To achieve exact recovery in the one-dimensional GHCM with one permissible relabeling, we require the following stronger distinctness assumption that all distributions in $P$ are distinct.  
\begin{assumption}[Strong distinctness]\label{ass:distinguishable-strong} 
For any $i, j, r, s \in Z$, $P_{ij} \neq P_{rs}$.
\end{assumption}

The following theorem states the achievability results.
\begin{theorem}[Achievability]\label{theorem:exact-recovery}
Under Assumptions \ref{ass:bounded-ratio} and \ref{ass:distinguishable}, exact recovery is achievable with efficient runtime if $\lambda \nu_d \min_{i \neq j}  D_+(\theta_i, \theta_j) > 1$ and one of the following holds: (1) $d\ge2$; or (2) $d = 1$ and $\lambda>1$. 
Under Assumptions \ref{ass:bounded-ratio} and \ref{ass:distinguishable-strong}, exact recovery is achievable with efficient runtime if $\lambda \nu_d \min_{i \neq j}  D_+(\theta_i, \theta_j) > 1$, $d=1$, and $|\Omega_{\pi, P}|=1$.
\end{theorem}
Under Assumptions \ref{ass:bounded-ratio} and \ref{ass:distinguishable}, Algorithm \ref{alg:almost-exact} succeeds in exact recovery via a two-phase procedure. In the first phase, it partitions the torus $\mathcal{S}_{d,n}$ into sufficiently small blocks. We show that, with high probability, there exists a search order of these blocks in which the vertices occupying a particular block are all visible to the vertices occupying the preceding block. Using this search order, Phase I constructs a preliminary labeling, which achieves almost exact recovery. In Phase II, the algorithm achieves exact recovery by refining the preliminary labeling of Phase I. The refinement step mimics the genie-aided estimator, which labels a vertex given the correct labels of its visible vertices. The runtime of Algorithm \ref{alg:almost-exact} is $O(n \log n)$, which is linear time with respect to the number of edges of the graph. Therefore, the GHCM does not exhibit a computational-statistical gap under these two assumptions. 

In the one-dimensional case with one permissible relabeling, under the strong distinctness assumption, our algorithm achieves exact recovery at smaller values of $\lambda$, even when multiple segments of occupied blocks are disconnected in the torus. The modified algorithm 
applies the two-phase approach to each segment of occupied blocks, ensuring exact recovery. 

\begin{table}[htbp]
\centering
\begin{tabular}{c|p{0.38\textwidth}|p{0.48\textwidth}}
\hline
 & \textbf{Impossibility} & \textbf{Achievability} \\
\hline
$d=1$
&
\begin{itemize}
    \item $\lambda \nu_d \min_{i \neq j} D_+(\theta_i,\theta_j) < 1$, or
    \item $\lambda < 1,\ |\Omega_{\pi,P}| \geq 2$
\end{itemize}
&
\begin{itemize}
    \item $\lambda \nu_d \min_{i \neq j} D_+(\theta_i,\theta_j) > 1$ and $\lambda > 1^\star$, or 
    \item $\lambda \nu_d \min_{i \neq j} D_+(\theta_i,\theta_j) > 1$ and $|\Omega_{\pi,P}| = 1^\dagger$
\end{itemize}
\\
\hline
$d\geq 2$
&
$\lambda \nu_d \min_{i \neq j} D_+(\theta_i,\theta_j) < 1$
&
$\lambda \nu_d \min_{i \neq j} D_+(\theta_i,\theta_j) > 1^\star$
\\
\hline
\end{tabular}
\caption{Summary of lower and upper bounds for exact recovery in the GHCM. In the achievability results, $^\star$ indicates that Assumptions \ref{ass:bounded-ratio} and \ref{ass:distinguishable} are required, and $^\dagger$ indicates that Assumptions \ref{ass:bounded-ratio} and \ref{ass:distinguishable-strong} are required.}
\label{tab:impossibility-achievability}
\end{table}

Phase I of Algorithm \ref{alg:almost-exact} achieves almost exact recovery for a wider range of parameterizations than the conditions required for exact recovery in Theorem \ref{theorem:exact-recovery}. 
  \begin{theorem}\label{theorem:almost-exact-recovery} Under Assumptions \ref{ass:bounded-ratio} and \ref{ass:distinguishable},
  almost exact recovery is achievable with efficient runtime in $G\sim \text{GHCM}(\lambda, n, \pi, P, d)$ whenever (1) $d \geq 2$ and $\lambda \nu_d > 1$;\footnote{Since the CH divergence is bounded as $0\le D_+(\cdot, \cdot)\le 1$, this condition here is weaker than that for the exact recovery stated in Theorem \ref{theorem:exact-recovery}.} or (2) $d = 1$ and $\lambda>1$\footnote{When $d=1$, $\nu_1 = 2$ and thus $\lambda \nu_1>1$ gives $\lambda>1/2$. The condition $\lambda>1$ here is hence stronger than the previous one $\lambda \nu_1>1$ for $d=1$.}. Under Assumptions \ref{ass:bounded-ratio} and \ref{ass:distinguishable-strong}, almost exact recovery is achievable with efficient runtime if $d=1$ and $|\Omega_{\pi, P}|=1$. 
\end{theorem}


\subsubsection{Special Cases of the GHCM}

Next, we explore the applications of Theorem \ref{theorem:exact-recovery} and Algorithm \ref{alg:almost-exact} to specific spatial graph inference problems that satisfy Assumptions \ref{ass:bounded-ratio} and \ref{ass:distinguishable}. In each application, the information-theoretic threshold is characterized using the CH divergence computed by the specific parameters.

\paragraph*{Geometric Stochastic Block Model}
By specifying $k=2$, $Z=\{\pm 1\}$, $\pi_{-1} = \pi_1 = 1/2$, and $P_{i,i} = \text{Bern}(a)$ and $P_{i,-i} = \text{Bern}(b)$ for $i\in \{\pm 1\}$, the GHCM in Definition \ref{def:ghcm} reduces to the two-community symmetric GSBM proposed by \cite{Sankararaman2018}. A follow-up work by \cite{Abbe2021} identified a parameter regime in which exact recovery is impossible. 
  \begin{theorem}[Theorem 3.7 in \cite{Abbe2021}]
  Let $\lambda > 0$, $d \in \mathbb{N}$, and $0 \leq b < a \leq 1$ satisfy 
  \$
  \lambda \nu_d \big(1-\sqrt{ab} - \sqrt{(1-a)(1-b)} \big) < 1, \$
  and let $G \sim \text{GSBM}(\lambda, n, a, b, d)$. Then any estimator $\widetilde{\sigma}$ fails to achieve exact recovery.
  \end{theorem}
  Abbe, Baccelli, and Sankararaman in \cite{Abbe2021} conjectured that the above result is tight, but only established that exact recovery is achievable for sufficiently large $\lambda > \lambda(a,b,d)$ \cite[Theorem 3.9]{Abbe2021}. In this regime, they provided a polynomial-time algorithm based on the observation that the relative community labels of two nearby vertices can be determined with high accuracy by counting their common neighbors. By taking $\lambda > 0$ large enough to amplify the density of vertices, the failure probability of pairwise classification 
  can be taken arbitrarily small. Our results in Theorem \ref{theorem:exact-recovery} using the $\text{GHCM}$ confirms the conjecture of \cite{Abbe2021}, yielding the second part of the following corollary. 
  
\begin{corollary}[Two-community symmetric GSBM]\label{cor:gsbm_2_com}
    For $\textnormal{GHCM}(\lambda, n, \pi, P, d)$ with $Z=\{\pm 1\}$, $\pi_{-1} = \pi_{1}=1/2$, $P_{1, 1}=P_{-1, -1}=\text{Bern}(a)$, and $P_{-1, 1}=P_{1, -1}=\text{Bern}(b)$,
    \begin{enumerate}
        \item any estimator fails at exact recovery whenever
        \[  \lambda \nu_d \big(1 - \sqrt{ab} - \sqrt{(1-a)(1-b)}\big) < 1, \]
        or whenever $d=1$ and $\lambda < 1.$
        \item there exists a polynomial-time algorithm achieving exact recovery whenever 
        \[  \lambda \nu_d \big(1 - \sqrt{ab} - \sqrt{(1-a)(1-b)}\big) > 1, \]
        and either (1) $d\ge2$; or (2) $d = 1$ and $\lambda>1$.
    \end{enumerate}
\end{corollary}

Additionally, Theorem \ref{theorem:exact-recovery} extends the result to the general GSBM with $k$ communities. 
\begin{corollary}[General GSBM]\label{cor:gsbm}
    Consider $\textnormal{GHCM}(\lambda, n, \pi, P, d)$ with $Z = [k]$, $P_{ij} = \text{Bern}(a_{ij})$ with $a_{ij} = a_{ji}$ for $i,j \in Z$. Then,
    \begin{enumerate}
        \item any estimator fails at exact recovery  whenever
        \[  \lambda \nu_d \Big( 1 - \max_{i\neq j}\inf_{t\in [0, 1]}\sum_{r=1}^k \pi_r \big(a_{ir}^t a_{jr}^{1-t} + (1-a_{ir})^t (1-a_{jr})^{1-t}\big) \Big) < 1, \]
         or whenever $d=1$, $\lambda<1$, and $|\Omega_{\pi, P}|\ge 2$.
        \item there exists a polynomial-time algorithm achieving exact recovery whenever Assumption \ref{ass:distinguishable} holds,
        \[  \lambda \nu_d \Big( 1 - \max_{i\neq j}\inf_{t\in [0, 1]}\sum_{r=1}^k \pi_r \big(a_{ir}^t a_{jr}^{1-t} + (1-a_{ir})^t (1-a_{jr})^{1-t}\big) \Big) > 1, \]
        and either (1) $d\ge2$; or (2) $d = 1$, $\lambda>1$; or (3) $d=1$, $|\Omega_{\pi, P}|=1$, and Assumption \ref{ass:distinguishable-strong} holds. 
    \end{enumerate}

\end{corollary}

\paragraph*{Geometric $\mathbb{Z}_2$ synchronization}
Taking $P_{ij}$ to be a Gaussian distribution for all $i,j$ yields the geometric versions of Gaussian network inference problems such as $\mathbb{Z}_2$-synchronization. Note that such Gaussian problems violate the bounded log-likelihood ratio condition of Assumption \ref{ass:bounded-ratio}. However, as we will show in Appendix \ref{sec:gaussian_proofs}, our algorithm still achieves exact recovery in the Gaussian case as long as the distinctness assumption is satisfied. As a result, we have the following results on the geometric $\Z_2$ synchronization.

\begin{proposition}[Geometric $\mathbb{Z}_2$ synchronization]\label{cor:gsync}
Consider $\textnormal{GHCM}(\lambda, n, \pi, P, d)$ with $Z=\{\pm 1\}$ and  $P_{ij} = \text{N}(ij, \sigma^2)$. Observe $|\Omega_{\pi, P}|= 2$ if and only if the prior is symmetric, i.e., $\pi_{-1}=\pi_{1}$. Otherwise, $|\Omega_{\pi, P}|= 1$. Therefore,
\begin{enumerate}
    \item any estimator fails at exact recovery with high probability whenever \[ \lambda \nu_d \Big( 1 - \exp\Big(-\frac{1}{2\sigma^2}\Big)\Big) < 1,\]
    or whenever $d=1$, $\lambda < 1$, and $\pi_{-1}=\pi_{1}$.
    \item there exists a polynomial-time algorithm achieving exact recovery whenever
    \[ \lambda \nu_d \Big( 1 - \exp\Big(-\frac{1}{2\sigma^2}\Big)\Big) > 1, \]
    and either (1) $d\geq 2$; or (2) $d=1$ and $\lambda > 1$. 
\end{enumerate}
    
\end{proposition}


\subsubsection{Robustness Under Monotone Adversaries}
Finally, we study the robustness of our algorithm to monotone corruptions, in the case of two-community GSBM. In order to consider the ``monotone'' adversarial errors, we impose the following assumption that the intra-community probabilities are larger than the inter-community probabilities.
\begin{assumption}\label{ass:monotone}
    In the GSBM, we assume $a_{ii} > a_{jr}$ for any $i,j,r\in Z$ and $j\neq r$.
\end{assumption}

Then, we define the semi-random graph model with monotone adversaries as follows.
\begin{definition}[Semi-random GSBM]
    In the semi-random GSBM$(\lambda, n, \pi, P, d)$, a random graph $G = (V, E)$ is first sampled from GSBM$(\lambda, n, \pi, P, d)$. Next, an adversary adds any edges within communities and removes any edges between communities, and returns the graph $G^\prime= (V, E^\prime)$.
\end{definition} 

While these changes seem helpful, some algorithms, like spectral methods, are not robust to such monotone corruptions in the standard SBM \citep{Feige2001}. We show that our algorithm remains statistically optimal in the two-community GSBM, in the presence of a monotone adversary.

\begin{theorem}\label{thm:monotone}
Suppose that Assumption \ref{ass:monotone} holds. Let $G\sim \mathrm{GSBM}(\lambda, n, \pi, P, d)$ with $k=2$ communities. Let $G'$ be any monotone perturbation of $G$ that adds intra-community edges and removes inter-community edges. Then, there exists an efficient algorithm achieving exact recovery under the conditions in Theorem \ref{theorem:exact-recovery}.
\end{theorem}

\section{Exact Recovery Algorithm}\label{sec:algorithm}
This section presents our two-phase algorithm to achieve exact recovery in Theorem \ref{theorem:exact-recovery}. Phase I constructs an almost-exact labeling $\widehat{x} \colon V \to Z\cup\{*\}$, where the label $*$ indicates uncertainty. 
Phase I is based on the following observation: for any $\delta > 0$, if we know the correct labels of some $\delta \log n$ vertices visible to a given vertex $v$, then we can correctly determine the label of $v$ with probability $1-n^{-c}$, for some $c(\delta)> 0$, by computing statistics using their pairwise observations. The algorithm partitions the torus into hypercubes of volume $\Theta (\log n)$ 
and produces an almost exact labeling of all blocks that contain at least $\delta \log n$ vertices 
by an iterative label propagation scheme. Next, Phase II refines the labeling $\widehat{x}$ to $\widetilde x$ by mimicking the genie-aided estimator. Phase II builds upon a well-established approach in the SBM literature \citep{Abbe2015, Mossel2015} to refine an almost-exact labeling with dispersed errors into an exact labeling. Therefore, the main novelty of our algorithm lies in Phase I, which exploits the dense local structure of the geometric model.

We first describe the algorithm specialized to the case $d=1$ in Section \ref{sec:d_1}, in which the torus, $\mathcal{S}_{1,n}$, is an interval of length $n$. Then in Section \ref{sec:general_d}, we study the general $d\geq 2$ case, where several additional ideas are required to ensure uninterrupted propagation of label estimates over all occupied blocks.
 
\subsection{Exact Recovery for \texorpdfstring{$d=1$}{}}\label{sec:d_1}
 We first describe the simplest case when $d=1$ and $\lambda >2$.
  
  \vskip6pt
  \paragraph*{Algorithm for $\lambda > 2$} 
  The algorithm is presented in Algorithm \ref{alg:almost-exact-eg}. 
  In Phase I, we first partition the interval into blocks of length ${\log n}/{2}$ (Line \ref{line:step1-eg}) and define $V_i$ as the set of vertices in the $i$th block for $i\in [2n/\log n]$. In this way, any pair of vertices in adjacent blocks are within a distance of $\log n$ and visible to each other.
  The density $\lambda>2$ ensures a high probability that all blocks have $\Omega(\log n)$  
  vertices, as we later show in \eqref{eq:connected-H-lambda2}. 
  
  Next, we label 
  $\varepsilon_0 \log n$ vertices from the initial block using the \texttt{Maximum a Posteriori (MAP)} subroutine (Line \ref{line:step2-eg}). If $|V_1| > \varepsilon_0 \log n$, we sample a subset $V_0\subset V_1$ with $|V_0| = \varepsilon_0 \log n$ to ensure the number of vertices labeled by the MAP estimator is small enough so that the runtime is $o(n \log n)$. In particular, we set $\varepsilon_0 \le 1/(2\log k)$ in Line \ref{line:epsilon_0-eg}. Note that there are $k^{\varepsilon_0\log n}$ possible labelings. Evaluating the posterior probability of each labeling requires all pairwise observations among vertices in $V_0$ and thus has runtime $O(\binom{\epsilon_0 \log n}{2}) = O(\log^2n)$. Thus, for sufficiently large $n$, the runtime of the MAP subroutine is 
\#\label{eq:map-runtime}
O\big(\log^2n \cdot k^{\epsilon_0 \log n} \big) = O\big(\log^2n \cdot n^{\epsilon_0 \log k}\big) \le O\big(\log^2n \cdot \sqrt{n} \big) = o(n\log n), 
\#
which is sublinear with respect to the number of edges in $G$.
  
Then the labeling of $V_0$ is propagated to $V_1^\prime$ and other blocks $V_i$ for $i\ge 2$ using the \texttt{Propagate} subroutine in Algorithm \ref{alg:propagation} (Lines \ref{line:step3-eg}-\ref{line:step3-end-eg}). The reference set $S$ in Algorithm \ref{alg:propagation} plays the role of $V_{i-1}$, and $S'$ plays the role of $V_i$. The \texttt{Propagate} subroutine labels the vertices in $V_i$ using the observations between $V_i$ and $V_{i-1}$ and the estimated labeling on $V_{i-1}$, by maximizing the log-likelihood, where $p_{jr}(y_{uv})$ is the PMF/PDF evaluated at $y_{uv}$. \textcolor{black}{This subroutine assigns labels to new nodes based on their observations to the vertices of most frequent label in the parent block.} In Theorem \ref{thm:phase1-summary}, we will demonstrate that Phase I achieves almost-exact recovery on $G$ under appropriate conditions. \textcolor{black}{The choice to propagate labels by referring to vertices of the most frequent label is for convenience. It is also possible to achieve almost exact recovery by labeling nodes relative to all nodes in the seed region, which we analyze in Appendix \ref{sect:propagation-alt}.}


\begin{algorithm}
\caption{Exact recovery for the GHCM ($d=1$ and $\lambda>2$)} \label{alg:almost-exact-eg}
\begin{algorithmic}[1]
    \Require{$G \sim \text{GHCM}(\lambda, n, \pi, P, 1)$ where $\lambda > 2$.}
    \Ensure{An estimated community labeling $\widetilde{x}: V \to Z$.}
    \vspace{5pt}
    \State{{\bf Phase I:}} 
    \State{Partition the interval $[-n/2, n/2]$ into $2n/\log n$ blocks\footnotemark \ of volume $\log n/2$ each. Let $B_i$ be the $i$th block and $V_i$ be the set of vertices in $B_i$ for $i\in [2n/\log n]$.} \label{line:step1-eg}
    \State{Set $\varepsilon_0 \leq \min\{1/(2\log k), |V_1|/\log n\}$. Sample $V_{0}\subset V_{1}$ such that $|V_{0}| = \varepsilon_0 \log n$. Set $V_{1}^\prime \leftarrow V_{1} \setminus V_{0}$.} \label{line:epsilon_0-eg} 
  \State{Apply \texttt{Maximum a Posteriori} (Algorithm \ref{alg:initial-block}) on input $G, V_{0}$ to obtain a labeling $\widehat{x}$ of $V_{0}$.}\label{line:step2-eg}  
  \State{If $V_1^\prime \neq \emptyset$, apply \texttt{Propagate} (Algorithm \ref{alg:propagation}) on input $G, V_{0}, V_{1}^\prime$ to determine the labeling $\widehat{x}$ on $V_{1}^\prime$.} \label{line:step3-eg}
      \For{$i=2,\cdots, 2n/\log n$}
      \State Apply \texttt{Propagate} (Algorithm \ref{alg:propagation}) on input $G, V_{i-1}, V_i$ to determine the labeling $\widehat{x}$ on $V_i$. 
      \EndFor \label{line:step3-end-eg}
      \vspace{5pt}
      \State{{\bf Phase II:}}
      \For{$u\in V$}  
      \State Apply \texttt{Refine} (Algorithm \ref{alg:refine}) on input $G, \widehat{x}, u$ to obtain $\widetilde{x}(u)$.
      \EndFor
      \end{algorithmic}
  \end{algorithm}
  \footnotetext[1]{The number of blocks is $\ceil{2n/\log n}$ if $2n/\log n$ is not an integer.}
  
  \begin{algorithm}
      \caption{\texttt{Maximum a Posteriori (MAP) Estimate}} 
      \label{alg:initial-block}
      \begin{algorithmic}[1]
      \Require{ Graph $G = (V,E)$ and vertex set $S \subset V$.}
      \Ensure{An estimated labeling $\widehat x_S \colon S\to Z$.} 
      \State Set
      \$
      \widehat x_S = \argmax_{x\colon S\to Z} \mathbb{P}(x^\star=x | G).
      \$
      \end{algorithmic}    
  \end{algorithm}
  
 \begin{algorithm}
    \caption{\texttt{Propagate}} \label{alg:propagation}
    \begin{algorithmic}[1]
    \Require{ Graph $G = (V,E)$, mutually visible vertex sets $S, S' \subset V$, $S \cap S' = \emptyset$, and $\widehat{x}_S\colon S\to Z$.}
    \Ensure{An estimated labeling $\widehat x_{S'} \colon S'\to Z$.} 
    \State{Take $j=\argmax_{r\in Z} |\{v\in S\colon \widehat x_S(v)=r\}|$, so that $j$ represents the largest community in $S$.}\label{line:largest-community}
    \For{$u \in S'$}
        \$
        \widehat x_{S'}(u) = \argmax_{r\in Z}  \sum_{v\in S,\, 
         \widehat x_S(v)=j} \log p_{jr}(y_{uv}).
        \$
    \EndFor
    \end{algorithmic}
\end{algorithm}

\begin{remark}
  The distinctness of distributions in Assumption \ref{ass:distinguishable} is crucial in allowing Algorithm \ref{alg:propagation} to succeed. To motivate the need for the distinctness assumption, we consider the case where two communities $i$ and $j$ are such that $P_{ir} = P_{j r}$ for all $r \in Z \setminus \{m\}$, while $P_{im} \neq P_{j m}$. Suppose that we are determining the labels of a set $S'$ relative to a reference set $S$. If $S$ does not contain vertices from community $m$, then communities $i$ and $j$ cannot be distinguished in the set $S'$. 
  Assumption \ref{ass:distinguishable} precludes such situation; any community $r \in Z$ distinguishes communities $i$ and $j$. Therefore, we simply choose the largest community in $S$ according to the initial labeling, as shown in Line \ref{line:largest-community} of Algorithm \ref{alg:propagation}.
  \end{remark}
  
  \begin{algorithm}
      \caption{\texttt{Refine}} \label{alg:refine}
      \begin{algorithmic}[1]
      \Require{ Graph $G \sim \text{GHCM}(\lambda, n, \pi, P, d)$, vertex $u \in V$, labeling $\widehat{x} \colon V \to Z\cup\{*\}$.}
      \Ensure{ An estimated labeling $\widetilde{x}(u) \in Z$.}  
      \State{Set}
      \$
      \widetilde{x}(u) = \argmax_{i\in Z}\sum_{v\in V\setminus\{u\}, v\sim u} \log p_{i,\widehat x(v)} (y_{uv}).
      \$
      \end{algorithmic}
  \end{algorithm}
  
  Phase II refines the almost-exact labeling $\widehat{x}$ obtained from Phase I. Our refinement procedure mimics the so-called \emph{genie-aided} estimator \citep{Abbe2017}, which labels a vertex $u$ given the correct labels of all other vertices (i.e., $\{x^\star(v)\colon v\in V \setminus\{u\}\}$).
  The genie-aided MLE estimator gives that
\$
x_{\textsf{genie}}(u)
 & = \argmax_{i\in Z} \log \pr(G=g\given x(u)=i, x_{-u} = x^\star_{-u}) \notag\\
 & = \argmax_{i\in Z} \sum_{v\in V\setminus\{u\}, v\sim u}\log p_{i, x^\star(v)} (y_{u v}).
 \$
Equivalently, by defining the likelihood function of class $i$ with reference labeling $x$ at a vertex $u\in V$ as
\#\label{eq:ell_global}
\ell_i(u, x) = \sum_{v\in V\setminus\{u\}, v\sim u}\log p_{i, x(v)} (y_{u v}),
\#
  the genie-aided estimator is $x_{\textsf{genie}}(u) = \argmax_{i\in Z}\{\ell_i(u, x^\star)\}$ for any $u\in V$. In contrast, the \texttt{Refine} subroutine (Algorithm \ref{alg:refine}) in Phase II uses the preliminary labeling from Phase I instead of the ground truth labeling. It assigns $\widetilde x(u) = \argmax_{i\in Z}(\ell_i(u,\widehat x))$ for any $u\in V$. 
  Since $\widehat{x}$ makes few errors compared with $x^\star$, we expect the log-likelihood $\ell_i(u,\widehat x)$ to be close to $\ell_i(u, x^\star)$ for any $u$. 

  \vskip6pt
  \paragraph*{Modified algorithm for general $\lambda > 1$} 
  If $1<\lambda<2$, partitioning $\cS_{1,n}$ into blocks of length $\log n/2$, as done in Line \ref{line:step1-eg} of Algorithm \ref{alg:almost-exact-eg}, prevents the algorithm from labeling all vertices. With smaller values of $\lambda$, there is a higher probability that a block is empty; by definition of the Poisson point process, each of the $2n/\log n$ blocks is independently empty with probability $e^{-\lambda\log n/2} = n^{-\lambda/2}$.
  Existence of an empty block prematurely terminates the propagation procedure, resulting in an incomplete labeling. To overcome this bottleneck, we instead adopt smaller blocks of length $\chi\log n$, where $\chi < (1-1/\lambda)/2$, for any $\lambda>1$. Then, the algorithm only labels blocks with sufficiently many vertices, according to the following definition. For the rest of the paper, let $V(B)\subset V$ denote the set of vertices in a subregion $B\subset\cS_{d,n}$. 
  \begin{definition}[Occupied block]
      Given any $\delta>0$, a block $B \subset \cS_{d,n}$ is $\delta$-occupied if $|V(B)| > \delta\log n$. Otherwise, $B$ is $\delta$-unoccupied.
  \end{definition}
  
  We show that for sufficiently small $\delta > 0$, all but a negligible fraction of blocks are $\delta$-occupied. As a result, achieving almost-exact recovery in Phase I only requires labeling the vertices within the occupied blocks. To ensure successful propagation, we introduce the notion of visibility. Two blocks $B_i, B_j \in \cS_{d,n}$ are \emph{mutually visible}, defined as $B_i \sim B_j$, if 
  \[
      \sup_{x \in B_i, y \in B_j} \Vert x - y \Vert \leq (\log n)^{1/d}.\]
  Thus, if $B_i \sim B_j$, then the distance between any pair of vertices $u \in B_i$ and $v \in B_j$ is at most $(\log n)^{1/d}$. In particular, if $B_j$ is labeled and $B_i \sim B_j$, then we can propagate labels to $B_i$.
  
  Similar to the case of $\lambda > 2$, we propagate labels from left to right (Figure \ref{fig:propagate}). Despite the presence of unoccupied blocks, we establish that if $\lambda>1$ and $\chi$ is chosen appropriately, each block $B_i$ following the initial block $B_1$ has a corresponding block $B_j$ ($j<i$) to its left that is occupied and satisfies $B_i\sim B_j$. We thus modify Lines \ref{line:step3-eg}-\ref{line:step3-end-eg} of Algorithm \ref{alg:almost-exact-eg} so that a given block $B_i$ is propagated from one of the visible, occupied blocks to its left (Figure \ref{fig:propagate}). The modification is formalized in the general algorithm (Algorithm \ref{alg:almost-exact}) for any dimension $d$.

    \vskip6pt
  \paragraph*{Modified algorithm for $\frac{1}{2} < \lambda < 1$}
Under the case that $|\Omega_{\pi, P}|=1$ and Assumption \ref{ass:distinguishable-strong} (which imposes a stronger distinctness-of-distributions) holds, we can allow for smaller values of $\lambda$. When $\frac{1}{2} < \lambda < 1$, the visibility graph is disconnected, but at the same time does not have any isolated vertices (see \cite{Penrose2003}). As a result, since there exists only one permissible relabeling, we can apply the initial block labeling and the propagation algorithm to each segment of occupied blocks to obtain an unambiguous labeling, before proceeding to the refinement step to achieve exact recovery. Note that for $d=1$, the condition $\lambda \nu_d \min_{i \neq j}  D_+(\theta_i, \theta_j) > 1$ implies $\lambda > 1/2$, so this discussion completes the picture of Theorem \ref{theorem:exact-recovery} for $d=1$. This regime where the graph is disconnected but does not have isolated vertices is special to the $d=1$ case; for $d\geq 2$, a random geometric graph has the same threshold for connectivity and for the presence of isolated vertices \cite{Penrose2003}.
  
  \subsection{Exact Recovery for General \texorpdfstring{$d$}{}}\label{sec:general_d} 
  The propagation scheme is more intricate for $d \geq 2$. It is crucial that the vertices form a connected graph when $|\Omega_{\pi, P}|\geq 2$, so that the algorithm can generate a search order to propagate labels. We first present a result from Penrose \cite{Penrose1997}, showing that in the GHCM with a radius $(\log n)^{1/d}$, the condition $\lambda \nu_d > 1$ guarantees the connectivity of vertices with high probability. We consider a related model: the random geometric graph, which fixes $n$ as the number of vertices and independently generates $n$ vertices on the unit hypercube in $\mathbb{R}^d$ uniformly at random. Then, an edge is drawn between each pair of vertices of distance at most $r_n$. Observe that, conditioned on the number of vertices, the graph obtained from the vertices of GHCM and connecting any mutually visible vertices is a random geometric graph with $r_n = (\log n/n)^{1/d}$, where the scaling is due to GHCM generating points on a hypercube of volume $n$; we denote this graph as the \emph{vertex visibility graph}. The result of Penrose \cite{Penrose1997} provides a sufficient condition under which random geometric graphs are connected with high probability. 

  \begin{theorem}[
  Result (1) of \cite{Penrose1997}] 
  \label{thm:penrose}
  Consider a random geometric graph $G_n(r_n)$ with $n$ vertices generated on the unit cube in $\mathbb{R}^d, d\geq 2,$ with the toroidal metric and visibility radius $r_n$. Let $\rho_n = \min\{s\colon G_n(s) \text{ is connected}\}$ be the minimum visibility radius at which $G_n$ is connected.  
  For any constant $\alpha\in \mathbb{R}$ , if $r_n$ satisfies $ n\nu_dr_n^d - \log n = \alpha$, then
  \$
      \lim_{n\to\infty} \mathbb{P}(\rho_n \leq r_n ) = \exp(-e^{-\alpha}). 
  \$
  \end{theorem}

Applying Theorem \ref{thm:penrose} with $\lambda n$ in place of n yields a critical radius of 
\[\left(\frac{\alpha + \log n}{\lambda \nu_d n}\right)^{1/d}\]
in the unit hypercube, which scales to a radius of 
\[\left(\frac{\alpha + \log n}{\lambda \nu_d}\right)^{1/d}\]
in the volume-$n$ hypercube. Therefore, the condition $\lambda \nu_d > 1$ in the GHCM implies that the radius $(\log n)^{1/d}$ guarantees connectivity of the vertex visibility graph with high probability; this is a necessary condition for exact recovery when $|\Omega_{\pi, P}|\geq 2$. 

Next, we extend the notion of 
connectivity from vertices to blocks. For $d\geq 2$, our exact recovery algorithm divides the torus $\cS_{d,n}$ into hypercubes\footnote{For $d=1,2,3$, the hypercubes represent line segments, squares, and cubes respectively.} with volume $\chi \log n$. We show the condition $\lambda \nu_d>1$ ensures the vertex visibility graph is connected with high probability in the GHCM. Moreover, the condition ensures that every vertex has $\Omega(\log n)$ vertices within its visibility radius of $(\log n)^{1/d}$. It turns out that the condition $\lambda \nu_d>1$ also ensures that blocks of volume $\chi \log n$ for $\chi > 0$ sufficiently small satisfy the same connectivity properties. 
 
  Propagation of the labels requires an ordering to visit all occupied blocks. However, the existence of unoccupied blocks precludes the use of a predefined schedule, such as a lexicographic order scan. Instead, we employ a data-dependent schedule. The schedule is determined by the set of occupied blocks, which in turn is determined by Step \ref{sample:step-1} of Definition \ref{def:ghcm}. Crucially, the schedule is thus independent of the community labels and edges, conditioned on the number of vertices in each block. We first introduce an auxiliary graph $H = (V^{\dagger}, E^{\dagger})$, which records the connectivity relation among occupied blocks.
  \begin{definition}[Visibility graph]
  Consider a Poisson point process $V\subset\cS_{d,n}$, the $(\chi\log n)$-block partition of $\cS_{d,n}$, $\{B_i\}_{i=1}^{n/(\chi\log n)}$, corresponding vertex sets $\{V_i\}_{i=1}^{n/(\chi\log n)}$, and a constant $\delta>0$. The $(\chi, \delta)$-visibility graph is denoted by $H = (V^{\dagger}, E^{\dagger})$, where the vertex set $V^{\dagger} = \{i \in [n/(\chi\log n)] : |V_i| \geq \delta \log n\}$ consists of all $\delta$-occupied blocks and the edge set is given by $E^{\dagger}=\{\{i,j\}\colon i,j \in V^{\dagger}, B_i \sim B_j\}$.
  \end{definition}
  We adopt the standard graph connectivity definition on the visibility graph. Lemma \ref{lemma:connectivity} shows that the visibility graph of the Poisson point process underlying the GHCM is connected with high probability. Based on this connectivity property, we establish a propagation schedule as follows. We construct a spanning tree of the visibility graph and designate a root block as the initial block. We specify an ordering of $V^{\dagger} = \{i_1, i_2, \dots \}$ according to a tree traversal (e.g., breadth-first search). The algorithm propagates labels according to this ordering, thus labeling vertex sets $V_{i_1}, V_{i_2}, \cdots$ (see Figure \ref{fig:propagate}). Letting $p(i)$ denote the parent of vertex $i \in V^{\dagger}$ according to the rooted tree, we label $V_{i_j}$ using $V_{p(i_j)}$ as reference.
  Importantly, the visibility graph and thus the propagation schedule is determined only by the locations of vertices, independent of the labels and edges between mutually visible blocks.

  \begin{figure}[!ht]
  \centering
  \includegraphics[width=.95\linewidth]{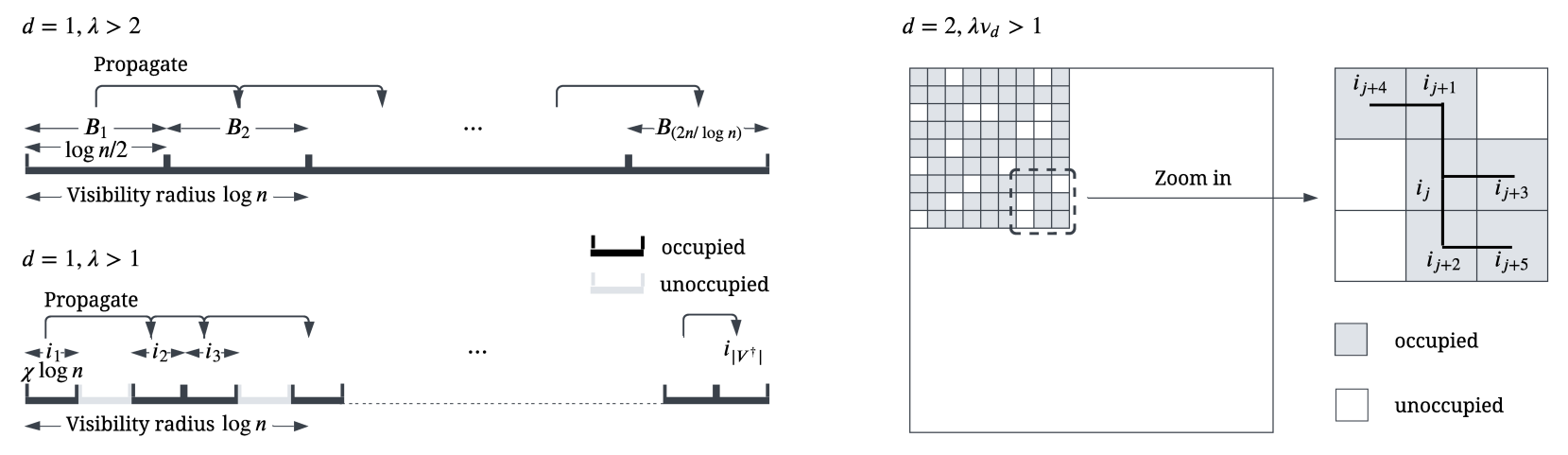}
    \caption{Propagation schedule for $d=1$ and $d=2$.}
    \label{fig:propagate}
  \end{figure}

  \begin{algorithm}
      \caption{Exact recovery for the GHCM} \label{alg:almost-exact}
      \begin{algorithmic}[1]
      \Require $G \sim \text{GHCM}(\lambda, n, \pi, P, d)$.
      \Ensure{ An estimated community labeling $\widetilde{x}: V \to Z$.}
      \vspace{5pt}
      \State{{\bf Phase I:}} 
      \State Take small enough $\chi,\delta>0$, satisfying the conditions to be specified in \eqref{eq:chi-fomula} and \eqref{eq:delta-formula} respectively. 
      \State Partition the region $\cS_{d,n}$ into $n/(\chi\log n)$ 
   blocks of volume $\chi\log n$ each. Let $B_i$ be the $i$th block and $V_i$ be the set of vertices in $B_i$ for $i\in [n/(\chi\log n)]$. \label{line:partition}
   \State Form the associated visibility graph $H = (V^{\dagger}, E^{\dagger})$. \label{line:visibility-graph}
      \If{$H$ is disconnected}\label{line:H}
      \State Return \texttt{FAIL}.
      \EndIf
      \State Find a rooted spanning tree of $H$, ordering $V^{\dagger} = \{i_1, i_2, \cdots\}$ in breadth-first order. \label{line:tree-order}
      \State Set $\varepsilon_0 \leq \min\{1/(2\log k), \delta\}$. Sample $V_{i_0}\subset V_{i_1}$ such that $|V_{i_0}| = \varepsilon_0 \log n$. Set $V_{i_1}^\prime \leftarrow V_{i_1} \setminus V_{i_0}$. 
  \State{Apply \texttt{Maximum a Posteriori} (Algorithm \ref{alg:initial-block}) on input $G, V_{i_0}$ to determine the labeling $\widehat{x}$ on $V_{i_0}$.}\label{line:step2}  
  \State{If $V_1^\prime \neq \emptyset$, apply \texttt{Propagate} (Algorithm \ref{alg:propagation}) on input $G, V_{0}, V_{1}^\prime$ to determine the labeling $\widehat{x}$ on $V_{1}^\prime$.} 
      \For{$j=2,\cdots, |V^\dagger|$} \label{line:step3}
      \State{Apply \texttt{Propagate} (Algorithm \ref{alg:propagation}) on input $G, V_{p(i_j)}, V_{i_j}$ to determine the labeling $\widehat{x}$ on $V_{i_j}$. }
      \EndFor \label{line:step3-end}
      \For{$u \in V \setminus \left(\cup_{i \in V^{\dagger}} V_i\right)$}
      \State Set $\widehat{x}(u) = *$.
      \EndFor
      \vspace{5pt}
      \State{{\bf Phase II:}}
      \For{$u\in V$}  
      \State Apply \texttt{Refine} (Algorithm \ref{alg:refine}) on input $G, \widehat{x}, u$ to determine $\widetilde{x}(u)$.
      \EndFor
      \end{algorithmic}
  \end{algorithm}
  
  Algorithm \ref{alg:almost-exact} presents our exact recovery algorithm for the general case. We partition the torus $\cS_{d,n}$ into blocks with volume $\chi \log n$, for a suitably chosen $\chi > 0$. A threshold level of occupancy $\delta>0$ is specified. The value of $\chi$ is carefully chosen to ensure that the visibility graph $H$ is connected with high probability in Line \ref{line:H}. In Line \ref{line:step2}, we label a subset of an initial $\delta$-occupied block, corresponding to the root of $H$, using the \texttt{Maximum a Posteriori} subroutine. Here $\epsilon_0\le 1/(2\log k)$ is chosen so that the number of vertices labeled by the MAP estimator is small enough to ensure the runtime of this step is $o(n \log n)$, as shown in \eqref{eq:map-runtime}. 
  In Lines \ref{line:step3}-\ref{line:step3-end}, we propagate labels of the occupied blocks along the tree order determined in Line \ref{line:tree-order}, using the \texttt{Propagate} subroutine. Vertices appearing in unoccupied blocks are assigned a label of $*$. At the end of Phase I, we obtain a first-stage labeling $\widehat{x}\colon V \to Z\cup\{*\}$, such that with high probability, all occupied blocks are labeled with few mistakes. Finally, Phase II refines the almost-exact labeling $\widehat{x}$ to an exact labeling $\widetilde{x}$ 
  according to Algorithm \ref{alg:refine}.
  
  To analyze the runtime, note that the input size, i.e., the number of edges, is $\Theta(n \log n)$ with high probability.
  The algorithm forms the visibility graph $H = (V^{\dagger}, E^{\dagger})$ in $O(n/\log n)$ time, since $|V^{\dagger}| = O(n/\log n)$ and each vertex has at most $\Theta(1)$ possible neighbors. 
  If $H$ is connected, a spanning tree can be constructed in $O(|E^{\dagger}| \log(|E^{\dagger}|))$ time using Kruskal's algorithm, in which $|E^{\dagger}| = O(n/\log n)$. Then, the \texttt{Maximum a Posteriori} subroutine labels a subset of the initial block with a runtime of $o(n \log n)$, as shown in \eqref{eq:map-runtime}. 
  Next, the \texttt{Propagation} subroutine examines all pairwise observations between any given vertex in an occupied block and the vertices in its reference block, yielding a runtime of $O(n \log n)$.
  Finally, \texttt{Refine} runs in $O(n \log n)$ time, since each visible neighborhood contains $O(\log n)$ vertices. In summary, we conclude that Algorithm \ref{alg:almost-exact} runs in $O(n \log n)$ time, which is linear in the number of edges.

  \section{Proof Outline}\label{sec:proof-outline}
  This section outlines the analysis of the main results in Section \ref{sect:limits}. First, we discuss the approach to prove the impossibility result of Theorem \ref{theorem:impossibility}. Then, we sketch the analysis of Algorithm \ref{alg:almost-exact}. For Phase I, we show that in addition to achieving almost exact recovery stated in Theorem \ref{theorem:almost-exact-recovery}, the obtained labeling also satisfies the following error dispersion property, as presented in Theorem \ref{thm:phase1-summary}. Let $\cN(u) = \{v\in V, \|u-v\|\le (\log n)^{1/d}\}$ be the set of visible vertices for a vertex $u$. For any $\eta>0$, we can take suitable $\chi, \delta >0$ such that, with high probability, every vertex has at most $\eta \log n$ incorrectly classified vertices in $\cN(u)$. Next, we outline the analysis of Theorem \ref{theorem:exact-recovery} to show Phase II of Algorithm \ref{alg:almost-exact} achieves exact recovery. Finally, we discuss the main ideas to extend Algorithm \ref{alg:almost-exact} to achieve robustness against monotone adversaries in the GSBM.

  \vskip6pt
  \paragraph*{Impossibility} We first show the impossibility result in Theorem \ref{theorem:impossibility} by generalizing the proof techniques from \cite{Abbe2021} for the two-community symmetric GSBM to the general GHCM. Since the Maximum a Posteriori (MAP) estimator is Bayes optimal, it suffices to show the MAP estimator fails with high probability. The MAP estimator fails when there exists a vertex that increases the posterior probability when labeled with an incorrect community.
  This reduces to show that a given vertex is misclassified by the genie estimator with probability $\omega(1/n)$. The misclassification probability is lower-bounded using Cram\'er's theorem of large deviations. We defer the details to Appendix \ref{sec:impossibility}.
  
  \vskip6pt
  \paragraph*{Achievability: Phase I. Connectivity of the visibility graph}
  We establish that the block partitions of the torus $\mathcal{S}_{d, n}$ specified in Algorithm \ref{alg:almost-exact} ensure the resulting visibility graph $H = (V^{\dagger}, E^{\dagger})$ is connected.
  Elementary analysis shows that any fixed subregion of $\mathcal{S}_{d, n}$ with volume $\nu \log n$ contains $\Omega(\log n)$ vertices with probability $1-o(n^{-1})$, whenever $\nu > 1/{\lambda}$.
  A union bound over vertices then implies all vertices have $\Omega(\log n)$ visible vertices. In the special case of $d=1$, the preceding block of a given vertex has volume $\log n / 2$. The observation with $\nu = 1/2$ implies when $\lambda > 2$, the preceding block of every vertex has $\Omega(\log n)$ points, guaranteeing connectivity. For smaller values of $\lambda$, we provide a stronger claim: if the block lengths are chosen sufficiently small by \eqref{eq:chi-fomula}, then we can ensure there are $\Omega(\log n)$ vertices among $\{V_j : B_j \sim B_i, j \neq i\}$ for a given vertex $v \in V_i$. As a result, an appropriate choice of $\delta$ given by \eqref{eq:delta-formula} guarantees the existence of at least one $\delta$-occupied, visible block preceding a given block $B_i$. Hence, the visibility graph is connected, as shown in Proposition \ref{lem:visibility-d1-small-lambda}.
  
  However, the analysis becomes more intricate when $d\geq 2$. In particular, while a propagation schedule based on lexicographic order succeeds for $d = 1$, it fails for $d \geq 2$. For example, when $d = 2$, we cannot claim for every vertex, its visibility region contains $\Omega(\log n)$ vertices in the top left quadrant, since the volume of the quadrant is only $\nu_d \log n/4$. Therefore, the algorithm cannot propagate labels from left to right and from top to bottom along the torus. 
  We therefore establish connectivity of $H$ using the fact that if $H$ is disconnected, then $H$ must contain an isolated connected component. If there exists an isolated connected component in $H$, then the corresponding occupied blocks in $\mathcal{S}_{d,n}$ must be surrounded by sufficiently many unoccupied blocks. However, Lemma \ref{lem:unoccupied-cluster} shows there cannot be numerous adjacent unoccupied blocks, which prevents the existence of isolated connected components. As a result, the visibility graph is connected, as shown in Lemma \ref{lemma:connectivity}.
  
  \vskip6pt
  \paragraph*{Achievability: Phase I. Labeling the initial block} We show the \texttt{MAP} subroutine in Algorithm \ref{alg:initial-block} achieves exact recovery on a subset of the first block, $V_{i_0} \subseteq V_{i_1}$, with high probability. It suffices to prove another estimator also achieves exact recovery due to the Bayes optimality of the MAP estimator. We establish this result by adapting the proof ideas from \cite[Section III.B]{Dhara2022}. Specifically, we consider the restricted MLE estimator, which maximizes the log-likelihood over $X_0^\star(\varepsilon)$, the restricted set of labelings in which the number of labels for each community is concentrated near its mean. We show all labelings in $X_0^\star(\varepsilon)$ that do not achieve exact recovery have a smaller log-likelihood than the ground-truth labeling $x^{\star}$, with high probability. We consider two cases based on the \emph{discrepancy} between a labeling and the ground truth $x^{\star}$, using a different analysis for low discrepancy and high discrepancy labelings. 
  
  \vskip6pt
  \paragraph*{Achievability: Phase I. Propagating labels among occupied blocks} 
  We show the \texttt{Propagate} subroutine ensures $\widehat{x}$ makes at most $M$ mistakes in each occupied block, where $M$ is a suitable constant. The analysis reduces to bounding the error probability that the estimator $\widehat{x}$ makes more than $M$ mistakes on a given occupied block $V_i$ for $i \in V^{\dagger}$, conditioned on the event that $\widehat x$ makes no more than $M$ mistakes on the preceding block $V_{p(i)}$. To consider the probability that a given vertex $v \in V_i$ is misclassified, we condition on the \emph{label configuration} of $V_{p(i)}$, which fixes the number of vertices correctly classified by $\widehat x$ as $j$ and incorrectly classified as $j$, for all possible labels $j \in Z$. We derive an uniform upper bound on the probability of misclassifying an individual vertex $v \in V_i$ when maximizing the log-likelihood given in Algorithm \ref{alg:propagation}, over all label configurations of $V_{p(i)}$ with at most $M$ mistakes. To bound the total number of mistakes in $V_i$, we observe the algorithm computes the labels of all vertices in $V_i$ based on disjoint subsets of pairwise observations between vertices in $V_i$ and $V_{p(i)}$. Therefore, conditioned on the label configuration of $V_{p(i)}$, the number of mistakes in $V_i$ can be stochastically dominated by a binomial random variable. Appendix \ref{sec:almost_exact_recovery} provides the detailed analysis to conclude the number of mistakes in $V_i$ is at most $M$ with probability $1-o(n^{-1})$, as long as $M$ is a suitably large constant.

  We extend Algorithm \ref{alg:almost-exact} to the case when $d=1$ and $|\Omega_{\pi, P}|=1$, i.e., the GHCM is one-dimensional with only one permissible relabeling, to achieve almost exact recovery at smaller values of $\lambda$, in which multiple segments of occupied blocks may become disconnected in the torus. We impose a stronger distinctness assumption (Assumption \ref{ass:distinguishable-strong}) so that the labelings are consistent across segments. We apply our algorithm to each segment, that is, conduct an initial labeling of a subset of the "leftmost" interval in each segment, and then propagate.
  
  \vskip6pt
  \paragraph*{Achievability: Phase II. Refining the labels} 
  The final step of Algorithm \ref{alg:almost-exact} is to refine the preliminary labeling $\widehat{x}$ from Phase I into a final labeling $\widetilde{x}$. Unfortunately, conditioning on a successful labeling $\widehat{x}$ destroys the independence of observations between pairs of vertices and makes bounding the error probability of $\widetilde{x}$ difficult. This issue can be remedied by the \emph{graph splitting} technique, used in the two-round procedure of \cite{Abbe2015}. Graph splitting forms two graphs, $G_1$ and $G_2$, from the original input graph $G$. A given edge in $G$ is independently assigned to $G_1$ with probability $p$, and $G_2$ with probability $1-p$, for $p$ chosen so that almost exact recovery can be achieved on $G_1$, while exact recovery can be achieved on $G_2$. Since the two graphs are nearly independent, conditioning on the success of almost exact recovery in $G_1$ maintains the independence of edges in $G_2$. 
  
  While we believe our Phase I algorithm, along with graph splitting, would achieve the information-theoretic threshold in the GHCM, we instead directly analyze the robustness of the Phase II procedure when refining a labeling that contains some mistakes. Specifically, we bound the error probability of classifying a given vertex $v \in V$ using the worst-case labeling among labelings that differ from the ground truth $x^\star$ on at most $\eta \log n$ vertices in the visible radius of $v$. Since $\widehat x$ makes at most $\eta \log n$ errors with probability $1-o(1/n)$ (Theorem \ref{thm:phase1-summary}), we immediately obtain a bound on the error probability of $\widetilde{x}(v)$. The proof in Section \ref{sec:exact_recovery} bounds the worst-case error probability. 
  To provide intuition for bounding the error probability at a given vertex $u\in V$, suppose that $x^\star(u)=i$, and consider 
  the genie-aided estimator $x_{\textsf{genie}}(u)$. Recalling the definition of $\ell_i(u,x)$ in \eqref{eq:ell_global}, the estimator $x_{\textsf{genie}}(u)$ makes a mistake when there exists $j\neq i$ such that $\ell_i(u,x^\star) - \ell_j(u,x^\star)\le 0$. We show this event occurs with probability at most $n^{-\lambda\nu_d D_+(\theta_i,\theta_j)}$. By formulating the worst-case labeling $x$ as a perturbation of $x^\star$ that differs from $x^\star$ on at most $\eta\log n$ vertices, we show $\ell_i(u,x)- \ell_j(u,x)\le  0$ implies $\ell_i(u,x^\star)- \ell_j(u,x^\star) \le \rho\eta\log n$ for a certain constant $\rho$. Similarly, the probability of such a mistake is at most $n^{-\lambda\nu_d D_+(\theta_i,\theta_j)+\rho\eta/2}$. 
  Thus, for small $\eta>0$, the condition $\lambda\nu_d D_+(\theta_i,\theta_j) > 1$ and a union bound over all vertices yields an error probability of $o(1)$, which certifies exact recovery.

  \vskip6pt
  \paragraph*{Robustness under monotone adversaries} We first show that for
$k \ge 3$, the monotone adversary can easily break the distinctness of distributions required in Assumption \ref{ass:distinguishable}. Therefore, we focus on the $k=2$ case with the following recovery goal. When $\pi_{-1} = \pi_{1} = 1/2$, we can only recover the correct partition, not the labeling, since the adversary can simulate $\text{GSBM}(\lambda, n, \pi, P^\prime,d)$ with $P_{-1,-1}^\prime = P_{11}^\prime$ by adding random edges.
When $\pi_{-1} \neq \pi_{1}$, the goal is to find the correct labeling. We adapt our local propagation scheme to establish Theorem \ref{thm:monotone}. For the initial block labeling, we use techniques from \cite[Algorithm 3]{Liu2022} that guarantee fewer than $\delta_1\log n$ errors in the semi-random model for any $\delta_1>0$. We then prove the propagation and refinement procedures still succeed despite the existence of monotone corruptions.

\vskip6pt
\paragraph*{Gaussian GHCM}
In addition to the distinctness requirement of Assumptions \ref{ass:distinguishable} and \ref{ass:distinguishable-strong}, our results require the log-likelihood ratio to be bounded by some $\eta$ (Assumption \ref{ass:bounded-ratio}) for any $i, j\in Z$. This assumption is used in three parts. First, in the exact recovery result of the \texttt{MAP} subroutine on the initial block, the assumption is needed to show that low-discrepancy labelings have a smaller likelihood than $x^\star$, with high probability. Second, we use the assumption to bound the error probability of misclassifying a vertex using the \texttt{Propagate} subroutine. Third, the assumption is needed to bound the error rate of the \texttt{Refine} procedure of Phase II. However, we can relax this assumption in certain cases of the GHCM, highlighting a special example when all distributions are Gaussian but still satisfy the distinctness assumption. Such a case gives rise to the geometric $\Z_2$ synchronization. In Appendix \ref{sec:gaussian_proofs}, we show that the above three results still hold in the Gaussian case; essentially, instead of providing bounds using $\eta$, we provide sufficiently strong probabilistic bounds using Gaussian concentration. As a result, our algorithm still achieves exact recovery even in the Gaussian case. 
 
  \section{Discussion and Future Work}
  \label{sec:discussion}
  In this paper, we have introduced the GHCM as a model of hidden communities in a spatial network. We have established a threshold below which no algorithm achieves exact recovery of the labels (in fact, all algorithms fail with high probability), and above which there is an efficient algorithm for exact recovery, under a mild distinctness condition. A key open problem is whether Theorem \ref{theorem:exact-recovery} continues to hold without the need for Assumption \ref{ass:distinguishable}. 
  One promising approach is to again divide the torus into blocks containing $\Theta(\log n)$ points in expectation, but perform a local label estimation in each block via the MLE. Such a procedure is computationally tractable, as running the MLE on $O(\log n)$ points requires $\text{poly}(n)$ time. The difficulty with such an approach lies in aggregating the local labels, when there are sparsely populated blocks. A related question is regarding statistical achievability: can we show exact recovery is achieved by the MAP estimator above the threshold conjectured in Conjecture \ref{conjecture}? Further, what can we say about exact recovery when pairwise observations have a more general dependence on distance, as in \cite{Abbe2021,Avrachenkov2024}?
  In addition, it is valuable to explore the robustness of these methods in the presence of geometric corruptions. For example, is there an algorithm that can recover the community labels if some entries of the location vectors are hidden, or if we only have access to a low-dimensional projection of the locations?
  
  
  Another direction is to investigate the model beyond the logarithmic degree regime. We expect that almost exact recovery is possible in sublogarithmic regimes, though possibly only when $|\Omega_{\pi, P}| = 1$. In sparse regimes, where each vertex sees only a constant number of other vertices, can we obtain results analogous to \cite{Abbe2021}? What can we say about the exact recovery problem when the observation signal strength is decreased while the visibility radius is increased? How does the information-theoretic threshold change, and does some variant of our algorithm achieve the information-theoretic threshold? Since the analysis of the propagation scheme is intimately connected to the connectivity properties of the visibility graph formed from the underlying Poisson point process, this question is related to work on connectivity of random geometric graphs \citep{Penrose2016}.
    
\paragraph*{Acknowledgements}
X.N. and E.W. were supported in
part by NSF ECCS-2030251 and CMMI-2024774. X.N. was supported in part by a Northwestern University Terminal Year Fellowship. C.G. and X.N. were supported in part by  NSF HDR TRIPODS 2216970. J.G. and C. G. were supported in part by NSF CCF-2154100, and J.G. was supported in part by NSF CAREER CCF-2440539.
 Thank you to Liren Shan for suggesting the study of the monotone adversary. Thank you to Colin Sandon for helpful comments in the proof of Proposition \ref{prop:high-discrepancy}, and to Andrew Jin for pointing out several errors in an earlier version. Thank you to Nicolas Fraiman for suggesting the work of Penrose \citep{Penrose2016}, and to Jiaming Xu for making us aware of \cite{Chen2016b}.

\bibliographystyle{plain} 
\bibliography{references}   

\appendix

\section{Additional Notation}
Throughout the analysis in the remaining sections, we rely on the following notations. For any $t\in [0,1]$, we define $\phi_t(\cdot, \cdot)$ on distributions $p,q\colon \cX\to \R$ as the moment-generating function of $\log(p(X)/q(X))$ evaluated at $t$, where  $X\sim q(\cdot)$, 
\#\label{eq:define-phi-t}
\phi_t(p,q) & = \E_q\Big[\exp\Big(t\log\big(\frac{p(X)}{q(X)}\big)\Big)\Big] = \E_q\Big[\Big(\frac{p(X)}{q(X)}\Big)^t\Big] \notag\\
& = \begin{cases}
\displaystyle\sum_{x\in\cX}p(x)^tq(x)^{1-t} & \mbox{for discrete } p, q; \\
\displaystyle\int_{x\in \cX} p(x)^tq(x)^{1-t} \, dx & \mbox{for continuous } p, q.
\end{cases}
\#
Let $t_{ij}$ be any arbitrary value in the following minimizing set
\#\label{eq:t_ij}
t_{ij} \in \argmin_{t \in [0,1]} \sum_{r=1}^k \pi_r \cdot \phi_t(p_{ir}, p_{jr}).
\# 
Then, the CH divergence can be expressed as
\#\label{eq:CH-phi}
D_+(\theta_i, \theta_j) = 1 - \sum_{r=1}^k \pi_r \cdot \phi_{t_{ij}}(p_{ir}, p_{jr}) =\sum_{r=1}^k \pi_r \big(1-\phi_{t_{ij}}(p_{ir}, p_{jr}) \big)>0.
\#
Thus, there exists $r_{ij}\in Z$ such that 
\#\label{eq:r_ij}
\phi_{t_{ij}}(p_{i,r_{ij}}, p_{j,r_{ij}}) <1.
\#
Moreover, we define $\rho$ as 
\#\label{eq:rho}
\rho = \max_{i\neq j} \phi_{t_{ij}}(p_{i,r_{ij}}, p_{j,r_{ij}}),
\#
and note that $\rho<1$ under the distinctness assumption (Assumption \ref{ass:distinguishable}).

\section{Proof of Impossibility (Theorem \ref{theorem:impossibility})}
\label{sec:impossibility}
In this section, we show exact recovery is impossible under the conditions in Theorem \ref{theorem:impossibility}. We first establish part (a), which states that any estimator fails at exact recovery below the information-theoretic threshold. \cite{Abbe2021} proved this result holds for the two-community symmetric GSBM. We generalize their proof techniques to the GHCM.
Since the Maximum a Posteriori (MAP) estimator is Bayes optimal, it suffices to show the MAP estimator fails with high probability. We further split the proof into two cases: $\lambda\nu_d < 1$ and $\lambda \nu_d \geq 1$. 

When $\lambda\nu_d < 1$, recall that the vertex visibility graph is a random geometric graph. The random geometric graph $RGG(d, n, x)$ is defined by distributing vertices within the $d$-dimensional unit torus by a Poisson point process of intensity $n$ and forming an edge between any pair of vertices whose distance is less than $x$. Re-scaling the volume $n$ torus by $1/n$ shows that the vertex visibility graph of $G$ is $RGG(d, \lambda n, (\frac{\log n}{n})^{1/d})$. We use the following result on isolated vertices in random geometric graphs. 

\begin{lemma}[Theorem 7.1 in \cite{Penrose2003}]
\label{lem:isolated_vertex_RGG}
    Denote 
    \[M_n:=\inf_{x\in\mathbb{R}}\{RGG(d,n,x) \text{ has no isolated vertex}\}.\] 
    Then, almost surely,
    \[ \lim_{n\to\infty} \frac{n\nu_d M_n^d}{\log n} = 1. \]
\end{lemma}
For a Poisson point process of intensity $\lambda n$, we obtain
\[ \lim_{n\to\infty} \frac{\lambda n\nu_d M_{\lambda n}^d}{\log \lambda n} = 1. \]
By definition, there will be an isolated vertex if the visibility radius is less than $M_n$. Since $M_{\lambda n} \sim (\frac{\log(\lambda n)}{\lambda n \nu_d})^{1/d}$, taking $\lambda \nu_d < 1$ implies that $M_{\lambda n} \gg \left(\frac{\log n}{n}\right)^{1/d}$, from which it follows that an isolated vertex exists with high probability. The MAP estimator fails to label the isolated vertex correctly with a constant nonzero probability. Thus, the MAP estimator fails, implying that any estimator fails to achieve exact recovery when $\lambda \nu_d < 1.$

Next, consider when $\lambda\nu_d \geq 1$. When $G \sim \text{GHCM}(\lambda, n, \pi, P, d)$, the posterior probability of a labeling $x$ conditioned on the realization $\{G=g\}$ is given by 
\[Q(x; g) := Z(V_g) \prod_{\{u,v\}\in V_g \times V_g, u\sim v} \pi_{x(u)} p_{x(u),x(v)}(y_{uv}),\]
where $V_g$ denotes the vertex set of $g$, and $Z(V_g)$ is a normalization factor that depends only on $V_g$.

A sufficient condition for the MAP estimator to fail is the existence of a vertex that increases the posterior probability when labeled with an incorrect community. Formally, if $\widetilde{x}$ is a labeling which differs from $x^\star$ only on vertex $u$, and 
\[  Q(\widetilde x ; G ) > Q(x^{\star} ; G ),\]
then we define $u$ as $\flipbad$. We define the random variable $W_n=\sum_{u\in V}\mathds{1}\{u \isflipbad \}$ for a graph with volume $n$. The event $\{W_n>0\}$ implies there exists at least one $\flipbad$ vertex. 
\begin{proposition}
    \label{prop:impossibility}
    If the event $\{W_n > 0\}$ occurs, then the MAP estimator fails at exact recovery. 
\end{proposition}
\begin{proof}
The event $\{W_n > 0\}$ implies that there exists a vertex $u\in V$ such that $Q(\widetilde x ; G ) > Q(x^{\star} ; G )$, 
where $\widetilde x$ differs from $x^\star$ only at $u$. By the definition of permissible relabelings (Definition \ref{def:permissible}), for any $\omega\in\Omega_{\pi, P}$, we have
\[ Q(\omega \circ x^{\star} ; G ) = Q(x^{\star} ; G ) < Q(\widetilde x ; G )\]
As a result, the MAP estimator will not return any element of $\{\omega\circ x^\star\colon \omega\in\Omega_{\pi, P}\}$, failing at exact recovery.
\end{proof}
Due to Proposition \ref{prop:impossibility}, it is sufficient to show that $\pr(W_n=0)= o(1)$. By Chebyshev's inequality, we have 
\begin{align*}
    \pr(W_n=0) &\leq \pr\big(| W_n-\mathbb{E}W_n | \geq \mathbb{E}W_n\big) \\
                 &\leq \frac{\var(W_n)}{\mathbb{E}[W_n]^2} \\
                 &= \frac{\mathbb{E}[W_n^2]}{\mathbb{E}[W_n]^2} - 1.
\end{align*}
Therefore, the condition 
\begin{align}
    &\limsup_{n\to\infty}\frac{\mathbb{E}[W_n^2]}{\mathbb{E}[W_n]^2}\leq 1 \label{flip-bad2} 
\end{align}
is sufficient to imply $\pr(W_n=0)=o(1)$. Let $G \sim \text{GHCM}(\lambda,n,\pi,P,d)$. 
Then, denoting $F_u^G$ as the event that vertex $u$ is $\flipbad$ in graph $G$, Proposition 8.1 in \cite{Abbe2021} implies that the following two conditions are sufficient to guarantee 
\eqref{flip-bad2},
\begin{align}
&\lim_{n\to\infty}n\mathbb{E}^0[\mathds{1}\{F_0^{G\cup\{0\}}\}] = \infty     \label{cond1} \\
    &\limsup_{n\to\infty} \frac{\int_{y\in \mathcal{S}_{d,n}}\mathbb{E}^{0, y}[\mathds{1}\{F_0^{G\cup\{0, y\}}\}\mathds{1}\{F_y^{G\cup\{0, y\}}\}]m_{n, d}(dy)}{n\mathbb{E}^0[\mathds{1}\{F_0^{G\cup\{0\}}\}]^2} \leq 1,  \label{cond2}
\end{align}
where $m_{n, d}$ is the Haar measure. $\mathbb{E}^{0}$ is the expectation with respect to $G\cup\{0\}$, which is the geometric graph obtained by adding a vertex at the origin of the torus with a random label sampled according to $\pi$, and sampling its pairwise observations according to $\text{GHCM}(\lambda, n, \pi, P, d)$. Similarly, $\mathbb{E}^{0, y}$ is the expectation with respect to the graph obtained by adding two vertices at the origin and $y$ and sampling their pairwise observations. The reason for introducing vertices at $0$ and at $y$ is due to an application of Campbell's theorem; see \cite{Abbe2021} for further details. In the next two subsections, we establish that conditions \eqref{cond1} and \eqref{cond2} hold whenever $\lambda \nu_d \min_{i \neq j} D_+(\theta_i, \theta_j) < 1$, which proves Theorem \ref{theorem:impossibility}, part (a).

\subsection{Proof of Condition \eqref{cond1}}
To prove Condition \eqref{cond1}, we apply Cram\'er's theorem of large deviations. Denote $\Lambda_X(t)$ as the cumulant-generating function (CGF) of any random variable $X$:
\begin{align*}
    \Lambda_X(t) &= \log \mathbb{E}[\exp(tX)].
\end{align*}

\begin{lemma}[\cite{dembo_zeitouni_LDP}]
    \label{lem:cgf_convex} For any random variable $X$, $\Lambda_X(t)$ is convex. If $X$ is nondegenerate (i.e., the support is not a single point), then $\Lambda_X(t)$ is strictly convex.
\end{lemma}
\begin{theorem}[Cram\'er's Theorem \citep{dembo_zeitouni_LDP}]  Let $X_1, X_2, \hdots$ be i.i.d. random variables with finite CGF $\Lambda_{X}$ and define its Legendre transform as
\[\Lambda_{X}^\star(\alpha) = \sup_{t\in\mathbb{R}}(t\alpha-\Lambda_X(t)). \]
For all $\alpha > \mathbb{E}[X_1]$, we have 
\[\lim_{n \to \infty} \frac{1}{n} \log \mathbb{P}\Big(\sum_{r = 1}^n X_{r} \geq n\alpha \Big) = - \Lambda_X^{\star}(\alpha).\]
\label{thm:cramer}
\end{theorem}

\begin{lemma}\label{cond1_pf}
For any $\lambda>0, d\in\mathbb{N}$, $\pi\in\mathbb{R}^k$, and $P$ such that $\lambda \nu_d \min_{i \neq j}  D_+(\theta_i, \theta_j) < 1$, there exists $\beta>0$ such that $\pr(0\isflipbad\text{ in }G\cup\{0\}) \geq n^{-1+\beta}$. It follows that Condition \eqref{cond1} holds.
\end{lemma}
\begin{proof}
     Note Condition \eqref{cond1} is equivalent to $\pr(0\isflipbad\text{ in }G\cup\{0\}) = \omega(1/n)$, where $0$ is the vertex added at the origin of $\mathcal S_{d, n}$. Thus, showing $\pr(0\isflipbad\text{ in }G\cup\{0\}) \geq n^{-1+\beta}$ implies Condition \eqref{cond1}. Let $i\neq j$ be two communities such that $\lambda \nu_d D_+(\theta_i, \theta_j) < 1$. Let $\widetilde x$ be a labeling such that $\widetilde x(u) = x^\star(u)$ for $u\in V$ and set $\widetilde x(0) = j$. 
     We want to show that the following event occurs with probability at least $n^{-1+\beta}$:
     \[ \{ Q(\widetilde{x} ; G \cup \{0\}) > Q(x^{\star} ; G \cup \{0\})\} \cap \{x^{\star}(0) = i\}. \]
     Therefore, it is sufficient to show
\begin{align}
   &\pi_j \prod_{v\in V, v\sim 0} p_{j, x^\star(v)} (y_{0,v}) > \pi_i \prod_{v\in V, v\sim 0} p_{i, x^\star(v)} (y_{0,v}) \label{eqn:zero_flipbad}
\end{align}
holds with probability at least $n^{-1+\beta}$, conditioned on $\{x^{\star}(0) = i\}$. By rearranging terms and applying the logarithm, the inequality in \eqref{eqn:zero_flipbad} is equivalent to 
\begin{align}
    \label{eqn:zero_flipbad2}
   &-\log\Big(\frac{\pi_i}{\pi_j}\Big) +  \sum_{v\in V, v\sim 0} \log\Big( \frac{p_{j, x^\star(v)} (y_{0,v})}{p_{i, x^\star(v)} (y_{0,v})}\Big) > 0.
\end{align}
Let  $\{X_{v}^{(i)}\}_{v\in V}$ be i.i.d. replicates drawn from a mixture distribution with weights $\{\pi_{s}\}$. For a fixed $v$, with probability $\pi_s$ independent of everything else, $X_v^{(i)}$ is distributed as $\log({p_{js}(Y^{(i, s)})}/{p_{is}(Y^{(i, s)})})$, 
where $Y^{(i, s)}$ comes from the $P_{is}$ distribution. Note that $X_v^{(i)}$ and $\log( {p_{j, x^\star(v)} (y_{0,v})}/{p_{i, x^\star(v)} (y_{0,v})})$ have the same distribution, conditioned on $\{x^{\star}(0) = i\}$. Thus, the inequality in \eqref{eqn:zero_flipbad2} can be expressed as
\begin{align*}
   &-\log\Big(\frac{\pi_i}{\pi_j}\Big) +  \sum_{v\in V, v\sim 0} X_v^{(i)}   > 0.
\end{align*}
Denote $N(0)$ as the number of visible neighbors of the vertex at the origin. By conditioning on $N(0)$, we obtain 
\begin{align}
    & \mathbb{P}\big(0 \text{ is }\textit{Flip-Bad}\text{ in }G\cup\{0\} \mid x^\star(0)=i\big) \nonumber \\
    &\qquad= \mathbb{P}\Big(\sum_{v\sim 0} X_v^{(i)} > \log\Big(\frac{\pi_i}{\pi_j}\Big) \Big) \nonumber \\
     &\qquad= \sum_{m=0}^{\infty} \mathbb{P}(N(0)=m) \mathbb{P}\Big(\sum_{r=1}^m X_r^{(i)} > \log\Big(\frac{\pi_i}{\pi_j}\Big) \Biggiven N(0)=m\Big),\label{eq:flipbad_cond}
\end{align}
where $\{X_r^{(i)}\}$ are i.i.d. copies of the mixture distribution defined above. Our approach is to bound the tail event $\mathbb{P}(\sum_{r=1}^m X_r^{(i)} > \log({\pi_i}/{\pi_j}) \mid N(0)=m)$ using large deviation theory.  
Recalling $\phi_t$ defined in \eqref{eq:define-phi-t}, we compute the moment generating function of $X_r^{(i)}$ as
\begin{align*}
    \mathbb{E}[\exp(tX_r^{(i)})] &= \sum_{a\in Z} \pi_a \phi_t(p_{ja}, p_{ia}).
\end{align*}
Thus, we have the CGF as follows,
\begin{align}
    \Lambda(t) &= \log\Big( \sum_{a\in Z} \pi_a \phi_t(p_{ja}, p_{ia}) \Big). \label{eq:Lambda}
\end{align}
Recall that $t_{ij}$ defined in \eqref{eq:t_ij} minimizes $\Lambda(t)$. We will show $\Lambda^{\star}(\alpha) \approx 
  -\Lambda(t_{ij})$ for $\alpha > 0$ sufficiently small. To this end, we will show that for any $\eta > 0$, there exists $\alpha >0$ such that we can restrict the supremum in Theorem \ref{thm:cramer} to $\{t : t_{ij} \leq t \leq t_{ij} + \eta \}$; in other words, we will show that $\Lambda^{\star}(\alpha) = \sup_{t : t_{ij} \leq t \leq t_{ij} + \eta } (t\alpha - \Lambda(t))$.

For any $t, t_0\in\mathbb{R}$, the tangent approximation and the convexity of $\Lambda(\cdot)$ yield
\#\label{eq:tangent-Lambda}
-\Lambda(t)\le -\Lambda(t_0) - \Lambda'(t_0)(t-t_0).
\#
For any $\eta>0$, the strict convexity in Lemma \ref{lem:cgf_convex} ensures $\Lambda'(t_{ij} + \eta)>0$. We choose $\alpha >0$ such that $\alpha < \Lambda'(t_{ij} + \eta)$. For the case $t> t_{ij} + \eta$, we take $t_0 = t_{ij} + \eta$ in \eqref{eq:tangent-Lambda} and obtain
\$
- \Lambda(t) \le -\Lambda(t_{ij} + \eta) -\Lambda'(t_{ij} + \eta)\big[t-(t_{ij} + \eta)\big].
\$
For any $t>t_{ij}+\eta$, substituting the above bound and considering the condition $\alpha < \Lambda'(t_{ij} + \eta)$ yield 
\$
t\alpha - \Lambda(t) & \le t\alpha -\Lambda(t_{ij} + \eta) -\Lambda'(t_{ij} + \eta) \big[t-(t_{ij} + \eta)\big] \notag\\
&= (t_{ij} + \eta)\alpha -\Lambda(t_{ij} + \eta) +\big[\alpha-\Lambda'(t_{ij} + \eta)\big] \big[t-(t_{ij} + \eta)\big] \notag\\
&\le (t_{ij} + \eta)\alpha -\Lambda(t_{ij} + \eta).
\$
For the case $t < t_{ij}$, recall $t_{ij}\in[0, 1]$ by definition. Thus, $t\alpha - \Lambda(t) < t_{ij} \alpha - \Lambda(t_{ij})$. We conclude that for fixed $\eta >0$ and $0<\alpha < \Lambda'(t_{ij} + \eta) $, the Legendre transform of $X_r^{(i)}$ is bounded by
\begin{align}
    \Lambda^{\star}(\alpha) &= \sup_{t : t_{ij} \leq t \leq t_{ij} + \eta } (t\alpha - \Lambda(t)) \nonumber \\
    &\leq \sup_{t : t_{ij} \leq t \leq t_{ij} + \eta } \big\{t\alpha\big\} + \sup_{t : t_{ij} \leq t \leq t_{ij} + \eta } \big\{ - \Lambda(t)\big\} \nonumber \\
    &= (t_{ij} + \eta)\alpha -\Lambda(t_{ij}) \label{eqn:Lambda_star}.
\end{align}
By Theorem \ref{thm:cramer} and noting $\mathbb{E}[X_r^{(i)}] \leq 0$, for our choice of $\alpha>0$ and for any $\varepsilon > 0$, there exists $m_0$ such that for all $m\geq m_0$ 
\begin{align*}
    \Big|\frac{1}{m} \log \mathbb{P}\Big(\sum_{r = 1}^m X_{r}^{(i)} \geq m\alpha \Big)  + \Lambda^{\star}(\alpha) \Big| \leq \varepsilon. 
\end{align*}
Rearranging terms yields
\begin{align}
     \mathbb{P}\Big(\sum_{r = 1}^m X_{r}^{(i)} \geq m\alpha \Big) \geq \exp\big(-m(\varepsilon+\Lambda^\star(\alpha))\big). \label{eqn:Cramer}
\end{align}
Note $N(0)\sim \text{Poisson}(\lambda\nu_d \log n)$. The Chernoff bound (Lemma \ref{lem:Chernoff-poisson}) yields
\begin{align*}
    \mathbb{P}(N(0) < \delta \log n) 
    &\leq \exp \Big( \delta\log n - \lambda \nu_d \log n - \log\Big(\frac{\delta}{\lambda \nu_d}\Big) \delta \log n \Big) \\
    &= n^{\delta - \lambda \nu_d - \delta \log(\delta /(\lambda \nu_d))}.
\end{align*}
Let $f(\delta):=\delta - \lambda \nu_d - \delta \log(\delta /(\lambda\ \nu_d))$. As $\delta \to 0$, we have $f(\delta) \to -\lambda \nu_d$ and $f(\delta)$ is strictly increasing on $\delta \in (0, \lambda \nu_d).$ Thus, we can choose $\delta$ small enough such that $f(\delta) < -\lambda\nu_d + c$. Since we assumed $\lambda \nu_d \geq 1$, we obtain $f(\delta) < -1 + c$. As a result, $\mathbb{P}(|N(0)| < \delta \log n) \leq n^{-1+c}$. For $n$ sufficiently large, we further have $\delta\log n\geq m_0$ and $\alpha\delta\log n \geq \log({\pi_i}/{\pi_j})$. Then, we lower bound (\ref{eq:flipbad_cond}) as 
\begin{align}
    \mathbb{P}\Big(\sum_{v\sim 0} X_v^{(i)} > \log\Big(\frac{\pi_i}{\pi_j}\Big) \Big) &\geq \sum_{m=\delta\log n}^{\infty} \mathbb{P}(N(0)=m) \mathbb{P}\Big(\sum_{r=1}^m X_r^{(i)} > \log\Big(\frac{\pi_i}{\pi_j}\Big) \Biggiven N(0)=m\Big) \nonumber \\
    &\geq \sum_{m=\delta\log n}^{\infty} \mathbb{P}(N(0)=m) \mathbb{P}\Big(\sum_{r=1}^m X_r^{(i)} > m\alpha\Big) \Biggiven N(0)=m\Big) \nonumber \\
    &\geq \sum_{m=\delta\log n}^{\infty} \mathbb{P}(N(0)=m) \exp\big(-m(\varepsilon+\Lambda^\star(\alpha))\big) \nonumber \\
     &\geq \sum_{m=\delta\log n}^{\infty} \mathbb{P}(N(0)=m) \exp\big(-m(\varepsilon+(t_{ij} + \eta)\alpha -\Lambda(t_{ij}))\big), \label{eq:flipbad_cond2}
\end{align}
where the last two inequalities follow from \eqref{eqn:Cramer} and \eqref{eqn:Lambda_star}, respectively. Observe that \eqref{eq:flipbad_cond2} is almost the moment generating function of $N(0)$ evaluated at $\varepsilon+(t_{ij} + \eta)\alpha -\Lambda(t_{ij})$. By introducing the first $\delta\log n$ terms, (\ref{eq:flipbad_cond2}) is equivalent to
\begin{align*}
    & \sum_{m=0}^{\infty} \mathbb{P}(N(0)=m) \exp\big(-m(\varepsilon+(t_{ij} + \eta)\alpha -\Lambda(t_{ij}))\big) \\
    &\qquad \qquad - \sum_{m=0}^{\delta\log n - 1} \mathbb{P}(N(0)=m) \exp\big(-m(\varepsilon+(t_{ij} + \eta)\alpha -\Lambda(t_{ij}))\big)  \\
     &\qquad\geq \sum_{m=0}^{\infty} \mathbb{P}(N(0)=m) \exp\big(-m(\varepsilon+(t_{ij} + \eta)\alpha -\Lambda(t_{ij}))\big) - \mathbb{P}(N(0) < \delta\log n) \\
     &\qquad\geq \sum_{m=0}^{\infty} \mathbb{P}(N(0)=m) \exp\big(-m(\varepsilon+(t_{ij} + \eta)\alpha -\Lambda(t_{ij}))\big) - n^{-1+c}.
\end{align*}
Substituting the moment-generating function of a Poisson random variable yields
\begin{align}
    &\mathbb{P}\bigg(\sum_{v\sim 0} X_v^{(i)} > \log\Big(\frac{\pi_i}{\pi_j}\Big) \bigg) \\
    &\geq \sum_{m=0}^{\infty} \mathbb{P}(N(0)=m) \exp\big(-m(\varepsilon+(t_{ij} + \eta)\alpha -\Lambda(t_{ij}))\big) - n^{-1+c} \nonumber \\
    &= \exp\bigg(\lambda \nu_d \log n \Big(\exp\big(-\varepsilon-(t_{ij} + \eta)\alpha +\Lambda(t_{ij})\big) 
 -1\Big)\bigg) - n^{-1+c} \nonumber\\
    &= n^{\lambda \nu_d \big(\exp(-\varepsilon-(t_{ij} + \eta)\alpha +\Lambda(t_{ij})) 
 -1\big)} - n^{-1+c}. \label{eqn:imp_ch} 
\end{align}

To show \eqref{eqn:imp_ch} is lower-bounded by $n^{-1+\beta}$, it suffices to show 
\begin{equation}
\lambda \nu_d \Big(\exp\big(-\varepsilon-(t_{ij} + \eta)\alpha +\Lambda(t_{ij})\big) 
 -1\Big) > -1 + 2\beta,   \label{eq:less-than-1b}
\end{equation}
which lower-bounds \eqref{eqn:imp_ch} by $n^{-1+2\beta} - n^{-1+c}$.
Since $c>0$ can be arbitrarily small by taking $\delta >0$ small enough, it follows that \eqref{eqn:imp_ch} is lower-bounded by $n^{-1+\beta}$ for $c < 2\beta$ and $n$ large enough.

Substituting the definition of $\Lambda(t_{ij})$ in \eqref{eq:Lambda}, the left hand side of \eqref{eq:less-than-1b} is equal to
\$
&\lambda \nu_d \bigg[\exp\bigg\{-(\varepsilon+(t_{ij} + \eta)\alpha + \log\bigg( \sum_{r\in Z} \pi_r \phi_{t_{ij}}(p_{jr}, p_{ir})\bigg)\bigg\} 
 -1\bigg] \nonumber\\
 &=\lambda \nu_d \bigg[\sum_{r\in Z} \pi_r \phi_{t_{ij}}(p_{jr}, p_{ir}) \cdot e^{-(\varepsilon+(t_{ij} + \eta)\alpha)} 
 -1\bigg] \nonumber\\
 &= \lambda \nu_d \big[ (1-D_+(\theta_i, \theta_j)) \cdot e^{-(\varepsilon+(t_{ij} + \eta)\alpha)} - 1\big]. 
\$

Finally, observe that for any $\delta > 0$, there exist choices of $\varepsilon, \eta, \alpha$ such that 
\begin{align}
   \lambda \nu_d \left[ (1-D_+(\theta_i, \theta_j)) \cdot e^{-(\varepsilon+(t_{ij} + \eta)\alpha)} - 1\right] &\geq \lambda \nu_d \left[ (1-D_+(\theta_i, \theta_j)) \cdot  (1-\delta) - 1 \right] \nonumber \\
    &= -\lambda \nu_d \left[ D_+(\theta_i, \theta_j)(1-\delta) + \delta \right]. \label{eq:ch_delta} 
\end{align}
The limit of \eqref{eq:ch_delta} as $\delta$ approaches $0$ is $-\lambda \nu_dD_+(\theta_i, \theta_j) > -1$. Let \[\beta = (1 - \lambda \nu_d D_+(\theta_i, \theta_j) )/3 > 0.\] Therefore, there exists $\delta_0 > 0$ such that for all $\delta \in (0,\delta_0]$, \[-\lambda \nu_d \left[ D_+(\theta_i, \theta_j)(1-\delta) + \delta \right] > -1 + 2\beta\] holds, verifying \eqref{eq:less-than-1b}.
\end{proof}

\subsection{Proof of Condition \eqref{cond2}}

\begin{lemma}\label{cond2_pf}
    For any $\lambda>0, d\in\mathbb{N}$, $\pi\in\mathbb{R}^k$, and $P$ such that $\lambda \nu_d \min_{i \neq j}  D_+(\theta_i, \theta_j) < 1$, Condition \eqref{cond2} holds.
\end{lemma}
Lemma \ref{cond2_pf} is a direct consequence of Lemma \ref{cond1_pf} and independence of disjoint regions of the Poisson point process. Its proof is an adaptation of the proof of \cite[Lemma 8.3]{Abbe2020b}.
\begin{proof}
     Defining $B(0, r) = \{x \in \cS_{d,n} : \Vert x\Vert_2 \leq r\}$ as the Euclidean ball around the origin with radius $r$, we obtain 
\begin{align}\label{eq:torus_decomposed}
    &\int_{y\in \mathcal{S}_{d, n}}\mathbb{E}^{0, y}[\mathds{1}\{F_0^{G\cup\{0, y\}}\}\mathds{1}\{F_y^{G\cup\{0, y\}}\}]m_{n, d}(dy) \nonumber \\
    &\quad= \int_{y\in B(0, 2\log(n)^{1/d})}\mathbb{E}^{0, y}[\mathds{1}\{F_0^{G\cup\{0, y\}}\}\mathds{1}\{F_y^{G\cup\{0, y\}}\}]m_{n, d}(dy)  \nonumber \\
    &\qquad +\int_{y\in \mathcal{S}_{d, n}\cap B(0, 2\log(n)^{1/d})^c}\mathbb{E}^{0, y}[\mathds{1}\{F_0^{G\cup\{0, y\}}\}\mathds{1}\{F_y^{G\cup\{0, y\}}\}]m_{n, d}(dy) \nonumber \\
    &\quad\leq \int_{y\in B(0, 2\log(n)^{1/d})}\mathbb{E}^{0, y}[\mathds{1}\{F_0^{G\cup\{0, y\}}\}]m_{n, d}(dy)  \nonumber \\
    &\qquad +\int_{y\in \mathcal{S}_{d, n}\cap B(0, 2\log(n)^{1/d})^c}\mathbb{E}^{0, y}[\mathds{1}\{F_0^{G\cup\{0, y\}}\}\mathds{1}\{F_y^{G\cup\{0, y\}}\}]m_{n, d}(dy) \nonumber \\
    &\quad= \int_{y\in B(0, 2\log(n)^{1/d})}\mathbb{E}^{0, y}[\mathds{1}\{F_0^{G\cup\{0, y\}}\}]m_{n, d}(dy)   \\
    &\qquad +\int_{y\in \mathcal{S}_{d, n}\cap B(0, 2\log(n)^{1/d})^c}\mathbb{E}^{0}[\mathds{1}\{F_0^{G\cup\{0\}}\}]\mathbb{E}^{y}[\mathds{1}\{F_y^{G\cup\{y\}}\}]m_{n, d}(dy), \notag
\end{align}
where the last equality is due to the fact that if $y \in \mathcal{S}_{d, n}\cup B(0, 2\log(n)^{1/d})^c$, then $\Vert y \Vert > 2(\log n)^{1/d}$; i.e., the events $F_0^{G\cup\{0\}}$ and $F_y^{G\cup\{y\}}$ are independent. By symmetry in the torus, we obtain
\begin{align}
    \mathbb{E}^{0}[\mathds{1}\{F_0^{G\cup\{0\}}\}] &= \mathbb{E}^{y}[\mathds{1}\{F_y^{G\cup\{y\}}\}]. \label{eq:torus_symmetry}
\end{align}

Denote $Y$ as the uniform random variable on $B(0, 2\log(n)^{1/d})$ and let $A$ be the event that there exists one other point in $B(0, 2\log(n)^{1/d})$ besides the origin. Then, 
\begin{align}
    \int_{y\in B(0, 2\log(n)^{1/d})}\mathbb{E}^{0, y}[F_0^{G\cup\{0, y\}}]m_{n, d}(dy)  &= \mathbb{P}(F_0^{G\cup\{0, Y\}}) \nonumber \\
    &= \mathbb{P}(F_0^{G\cup\{0\}} \mid A) \nonumber \\
    &\leq \frac{\mathbb{P}(F_0^{G\cup\{0\}})}{\mathbb{P}(A)} \nonumber \\
    &= \frac{\mathbb{E}^{0}[\mathds{1}\{F_0^{G\cup\{0\}}\}]}{1 - n^{-2^d\lambda\nu_d}},  \label{eq:torus_ub}
\end{align}
where the inequality is due to Bayes' rule. Combining  (\ref{eq:torus_decomposed}), (\ref{eq:torus_symmetry}), and (\ref{eq:torus_ub}) yields
\begin{align*}
    &\int_{y\in \mathcal{S}_{d, n}}\mathbb{E}^{0, y}[\mathds{1}\{F_0^{G\cup\{0, y\}}\}\mathds{1}\{F_y^{G\cup\{0, y\}}\}]m_{n, d}(dy) \\
    &\leq \frac{\mathbb{E}^{0}[\mathds{1}\{F_0^{G\cup\{0\}}\}]}{1 - n^{-2^d\lambda\nu_d}} + (n-2^d\nu_d\log(n))\mathbb{E}^{0}[\mathds{1}\{F_0^{G\cup\{0\}}\}]^2.
\end{align*}
By Lemma \ref{cond1_pf}, we have $n\mathbb{E}^{0}[\mathds{1}\{F_0^{G\cup\{0\}}\}]\geq n^\beta$ for some $\beta>0$ for $n$ sufficiently large, so dividing both sides by $n\mathbb{E}^{0}[\mathds{1}\{F_0^{G\cup\{0\}}\}]^2$ yields 
\begin{align*}
    & \frac{\int_{y\in \mathcal{S}_{d, n}}\mathbb{E}^{0, y}[\mathds{1}\{F_0^{G\cup\{0, y\}}\}\mathds{1}\{F_y^{G\cup\{0, y\}}\}]m_{n, d}(dy)}{n\mathbb{E}^0[\mathds{1}\{F_0^{G\cup\{0\}}\}]^2}  \\
    &\leq \frac{1}{n\mathbb{E}^{0}[\mathds{1}\{F_0^{G\cup\{0\}}\}](1 - n^{-2^d\lambda\nu_d})} + \frac{n-2^d\nu_d\log(n)}{n} \\
    &\leq 1 - \frac{2^d\nu_d\log(n)}{n} + \frac{1}{n^\beta(1 - n^{-2^d\lambda\nu_d})}.
\end{align*}
Hence Condition \eqref{cond2} is satisfied. 
\end{proof}

\subsection{Impossibility for $d = 1$ and $\lambda<1$}

We now prove part (b) of Theorem \ref{theorem:impossibility}. 
\begin{proof}[Proof of Theorem \ref{theorem:impossibility} (b)]
In the one-dimensional GHCM, the vertices are generated on $\mathcal{S}_{1, n}$, which is a circle. We partition the circle into blocks of length $\log n$ and construct the visibility graph $H^\dagger$ as follows. First, we create a vertex for every block in $\mathcal{S}_{1, n}$. Second, we add an edge between any two vertices if their corresponding blocks contain some points that are visible to each other in $\mathcal{S}_{1, n}.$ Define a segment as the set of blocks corresponding to a connected component of the visibility graph $H^{\dagger}$. If there is more than one segment, then any estimator must label the points in each segment independently from the other segments. Exact recovery becomes impossible because we can permute the labeling of each segment using any permissible $\omega \in \Omega_{\pi, P}$. Therefore, if there are $L$ segments, then there exist at least $|\Omega_{\pi, P}|^L - 1$ labelings which are different from $x^{\star}$ yet have the same posterior probability. It follows that, conditioned on the existence of at least $2$ segments, the condition $| \Omega_{\pi, P}| \geq 2$ implies the MAP estimator fails with probability at least $\frac{1}{3}$.

We next show that there are at least $2$ segments with high probability. Let $\mathcal{X}$ be the event that there exist two empty blocks $B_a$, $B_b$, and non-empty blocks $B_c$, $B_d$ with $a < c< b < d$ or $c < a < d < b$. In other words, $\mathcal{X}$ is the event that there are non-empty blocks that belong to different segments. We prove that $\cX$ occurs with high probability if $\lambda<1$, guaranteeing failure at exact recovery. Let $\cY_\ell$ be the event of having exactly $\ell$ empty blocks, among which at least two of them are non-adjacent. Since each block is independently empty with probability $\exp(-\lambda \log n) = n^{-\lambda}$, we have 
      \$
          \pr(\cX) &= \sum_{\ell=2}^{n/\log n-2} \pr(\cY_\ell) = \sum_{\ell=2}^{n/\log n - 2} \Big(\binom{{n}/{\log n}}{\ell} - {n}/{\log n}\Big) \big(n^{-\lambda}\big)^\ell\big(1 - n^{-\lambda}\big)^{n/\log n-\ell} \\
          &= \sum_{\ell=1}^{n/\log n - 1} \Big(\binom{{n}/{\log n}}{\ell} - {n}/{\log n}\Big) \big(n^{-\lambda}\big)^\ell\big(1 - n^{-\lambda}\big)^{n/\log n-\ell} \\
          &\ge \sum_{\ell=1}^{n/\log n} \Big(\binom{{n}/{\log n}}{\ell} - {n}/{\log n}\Big) \big(n^{-\lambda}\big)^\ell\big(1 - n^{-\lambda}\big)^{n/\log n-\ell} \\
          & = \sum_{\ell=1}^{n/\log n } \binom{{n}/{\log n}}{\ell} (n^{-\lambda})^\ell(1 - n^{-\lambda})^{n/\log n-\ell} \\
          &\qquad - \frac{n}{\log n}(1 - n^{-\lambda})^{n/\log n}\sum_{\ell=1}^{n/\log n} 
          \big[n^{-\lambda}/(1 - n^{-\lambda})\big]^{\ell} \\
          & \ge 1- \big(1 - n^{-\lambda}\big)^{n/\log n} - \big(1 - n^{-\lambda}\big)^{n/\log n} \cdot\frac{n}{\log n} \cdot\frac{n^{-\lambda}}{1-2n^{-\lambda}} \\
          & \ge 1 - \big(1 - n^{-\lambda}\big)^{n/\log n}\cdot\big(1 + 2n^{1-\lambda}/\log n\big) \\
          & = 1 - \big[(1 - n^{-\lambda})^{n^{\lambda}}\big]^{n^{1-\lambda}/\log n}\cdot\big(1 + 2n^{1-\lambda}/\log n\big) \\
          &= 1 - O\big(\exp(-{n^{1-\lambda}/\log n})\cdot(1 + 2n^{1-\lambda}/\log n)\big) =  1- o(1),
      \$
  where the first inequality holds since $\binom{n/\log n}{n/ \log n} - n \log n = 1 - n\log n < 0$, the second inequality follows since the first summation is a partial sum of the PMF of a binomial random variable and the second summation is a geometric series, and the last inequality holds since $1-2n^{-\lambda}\ge1/2$ for large enough $n$.
\end{proof}

\section{Connectivity of the Visibility Graph}
  \label{sec:connectivity}
  In this section, we establish the connectivity of the visibility graph $H = (V^{\dagger}, E^{\dagger})$ from Line \ref{line:visibility-graph} of Algorithm \ref{alg:almost-exact}. We begin by defining sufficiently small constants $\chi$ and $\delta$ used in Algorithm \ref{alg:almost-exact}. We define $\chi$ to satisfy the following condition, relying on $\lambda$ and $d$:
  \#\label{eq:chi-fomula}
  \nu_d \big(1 - 3\sqrt{d}\chi^{1/d}/2  \big)^d \ge (\nu_d + 1/{\lambda} )/2 \text{ and } 0<\chi <[(\mathds{1}_{d=1} + \nu_d\cdot\mathds{1}_{d\ge2})- 1/{\lambda} ]/2.
  \#
  The first condition is satisfiable since $\lim_{\chi \to 0} \nu_d \big(1 - 3\sqrt{d}\chi^{1/d}/2  \big)^d = \nu_d$
  and we have $\nu_d>(\nu_d + 1/{\lambda} )/2$ when $\lambda \nu_d>1$. The second one is also satisfiable since $\mathds{1}_{d=1} + \nu_d\cdot\mathds{1}_{d\ge2} = 1 > 1/\lambda$ if $d = 1$ and otherwise $\mathds{1}_{d=1} + \nu_d\cdot\mathds{1}_{d\ge2} = \nu_d > 1/\lambda$, under the conditions of Theorems \ref{theorem:exact-recovery} and \ref{theorem:almost-exact-recovery}.
  Associated with the choice of $\chi$, there is a constant $\delta'(\text{or }\widetilde\delta\text{ for }d\ge 2)>0$ such that for any block $B_i$, its visible blocks $\bigcup_{j\in V}\{V_j\colon B_j \sim B_i\}$ contain at least $\delta'\log n$ (or $\widetilde\delta\log n$) vertices with probability $1 - o(n^{-1})$. 
  With specific values of $\delta'$ and $\widetilde\delta$ to be determined in Propositions \ref{lem:visibility-d1-small-lambda} and \ref{lem:fN_i-size}, respectively, we take $\delta$ such that
  \#\label{eq:delta-formula}
  0<\delta<(\delta'\chi)\cdot \mathds{1}_{d=1} + [\widetilde{\delta}\chi/\nu_d]\cdot\mathds{1}_{d\ge2}.
  \#
  Propositions \ref{lem:visibility-d1-small-lambda} and \ref{lemma:connectivity} will present the connectivity properties of $\delta$-occupied blocks of volume $\chi \log n$, for $\chi$ and $\delta$ satisfying the conditions in \eqref{eq:chi-fomula} and \eqref{eq:delta-formula}, respectively. 

  We now record some preliminaries (see \cite{boucheron2013concentration}).
  \begin{lemma}[Chernoff bound, Poisson]\label{lem:Chernoff-poisson} Let $X\sim\text{Poisson}(\mu)$ with $\mu>0$. For any $t>0$,
  \$
  \pr(X\ge \mu + t) 
  \le \exp\Big(-\frac{t^2}{2(\mu+t)}\Big).
  \$
  For any $0<t<\mu$, we have 
  \$
  \pr(X\le \mu - t) \le \exp\Big(-(\mu - t) \log \Big(1 - \frac{t}{\mu} \Big)-t\Big).
  \$
  \end{lemma}
  
  \begin{lemma}[Hoeffding's inequality]
          Let $X_1,\cdots, X_n$ be independent bounded random variables with values $X_i\in[0,1]$ for all $1\le i\le n$. Let $X=\sum_{i=1}^n X_i$ and $\mu = \mathbb{E}[X]$. Then for any $t\ge0$, it holds that 
          \$
          \pr(X \ge \mu+ t) \le \exp(-2t^2/n), \quad
          \pr(X \le \mu -t) \le \exp(-2t^2/n).
          \$
  \end{lemma}
  
  \begin{lemma}[Chernoff upper tail bound]\label{lem:Chernoff-binomial} 
  Let $X_1,\cdots, X_n$ be independent Bernoulli random variables. Let $X=\sum_{i=1}^n X_i$ and $\mu = \E(X)$. Then for any $t>0$, we have
      \$
      \pr(X\ge (1+t)\mu) \le \left(\frac{e^t}{(1+t)^{(1+t)}} \right)^{\mu}.
      \$
  \end{lemma} 

  \begin{lemma}[Chernoff lower tail bound]\label{lem:Chernoff-binomial-lower} 
  Let $X \sim \text{Bin}(n, p)$ and let $\mu = \E(X)$. Then for any $0 < \delta < 1$, we have
      \$
      \pr(X\le (1-\delta)\mu) \le \exp \Big( \frac{-\delta^2 \mu}{2} \Big).
      \$
  \end{lemma} 
  
We also define a homogeneous Poisson point process used to generate locations as described in Definition \ref{def:ghcm}.
  \begin{definition}[\cite{kingman1992poisson}]\label{def:PPP}
    A homogeneous Poisson point process with intensity $\lambda$ on $S \subseteq \R^{d}$ is a random countable set $\Phi := \{v_1, v_2,\cdots \} \subset S$ such that
    \begin{enumerate}
        \item For any bounded Borel set $B\subset \R^d$, the count $N_\Phi(B) := |\Phi\cap B| = |\{i\in \N\colon v_i\in B\}|$ has a Poisson distribution with mean $\lambda\text{vol}(B)$, where $\text{vol}(B)$ is the measure (volume) of $B$.
        \item For any $k\in \N$ and any disjoint Borel sets $B_1,\cdots, B_k \subset \R^d$, the random variables $N_\Phi(B_1),$ $\cdots,$ $N_\Phi(B_k)$ are mutually independent.
    \end{enumerate}
\end{definition}
In the GHCM, the set of locations $V=\{v_1,v_2,\cdots\}$ are generated by a homogeneous Poisson point process with intensity $\lambda$ on $\cS_{n,d}$. The established properties guarantee that $|V|$ follows $\text{Poisson}(\lambda n)$. Moreover, conditioned on $|V|$, the locations $\{v_i\}_{i\in[|V|]}$ are independently and uniformly distributed in $\cS_{n,d}$. This gives a simple construction of a Poisson point process as follows:
\begin{enumerate}
    \item Sample $N_V \sim \text{Poisson}(\lambda n)$;
    \item Sample $v_1, \cdots, v_{N_V}$ independently and uniformly in the region $\cS_{n,d}$.
\end{enumerate}
This procedure ensures that the resulting set $\{v_1, \cdots, v_{N_V}\}$ constitutes a Poisson point process as desired.
  
  We now start the analysis. The following lemma shows that regions of appropriate volume have $\Omega(\log n)$ vertices with high probability.
  \begin{lemma}\label{lemma:volume-points}
      For any fixed subset $B\subset\cS_{d,n}$ with a volume $\nu\log n$ such that $\lambda \nu>1$, there exist constants $0 < \gamma < \lambda \nu$ and $\epsilon>0$ such that \$\pr(|V(B)| > \gamma \log n) \geq 1- n^{-1-\epsilon}.\$ 
  \end{lemma}
  \begin{proof}
      For a subset $B$ with $\text{vol}(B)=\nu\log n$, we have $|V(B)|\sim \text{Poisson}(\lambda \nu\log n)$. To show the lower bound, we define a function $g:(0, \lambda \nu]\to\R$ as $g(x) = x(\log x - \log (\lambda \nu)) +\lambda \nu-x$. 
      It is easy to check that $g$ is continuous and decreases on $(0, \lambda \nu]$ with $\lim_{x\to0}g(x)=\lambda \nu$ and $g(\lambda \nu)=0$. When $\lambda \nu>1$, it holds that $\lim_{x\to0}g(x)=\lambda \nu>(1+\lambda \nu)/2$ and thus there exists a constant $\gamma \in (0, \lambda \nu)$ such that $g(\gamma) > (1+\lambda \nu)/2$. Thus, the Chernoff bound in Lemma \ref{lem:Chernoff-poisson} yields that 
      \$
      \pr(|V(B)|\le \gamma \log n) & \le 
      \exp\Big(-[\gamma(\log \gamma - \log (\lambda \nu)) +\lambda \nu-\gamma]\log n\Big) \\
      & = n^{-g(\gamma)} \leq n^{-(1+\lambda \nu)/2}.
      \$
      Taking $\epsilon = (\lambda \nu - 1)/2 >0$ concludes the proof.
  \end{proof}
  
  \subsection{The Simple Case When \texorpdfstring{\(d=1\)}{} \text{and} \texorpdfstring{\(\lambda >1\)}{}}
  We start with the simple case when $d=1$.
  
  \vskip6pt
  \paragraph*{An example when $\lambda >2$} We first study an example when $d=1$ and $\lambda>2$. If $\lambda>2$ and $\vol(B_i) = \log n/2$, we have $\lambda\vol(B_i)/\log n>1$, and thus Lemma \ref{lemma:volume-points} ensures the existence of positive constants $\gamma$ and $\epsilon$ such that $\pr(|V_i| > \gamma \log n) \geq 1 - n^{-1-\epsilon}$ for all $i\in[2n/\log n]$. Thus, the union bound gives that 
  \#\label{eq:connected-H-lambda2}
\pr\Big(\bigcap_{i=1}^{2n/\log n}\big\{|V_i|> \gamma \log n\big\}\Big) &= 1 - \pr\Big(\bigcup_{i=1}^{2n/\log n}\big\{|V_i|\le \gamma \log n\big\}\Big) \notag\\
& \geq 1 - \frac{2n}{\log n}\cdot n^{-1-\epsilon} = 1-o(1).
  \#
  Since all blocks are $\gamma$-occupied, the $(1/2, \gamma)$-visibility graph $H = (V^{\dagger}, E^{\dagger})$ is trivially connected.

  \vskip6pt
  \paragraph*{General case when $\lambda >1$} For small density $\lambda$, we partition the interval into small blocks and establish the existence of visible occupied blocks on the left side of each block.
  \begin{proposition}\label{lem:visibility-d1-small-lambda}
  If $d=1$ and $\lambda >1$, with $0<\chi <(1-1/\lambda)/2$, we consider the blocks $\{B_i\}_{i=1}^{n/(\chi\log n)}$ obtained from Line \ref{line:partition} in Algorithm \ref{alg:almost-exact}. Then there exists a constant $\delta'>0$ such that for any $0<\delta<\delta'\chi$, it holds that 
      \$
      \pr\Big(\bigcap_{i=1}^{n/(\chi\log n)}\big\{\exists j \colon j < i, B_j \sim B_i, \text{ and } B_j \text{ is $\delta$-occupied}\big\}\Big) = 1 - o(1).
      \$
  It follows that the $(\chi, \delta)$-visibility graph is connected with high probability.
  \end{proposition}
  \begin{proof}
      For any $i\in[n/(\chi\log n)]$, we define $ U_i = \bigcup_{j\colon j < i, B_j \sim B_i}{B}_j$ as the union of visible blocks on the left-hand side of ${B}_i$. We have $\text{vol}( U_i)=(\floor{1/\chi}-1)\chi\log n\ge (1 -2\chi)\log n$ and $\lambda\text{vol}( U_i)/\log n\ge \lambda(1-2\chi)>1$ when $\lambda>1$ and $\chi<(1-1/\lambda)/2$. Thus, Lemma \ref{lemma:volume-points} ensures the existence of positive constants $\delta'$ and $\epsilon$ such that $\pr( |\bigcup_{j\colon j < i, B_i \sim B_j}V_j| \le \delta' \log n) \leq n^{-1-\epsilon}$. We note that $|\{j\colon j <i, B_j\sim B_i\}| \le (\ceil{1/\chi} - 1) \le 1/\chi$. Thus, we take $0<\delta < \delta'\chi$ and obtain that
      \$
  \pr\Big(\bigcap_{j\colon j< i, B_j\sim B_i}\big\{|V_j|\le \delta\log n \big\}\Big) & \le  \pr\Big(\big|\bigcup_{j\colon j< i, B_j\sim B_i}V_j\big|\le \delta\log n /\chi \Big) \\
  &\le \pr\Big(\big|\bigcup_{j\colon j < i, B_j\sim B_i}V_j\big|\le \delta'\log n \Big) \leq n^{-1-\epsilon}.
  \$
  Therefore, the union bound over all $i\in[n/(\chi\log n)]$ gives
\$
  &\pr\Big(\bigcap_{i=1}^{n/(\chi\log n)}\big\{\exists j \colon j < i, B_j \sim B_i, \text{ and } B_j \text{ is $\delta$-occupied}\big\}\Big) \\
  & \quad = 1 - \pr\Big(\bigcup_{i=1}^{n/(\chi\log n)}\bigcap_{j\colon j < i, B_j\sim B_i}\big\{|V_j|\le \delta\log n \big\}\Big) \\
  &\quad \ge 1 - \frac{n}{\chi\log n}\cdot n^{-1-\epsilon} = 1 - o(1).
\$
  \end{proof}
  
  \subsection{The General Case When \texorpdfstring{\(d\ge 2\)}{} \text{and} \texorpdfstring{\(\lambda \nu_d>1\)}{}}
  We first show that for any block $B$, the set of surrounding visible blocks $\{B'\colon B \sim B', B' \neq B\}$ contains $\Omega(\log n)$ vertices. Let $C_i$ be the ball of radius $R_d \triangleq (\log n)^{1/d}$ centered at the center of $B_i$. We define 
  \[U_i = \bigcup_{j\colon j\neq i, B_j\sim B_i} B_j 
  \]
  as the union of all visible blocks to $B_i$, excluding $B_i$ itself. Since any point in a visible block $B_j$ lies within a distance of $(\log n)^{1/d}$ from the center of $B_i$, we have $U_i\cup B_i \subset C_i$ (see Figure \ref{fig:geometry}). In addition, let $C_i''$ and $C_i'$ be the balls of radiuses $(1 - \sqrt{d}(\chi)^{1/d})(\log n)^{1/d}$ and $(1 - 3\sqrt{d}(\chi)^{1/d}/2)(\log n)^{1/d} \triangleq R_d'$ centered at the center of $B_i$, respectively. For any block $B_i\subset\cS_{d,n}$ with $\text{vol}(B_i)=\chi\log n$, the length of its longest diagonal is given by $\sqrt{d}(\chi\log n)^{1/d}$. Observe that
  \[\sup_{x \in B_i, y \in C_i''} \Vert x - y \Vert = \frac{1}{2}\sqrt{d}(\chi\log n)^{1/d} + (R_d - \sqrt{d}(\chi)^{1/d}) (\log n)^{1/d} = (\log n)^{1/d}.\]
  It follows that if $B_j \subseteq C_i''$, then $B_i \sim B_j$. 
  Shrinking the radius of $C_i''$ by the length of the longest diagonal of a block, namely $\sqrt{d}(\chi\log n)^{1/d}$, gives the smaller ball $C_i'$. Based on geometric observations, we note that $C_i^{\prime} \subset U_i\cup B_i$ (see Figure \ref{fig:geometry}). Observe that as $\chi \to 0$, the volume of the black region approaches the volume of $C_i$ and $C_i'$. The following lemma quantifies this observation, showing that our conditions on $\chi$ guarantee that $U_i$ (and any set with the same volume as $U_i$) will contain sufficiently many vertices.

  \begin{figure}[!ht]
\centering\includegraphics[width=0.9\linewidth]{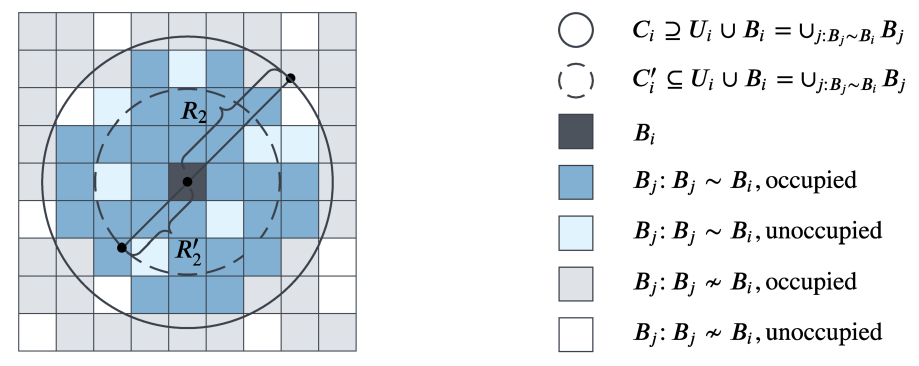}
      \caption{Geometry around block $B_i$, showing a portion of the region $\cS_{2,n}$. The set $U_i$ is comprised of dark and light black blocks.}
      \label{fig:geometry}
  \end{figure}

  \begin{lemma}\label{lem:fN_i-size}
  If $\chi$ satisfies the condition in \eqref{eq:chi-fomula} and $\lambda \nu_d>1$, there exist positive constants $\widetilde\delta$ and $\epsilon$, depending on $\lambda$ and $d$, such that for any subset $S\in\cS_{d,n}$ with $\text{vol}(S) = \text{vol}( U_i)$, we have
  \$
  \pr\big(|V(S)|>\widetilde\delta\log n\big) \geq 1- n^{-1-\epsilon}.
  \$ 
  \end{lemma}
  
  \begin{proof}
  We first evaluate the volume of $U_i\subset C_i$.\footnote{This is similar to the Gauss circle problem \citep{ivic2006lattice}.} 
  Note that $C_i^{\prime} \subset U_i\cup B_i
  $. The condition in \eqref{eq:chi-fomula} implies that $3\sqrt{d}\chi^{1/d}/2<1$ and thus $R_d'>0$. It follows that $\text{vol}( U_i\cup B_i) \ge \text{vol}(C_i^{\prime}) = \nu_d(R_d')^d$, and thus $\text{vol}( U_i) \ge \nu_d(R_d')^d-\chi\log n$.
  
  We now show that when $\lambda \nu_d>1$, the conditions in \eqref{eq:chi-fomula} imply $\lambda [\nu_d (1 - 3\sqrt{d}\chi^{1/d}/2)^d - \chi]>1$ by observing the following relations:
  \begin{align*}
  \nu_d \big(1 - 3\sqrt{d}\chi^{1/d}/2  \big)^d - \chi
  &\geq (\nu_d + 1/{\lambda} )/2 - \chi\\
  &\geq 1/\lambda.
  \end{align*}
  In summary, we have shown $\text{vol}(S) = \text{vol}( U_i)\ge \nu_d(R_d')^d-\chi\log n$ and $\lambda [\nu_d (R_d')^d/\log n - \chi]>1$. Thus, Lemma \ref{lemma:volume-points} ensures the existence of positive constants $\widetilde\delta$ and $\epsilon$ such that $\pr(|V(S)|>\widetilde\delta\log n) > 1- n^{-1-\epsilon}$.
  \end{proof}
  
  Henceforth, we use the term ``occupied block'' to refer to $\delta$-occupied blocks, as well as ``unoccupied block'', with the constant threshold $\delta=\delta(\lambda, d)$ defined in \eqref{eq:delta-formula} in the rest of the section. We define $K=|\{j\colon B_j\subset U_i\}|$ as the number of blocks in $U_i$, a constant relying on $\lambda$ and $d$. We note that $K\le\nu_d/\chi - 1 < \widetilde{\delta}/\delta$ since $U_i\cup B_i\subset C_i$. The key observation in establishing connectivity is that there cannot be a large \emph{cluster} of unoccupied blocks.
  \begin{definition}[Cluster of blocks]
      Two blocks are adjacent if they share an edge or a corner. We say that a set of blocks $\mathcal{B}$ is a cluster if for every $B, B' \in \mathcal{B}$, there is a path of blocks of the form $(B = B_{j_1}, B_{j_2}, \dots, B_{j_m} = B')$, where $B_{j_k} \in \mathcal{B}$ for $k \in [m]$ and $B_{j_k}, B_{j_{k+1}}$ are adjacent.
  \end{definition}
  
  The following lemma shows that all clusters of unoccupied blocks have fewer than $K$ blocks, with high probability. This also implies that $U_i$ contains at least one occupied block for each $i$.
  \begin{lemma}\label{lem:unoccupied-cluster}
  Suppose $d\ge2$ and $\lambda \nu_d >1$. Let
  $Y$ be the size of the largest cluster of unoccupied blocks produced in Line \ref{line:partition} in Algorithm \ref{alg:almost-exact}. Then $\pr(Y < K) = 1-o(1)$.
  \end{lemma}
  \begin{proof}
  We first bound the probability that all $K$ blocks in any given set are unoccupied. For any set of $K$ blocks $\{B_{j_k}\}_{k=1}^K$, we have 
  \#\label{eq:pr-all-unoccupied}
  \pr\Big(\bigcap_{k=1}^K\big\{|V_{j_k}|\le\delta\log n\big\}\Big) & \le \pr\Big(\big|\bigcup_{k=1}^K V_{j_k}\big|\le\delta K \log n\Big) \notag\\
  & \le \pr\Big(\big|\bigcup_{k=1}^K V_{j_k}\big|<\widetilde\delta \log n\Big) \notag\\
  & \le n^{-1-\epsilon},
  \#
  where the second inequality holds due to $K<\widetilde\delta/\delta$ and the last inequality follows from Proposition \ref{lem:fN_i-size} and the fact that $\vol(\bigcup_{k=1}^K B_{j_k})=\vol(U_i)$.
  
  Let $Z$ be the number of unoccupied block clusters with a size of $K$. Then we have $\pr(Y\ge K) = \pr(Z \ge 1)$. Let $\fS$ be the set of all possible shapes of clusters of blocks with a size of $K$. Clearly, $|\fS|$ is a constant depending on $K$ and $d$. For any $s\in\fS$, $i\in [n/(\chi\log n)]$, and $j\in[K]$, we define $\cZ_{s, i,j}$ as the event that there is a cluster of unoccupied blocks, characterized by shape $s$ with block $B_i$ occupying the $j$th position. Due to \eqref{eq:pr-all-unoccupied}, we have $\pr(\cZ_{s, i, j})\le n^{-1-\epsilon}$. Thus, the union bound gives
  \$
  \pr(Y\ge K) &= \pr(Z \ge 1) = \pr\Big(\bigcup_{s\in\fS, i\in[n/(\chi\log n)], j\in[K]}\cZ_{s, i,j}\Big) \\
  &\le |\fS| \cdot \frac{n}{\chi\log n}\cdot  K \cdot n^{-1-\epsilon} =  o(1). 
  \$
  \end{proof}
   
  Finally, we establish the connectivity of the visibility graph.
  \begin{proposition}\label{lemma:connectivity}
  Suppose that $d \ge 2$ and $\lambda \nu_d > 1$. Let $V \subset \cS_{d,n}$ be a Poisson point process on $\cS_{d,n}$ with intensity $\lambda$. Then for $\chi$ and $\delta$ given in \eqref{eq:chi-fomula} and \eqref{eq:delta-formula}, respectively, the $(\chi, \delta)$-visibility graph $H$ on $V$ is connected with probability $1-o(1)$.
  \end{proposition}
  \begin{proof}
      For a visibility graph $H = (V^{\dagger},E^{\dagger})$, we say that $S \subset V^{\dagger}$ is a \emph{connected component} if the subgraph of $H$ induced on $S$ is connected.
      Let $\cE$ be the event that $H$ contains an isolated connected component. Formally, $\cE$ is the event that there exists $S\subset V^{\dagger}$ \footnote{The notation $\subset$ denotes a strict subset.} such that (1) $S \neq \emptyset$ and $S\neq V^\dagger$; (2) $S$ is a connected component; (3) for all $i \in S, j \not \in S$ we have $\{i,j\} \not \in E^{\dagger}$. Observe that $\{H \text{ is disconnected}\} = \cE$.
  
      For any $S \neq \emptyset$ and $S\subset V^\dagger$ to be an isolated connected component, it must be completely surrounded by a cluster of unoccupied blocks. In other words, all blocks in the cluster $(\bigcup_{i\in S}U_i)\setminus(\bigcup_{i\in S}B_i)$ must be unoccupied. We 
   next show that for any isolated, connected component $S$, we have $|\{j\colon B_j\subset (\bigcup_{i\in S}U_i)\setminus(\bigcup_{i\in S}B_i)\}| \ge K$; that is, the number of unoccupied blocks visible to an isolated connected component is at least $K$.
  
      We prove the claim by induction on $|S|$. In fact, we prove it for $S$ that is isolated, but not necessarily connected.
      The claim holds true whenever $|S| = 1$ by the definition of $K$. Suppose that the claim holds for every isolated component with $k$ blocks. Consider an isolated component $S$, with $|S| = k +1$. Let $F = (\bigcup_{i \in S} B_i) \bigcup (\bigcup_{i \in S} U_i)$ be the collective ``footprint'' of all elements of $S$ along with the surrounding unoccupied blocks. For each $j \in S$, let $F_j = (\bigcup_{i \in S, i \neq j} B_i) \bigcup (\bigcup_{i \in S, i \neq j} U_i)$ be the footprint of all blocks in $S$ excluding $j$. Let $G_{j}$ be the graph formed from $G$ by removing all vertices from $V_{j}$, thus rendering $V_{j}$ unoccupied. Observe that there must exist some $j^{\star} \in S$ such that $F_{j^{\star}} \neq F$ and $F_{j^{\star}} \subset F$, as the regions $\{B_i \cup U_i\}_{i \in S}$ are translations of each other. Since $S \setminus \{j^{\star}\}$ is an isolated component in $G_{j^{\star}}$, the inductive hypothesis implies that $S \setminus \{j^{\star}\}$ has at least $K$ surrounding unoccupied blocks in $G_{j^{\star}}$. Comparing $G_{j^{\star}}$ to $G$, there are two cases (see Figure \ref{fig:isolated-component} for examples in $\cS_{2,n}$). \\
      \emph{Case I.} In the first case, $F \setminus F_{j^{\star}}$ contains at least one unoccupied block. In that case, the inclusion of $V_{j^{\star}}$ changes one block from unoccupied to occupied, and increases the number of surrounding unoccupied blocks by at least one. Thus, $S$ contains at least $K$ surrounding unoccupied blocks. \\
      \emph{Case II.} In the second case, $F \setminus F_{j^{\star}} $ contains only occupied blocks. Since there are $k+1$ total occupied blocks in $F$ and $k$ of them are in $F_{j^\star}$, we have $F \setminus F_{j^{\star}} = B_{j^{\star}}$, so that $B_{j^{\star}} \cap F_{j^{\star}} = \emptyset$. In this case, the set of $K$ surrounding unoccupied blocks in $F_{j^\star}$ remains unoccupied in $F$.
      \begin{figure}[!ht]
  \centering
  \includegraphics[width=.9\linewidth]{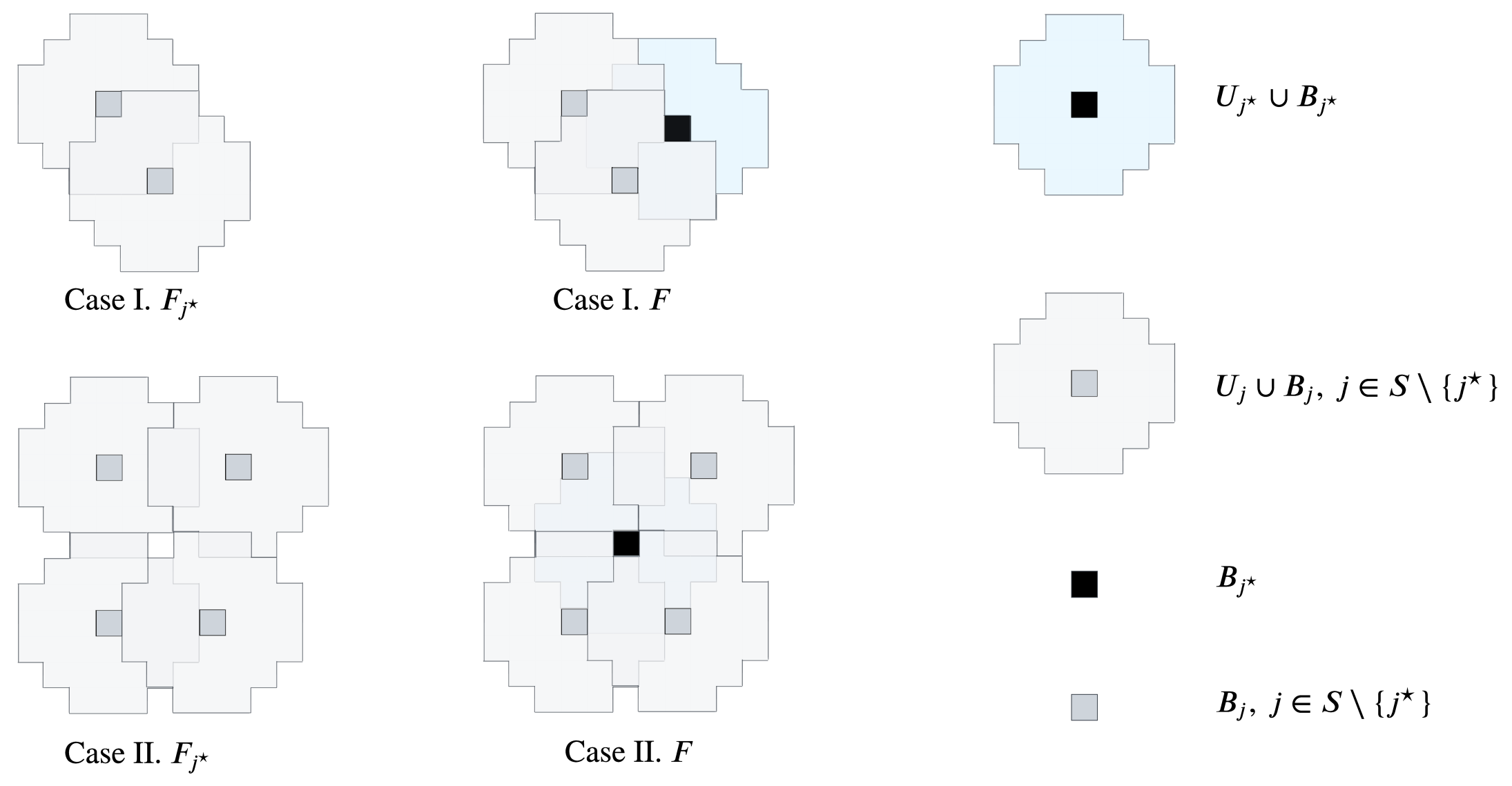}
    \caption{Possible isolated components in $\cS_{2,n}$ for Proposition \ref{lemma:connectivity}.}
    \label{fig:isolated-component}
  \end{figure}
  
      Thus, $\mathcal{E}$ implies $\{Y \geq K\}$. The result follows from Lemma \ref{lem:unoccupied-cluster}.
  \end{proof}

  In summary, Propositions \ref{lem:visibility-d1-small-lambda} and \ref{lemma:connectivity} establish the connectivity of visibility graphs for cases when $d=1$ and $\lambda>1$, or $d\ge 2$ and $\lambda \nu_d>1$, ensuring successful label propagation in the algorithm. For convenience, let $\cH = \{H \text{ is connected} \}.$ We conclude that $\pr(\cH) = 1- o(1)$.
\section{Phase I: Proof of Almost Exact Recovery}
\label{sec:almost_exact_recovery}

 In this section, we prove that Phase I of Algorithm \ref{alg:almost-exact} achieves almost exact recovery. Before proceeding further and showing the algorithm's success, we provide a high probability upper bound on the maximal occupancy of any block. 
  
  \begin{lemma}\label{lem:B-upper-bound}
  For the blocks obtained from Line \ref{line:partition} in Algorithm \ref{alg:almost-exact}, there exists a constant $\Delta>0$ such that 
  \$
  \pr\Big(\bigcap_{i=1}^{n/(\chi\log n)} \big\{|V_i|< \Delta\log n\big\}\Big) = 1- o(1).
  \$
  \end{lemma}
  \begin{proof}
  For a block ${B}_i$ with $\text{vol}({B}_i)=\chi\log n$, we have $|V_i|\sim \text{Poisson}(\lambda \chi\log n)$. 
  Thus, the Chernoff bound in Lemma \ref{lem:Chernoff-poisson} implies that, for $\Delta > (\lambda \chi + 1+\sqrt{2\lambda \chi+1})$, we have 
  \$ \pr(|{V}_i|\ge \Delta\log n) \le \exp\Big(-\frac{(\Delta-\lambda \chi)^2\log n}{2\Delta}\Big) = n^{-\frac{(\Delta-\lambda \chi)^2}{2\Delta}} < n^{-1}, 
  \$ 
  where the last inequality holds by straightforward calculation. Thus, the union bound gives that
  \$
\pr\Big(\bigcap_{i=1}^{n/(\chi\log n)} \big\{|V_i|< \Delta\log n\big\}\Big) & = 1 - \pr\Big(\bigcup_{i=1}^{n/(\chi\log n)} \big\{|V_i|\ge \Delta\log n\big\}\Big) \\
& > 1 - \frac{n}{\chi\log n}\cdot n^{-1} = 1- o(1). 
  \$
  \end{proof}
  
  For $\delta$ given by \eqref{eq:delta-formula} and $\Delta$ given by Lemma \ref{lem:B-upper-bound}, we define $\cI$ as follows and conclude that $\pr(\cI)=1-o(1)$. 
  \$\cI = \bigcap_{i\in V^\dagger} \{\delta\log n <|V_i|< \Delta\log n\}.\$

Recall the definition of $\phi_t(\cdot)$ in \eqref{eq:define-phi-t}. Note when $t=0$ or $t=1$, it holds that $\phi_t(p,q)=1$ for any $p,q$. When $0<t<1$, we have $\phi_t(p,q) = 1$ if and only if $p \stackrel{d}{=} q$. Throughout this section, for any $t\in(0,1)$, we define $\Phi_t$ as
\begin{align}
    \Phi_t = \max\{\phi_t(p_{ij}, p_{k \ell}) : i,j,k,\ell \in Z,p_{ij}\neq p_{k \ell}\} < 1.  \label{eqn:Phi_t}
\end{align}

\subsection{Labeling the Initial Block}
Algorithm \ref{alg:almost-exact} labels a subset of the first block $V_{i_0} \subseteq V_{i_1}$ by the \texttt{MAP} subroutine in Algorithm \ref{alg:initial-block}; namely,
\[
\widehat{x}_0 = \argmax_{x\colon V_{i_0}\to Z} \mathbb{P}(x^\star=x | G).
\]
In this section, we show that $\widehat{x}_0$ achieves exact recovery on $V_{i_0}$ with high probability. 
\begin{theorem}
\label{thm:initial_block}
Suppose that Assumption \ref{ass:bounded-ratio} holds. Let $G\sim\text{GHCM}(\lambda, n, \pi, P,d)$. Then
\$
\pr(\widehat x_0 \text{ achieves exact recovery on } V_{i_0}) = 1 - o(1).
\$
\end{theorem}

Rather than directly analyzing the MAP estimator, we analyze a different estimator - the \textit{restricted} maximum likelihood estimator (MLE). For any labeling $x\colon V_{i_0}\to Z$, we define the log-likelihood function as
\#\label{eq:def-ell-0} \ell_0(G, x) = \sum_{u\in V_{i_0}}\sum_{v\in V_{i_0}, v\neq u} \log \big(p_{x(u), x(v)}(y_{uv})\big). 
\#
The restricted MLE maximizes $\ell_0(G, x)$ over a set of labelings where the number of labels for each community is close to its mean. For any $\varepsilon>0$, define 
\begin{align}
X^\star_0(\varepsilon) = \big\{x\colon V_{i_0}\to Z \text{ such that } |u\in V_{i_0}\colon x(u)=i|\in \big((\pi_i-\varepsilon) |V_{i_0}|, (\pi_i + \varepsilon)|V_{i_0}| \big), \forall i\in Z \big\}. \label{eq:restricted_labels}
\end{align}
In particular, recall $\eta$ and $\rho$ defined in Assumption \ref{ass:bounded-ratio} and \eqref{eq:rho}, respectively. Let $\pi_{\min} = \min_{r\in Z} \pi_r$.
We choose positive constants $c$ and $\varepsilon$ as follows,
\begin{align}
c \le \frac{\varepsilon_0\pi_{\min}\log(1/ \rho)}{3\eta}, \quad    \varepsilon \le \frac{1}{3}\min\Big\{ \pi_{\min}, \min_{i,j\in Z, \pi_i \neq \pi_j}\{ |\pi_i - \pi_j| \}, \frac{c}{(k-1)\varepsilon_0 } \Big\}. \label{eq:epsilon}
\end{align}
The restricted MLE is then defined as
\[ \overline{x}_0 = \argmax_{x \in X^\star_0(\varepsilon)} \ell_0(G, x). \]
With a slight abuse of notation, we use $x^\star$ to denote the truncation of $x^\star$ on $V_{i_0}$ and condition on the event $\{x^\star\in X^\star_0(\varepsilon)\}$. The following lemma states that this event occurs with high probability. 
\begin{lemma}\label{eq:x^star-set} For any $0<\varepsilon<\pi_{\min}/2$, it holds that
$\pr(x^\star \in X^\star_0(\varepsilon)) =  1 - o(1)$.
\end{lemma}
\begin{proof}
    Fix $0<\varepsilon<\min_i\{\pi_i\}/2$. For each $i\in Z$, let $\xi_i\sim\text{Bin}(\varepsilon_0\log n, \pi_i)$. Hoeffding's inequality gives that
    \$
    \pr\Big(|u\in V_{i_0}\colon x^\star(u)=i|\notin \big((\pi_i-\varepsilon) |V_{i_0}|, (\pi_i + \varepsilon)|V_{i_0}|\big)\Big) &= \pr\big( |\xi_i-\pi_i \varepsilon_0\log n| \ge \varepsilon \varepsilon_0\log n \big) \\
    & \le 2\exp(-2\varepsilon^2 \varepsilon_0\log n) \\
    &= 2n^{-2\varepsilon^2\varepsilon_0}.
    \$
    Thus, a union bound over all $i\in Z$ yields that $\pr(x^\star \notin X^\star_0(\varepsilon)) \le 2kn^{-2\varepsilon^2\varepsilon_0} = o(1)$.
\end{proof}

We denote $\pr_{x^\star}(\cdot) = \pr(\cdot \given x^\star\in X_0^\star(\varepsilon))$. Due to the optimality of the MAP estimator, it is sufficient to show that the restricted MLE achieves exact recovery on $V_{i_0}$ with high probability. 
\begin{lemma}
\label{thm:initial_block_restricted_MLE}
Suppose that Assumption \ref{ass:bounded-ratio} holds. Let $G\sim\mathrm{GHCM}(\lambda, n, \pi, P,d)$. Then
\$
\pr(\overline x_0 \text{ achieves exact recovery on }V_{i_0}) = 1 - o(1).
\$
\end{lemma}
We establish Lemma \ref{thm:initial_block_restricted_MLE} by adapting the proof ideas from \cite[Section III.B]{Dhara2022}. Our approach is to show that with high probability, all labelings in $X_0^\star(\varepsilon)$ that do not achieve exact recovery satisfy $\ell_0(G,x) < \ell_0(G,x^\star)$. We divide labelings in $X_0^\star(\varepsilon)$ into two groups, according to their \emph{discrepancy} relative to the ground-truth labeling $x^{\star}$. 
\begin{definition}[Discrepancy] The \emph{Hamming distance} between two labelings $x, x^\prime \colon V_{i_0}\to Z$ is defined as
\$
d_H(x, x^\prime) = \sum_{u\in V_{i_0}} \mathds{1}\{x(u) \neq x^\prime(u)\}.
\$
The \emph{discrepancy} between two labelings $x, x^\prime \colon V_{i_0}\to Z$ is defined as 
\$
\text{DISC}(x, x^\prime) = \min_{\omega\in\Omega_{\pi, P}}\big\{d_H(\omega\circ x, x^\prime)\big\}.
\$
\end{definition}
We divide all labelings in the set $X^\star_0(\varepsilon)$ into low-discrepancy and high-discrepancy labelings. Our goal is to show, in either case, that their log-likelihood is strictly less than that of the ground truth.
\begin{proposition}[Low discrepancy]\label{prop:low-discrepancy}
Suppose that Assumption \ref{ass:bounded-ratio} holds. Then, with $\varepsilon$ defined in \eqref{eq:epsilon}, there exists $c\in(0, 1)$ satisfying the condition in \eqref{eq:epsilon} such that  with high
probability
\$
\forall x\colon V_{i_0}\to Z \text{ such that } 0<\textup{DISC}(x, x^\star) <c\log n, \text{ it holds that } \ell_0(G, x) < \ell_0(G, x^\star).
\$
\end{proposition}

\begin{proposition}[High discrepancy]\label{prop:high-discrepancy}
Fix any $c\in(0, 1)$. Then, with $\varepsilon$ defined in \eqref{eq:epsilon},  with high probability
\$
\forall x\in X^\star_0(\varepsilon) \text{ such that } \textup{DISC}(x, x^\star) \ge c\log n, \quad \ell_0(G, x) < \ell_0(G, x^\star).
\$
\end{proposition}

\begin{proof}[Proof of Lemma \ref{thm:initial_block_restricted_MLE}]
Fix $c$ and $\varepsilon$ from \eqref{eq:epsilon}. We consider a labeling from the set of labelings that maximize the MLE restricted to the set $X_0^\star(\varepsilon)$, 
\$
\overline x_0 \in \argmax_{x\in X_0^\star(\varepsilon)} \ell_0(G, x).
\$
Propositions \ref{prop:low-discrepancy} and \ref{prop:high-discrepancy} yield that
\$
\pr\Big(\bigcap_{\substack{x\in X_0^\star(\epsilon)\\ \text{DISC}(x, x^\star)\neq 0}} \ell_0(G, x) < \ell_0(G, x^\star)\Big) = 1 - o(1).
\$
Moreover, Lemma \ref{eq:x^star-set} gives $\pr(x^\star \in X_0^\star(\epsilon)) = 1 -o(1)$, which implies that $\overline{x}_0 \in \{\omega \circ x^{\star} : \omega \in \Omega_{\pi, P}\}$ with high probability.
\end{proof}

It remains to prove the low-discrepancy and high-discrepancy results in Propositions \ref{prop:low-discrepancy} and \ref{prop:high-discrepancy}.

\begin{proof}[Proof of Proposition \ref{prop:low-discrepancy}] Fix $x\colon V_{i_0}\to Z$ such that $\text{DISC}(x, x^\star)>0$. 
Let $\omega\in\Omega_{\pi, P}$ be the permissible relabeling such that $d_{H}(\omega\circ x, x^\star)= \text{DISC}(x, x^\star)$. Since $\ell_0(G, x) = \ell_0(G, \omega\circ x)$ for any $\omega\in\Omega_{\pi, P}$, we can without loss of generality assume that $\omega$ is the identity and $d_{H}(x, x^\star)=\text{DISC}(x, x^\star)$. Recall $\ell_0(G,x)$ defined in \eqref{eq:def-ell-0}. We decompose the targeted difference as follows, 
\#\label{eq:diff-phi}
& \ell_0(G, x) - \ell_0(G, x^\star)  = \sum_{u\in V_{i_0}}\sum_{\substack{v\in V_{i_0} \\ v\neq u }} \log \Big(\frac{p_{x(u), x(v)}(y_{uv})}{p_{x^\star(u), x^\star(v)}(y_{uv})}\Big) \\
& \quad = \sum_{u\in V_{i_0}}\sum_{\substack{v\in V_{i_0} \\ v\neq u }} \log \Big(\frac{p_{x(u), x^\star(v)}(y_{uv})}{p_{x^\star(u), x^\star(v)}(y_{uv})}\Big) + \sum_{u\in V_{i_0}}\sum_{\substack{v\in V_{i_0} \\ v\neq u }} \log \Big(\frac{p_{x(u), x(v)}(y_{uv})}{p_{x(u), x^\star(v)}(y_{uv})}\Big) \notag\\
& \quad = \sum_{\substack{u\in V_{i_0}\\ u\colon x(u)\neq x^\star(u)}}\sum_{\substack{v\in V_{i_0} \\ v\neq u }} \log \Big(\frac{p_{x(u), x^\star(v)}(y_{uv})}{p_{x^\star(u), x^\star(v)}(y_{uv})}\Big) + \sum_{\substack{v\in V_{i_0}\\ v\colon x(v)\neq x^\star(v)}}\sum_{\substack{u\in V_{i_0} \\ u\neq v }} \log \Big(\frac{p_{x(u), x(v)}(y_{uv})}{p_{x(u), x^\star(v)}(y_{uv})}\Big). \notag
\#
We define the two terms above as $A(x)$ and $B(x)$ and bound them in the sequel. We start from term $A(x)$ by considering the contribution of a fixed $u\in V_{i_0}$, where $x^\star(u)=a$, $x(u)=b$, and $a\neq b$. Then, with $\rho$ defined in \eqref{eq:rho}, we fix $0<c_1 \le {\varepsilon_0\pi_{\min}\log (1/ \rho)}/{3}$. The Chernoff bound gives that, for  $t_{ab}\in(0,1)$ defined in \eqref{eq:t_ij},
\$
& \pr_{x^\star}\bigg(\sum_{\substack{v\in V_{i_0} \\ v\neq u }} \log \Big(\frac{p_{b, x^\star(v)}(y_{uv})}{p_{a, x^\star(v)}(y_{uv})}\Big) \ge -c_1\log n\bigg) \\
&\qquad=\pr_{x^\star}\bigg(\exp\bigg(t_{ab}\sum_{\substack{v\in V_{i_0} \\ v\neq u }} \log \Big(\frac{p_{b, x^\star(v)}(y_{uv})}{p_{a, x^\star(v)}(y_{uv})}\Big)\bigg) \ge \exp(-t_{ab} c_1\log n)\bigg) \\
&\qquad\le\E\bigg[\exp\bigg(t_{ab}\sum_{\substack{v\in V_{i_0} \\ v\neq u }} \log \Big(\frac{p_{b, x^\star(v)}(y_{uv})}{p_{a, x^\star(v)}(y_{uv})}\Big)\bigg)\bigg]\cdot \exp(t_{ab} c_1\log n) \\
&\qquad =n^{t_{ab}c_1}\cdot\prod_{\substack{v\in V_{i_0} \\ v\neq u }} \phi_{t_{ab}}(p_{b, x^\star(v)}, p_{a, x^\star(v)}) \\
&\qquad \le n^{c_1}\cdot \phi_{t_{ab}}(p_{b, r_{ab}}, p_{a, r_{ab}})^{(\pi_{r_{ab}}-\varepsilon)\varepsilon_0\log n} \\
&\qquad \le n^{c_1}\cdot \rho^{(\pi_{\min}-\varepsilon)\varepsilon_0\log n} \\
&\qquad = n^{(\pi_{\min}-\varepsilon)\varepsilon_0\log\rho + c_1} \le n^{\varepsilon_0\pi_{\min}\log \rho /3}, 
\$
where $\phi_t$, $r_{ab}$, and $\rho$ are defined in \eqref{eq:define-phi-t}, \eqref{eq:r_ij}, and \eqref{eq:rho}, the second inequality follows from $t_{ab} \le 1$ and $x^\star\in X^\star_0(\varepsilon)$, and the last inequality holds since $\varepsilon\le \pi_{\min}/3$ and $c_1 \le -{\varepsilon_0\pi_{\min}\log \rho}/{3}$. 
We remark that $\log \rho<0$. The above holds for any $u\in V_{i_0}$ and any $b\neq x^\star(u)\in Z$ and is independent of $x$. As a result, the following observation ensures that term $A(x)$ in \eqref{eq:diff-phi} satisfies 
\$
& \pr_{x^\star}\Big(\bigcap_{x\colon V_{i_0}\to Z} \big\{ A(x) \le - d_H(x, x^\star)\cdot c_1\log n \big\}\Big) \\
&\qquad  \ge \pr_{x^\star}\bigg(\bigcap_{u\in V_{i_0}} \bigcap_{b\neq x^\star(u)\in Z} \Big\{\sum_{\substack{v\in V_{i_0} \\ v\neq u }} \log \Big(\frac{p_{b, x^\star(v)}(y_{uv})}{p_{x^\star(u), x^\star(v)}(y_{uv})}\Big) \le -c_1\log n \Big\}\bigg) \\
&\qquad  \ge 1 - k|V_{i_0}| \cdot n^{\varepsilon_0\pi_{\min}\log \rho /3} \ge 1 - k\varepsilon_0 \log n \cdot n^{\varepsilon_0\pi_{\min}\log \rho /3} \\
&\qquad  = 1 - o(1),
\$
where the second inequality is a union bound over all $u\in V_{i_0}$ and all $b \neq x^{\star}(u)\in Z$. Since the conditional event satisfies $\pr(x^\star \in X^\star_1(\varepsilon)) =  1 - o(1)$ from Lemma \ref{eq:x^star-set}, the law of total probability gives
\#\label{eq:diff-phi-I}
\pr\Big(\bigcap_{x\colon V_{i_0}\to Z} \big\{ A(x) \le - d_H(x, x^\star)\cdot c_1\log n \big\}\Big) = 1 - o(1).
\#
Next, we aim to bound the term $B(x)$ in \eqref{eq:diff-phi}. We decompose it as 
\$
B(x) 
& = \sum_{\substack{v\in V_{i_0}\\ v\colon x(v)\neq x^\star(v)}}\sum_{\substack{u\neq v\in V_{i_0} \\ u\colon x(u) = x^\star(u) }} \log \Big(\frac{p_{x^\star(u), x(v)}(y_{uv})}{p_{x^\star(u), x^\star(v)}(y_{uv})}\Big) \\
& \qquad + \sum_{\substack{v\in V_{i_0}\\ v\colon x(v)\neq x^\star(v)}}\sum_{\substack{u\neq v\in V_{i_0} \\ u\colon x(u)\neq x^\star(u) }} \log \Big(\frac{p_{x(u), x(v)}(y_{uv})}{p_{x(u), x^\star(v)}(y_{uv})}\Big) \\
& = \sum_{\substack{v\in V_{i_0}\\ v\colon x(v)\neq x^\star(v)}}\sum_{u\neq v\in V_{i_0}} \log \Big(\frac{p_{x^\star(u), x(v)}(y_{uv})}{p_{x^\star(u), x^\star(v)}(y_{uv})}\Big) \\
&\qquad+ \sum_{\substack{v\in V_{i_0}\\ v\colon x(v)\neq x^\star(v)}}\sum_{\substack{u\neq v\in V_{i_0} \\ u\colon x(u)\neq x^\star(u) }} \bigg[\log \Big(\frac{p_{x^\star(u), x^\star(v)}(y_{uv})}{p_{x^\star(u), x(v)}(y_{uv})}\Big) + \log \Big( \frac{p_{x(u), x(v)}(y_{uv})}{p_{x(u), x^\star(v)}(y_{uv})}\Big) \bigg].
\$
Let the two terms above be $B_1(x)$ and $B_2(x)$. Since $B_1(x) = A(x)$, it follows from \eqref{eq:diff-phi-I} that
\#\label{eq:diff-phi-IIa}
\pr\Big(\bigcap_{x\colon V_{i_0}\to Z} \big\{ B_1(x) \le - d_H(x, x^\star)\cdot c_1\log n \big\}\Big) = 1 - o(1).
\#
For the term $B_2(x)$, Assumption \ref{ass:bounded-ratio} gives $\log(p_{ij}(y)/p_{ab}(y)) <\eta$ for any $i,j,a,b$ and $y$; thus, we have
\#\label{eq:diff-phi-IIb}
B_2(x) \le 2\eta d_H(x, x^\star)^2.
\#
Combining \eqref{eq:diff-phi}, \eqref{eq:diff-phi-I}, \eqref{eq:diff-phi-IIa}, and \eqref{eq:diff-phi-IIb}, we have, with high probability, for any $x\colon V_{i_0}\to Z$, 
\$
\ell_0(G, x) - \ell_0(G, x^\star) & = A(x)+ B_1(x) +  B_2(x) \\
&\le - 2d_H(x, x^\star)\cdot c_1\log n + 2\eta d_H(x, x^\star)^2.
\$
Therefore, with high probability, for any $x$ such that $0<\text{DISC}(x, x^\star) <c\log n$ with $c\le c_1/\eta$, 
the above is negative and thus $\ell_0(G, x) - \ell_0(G, x^\star) < 0$ with high probability.
\end{proof}

\begin{proof}[Proof of Proposition \ref{prop:high-discrepancy}]
Fix any $c\in(0,1)$ and consider any $x\in X_0^\star(\varepsilon)$ such that $\textsc{DISC}(x, x^\star) \ge c\log n$. Let
\$
N_{ij} =\{u\in V_{i_0}:x^\star(u)=i, x(u)=j\}, \quad
        n_{ij} = |N_{ij}|, \quad \text{for any } i, j\in Z.
\$
Recall $\ell_0(G,x)$ defined in \eqref{eq:def-ell-0}. The error probability is bounded via a Chernoff bound, for any $t\in(0,1)$,  
\begin{align}
    \mathbb{P}( \ell_0(G, x) - \ell_0(G, x^\star) > 0 ) &\leq \mathbb{E}\big[ e^{t(\ell_0(G, x) - \ell_0(G, x^\star))} \big] \nonumber \\
    &= \prod_{u\in V_{i_0}}\prod_{\substack{v\in V_{i_0} \\ v\neq u}}\phi_t(p_{x(u), x(v)}, p_{x^\star(u), x^\star(v)}) \nonumber \\
    &= \prod_{u\in V_{i_0}}\prod_{\substack{v\in V_{i_0}, v\neq u \\ p_{x(u), x(v)} \neq p_{x^\star(u), x^\star(v)}}}
    \phi_t(p_{x(u), x(v)}, p_{x^\star(u), x^\star(v)}) \nonumber \\
    &= \prod_{\substack{i,j,k,l \in Z \\ p_{ik}\neq p_{jl}}} \big( \phi_t(p_{jl}, p_{ik}) \big)^{n_{ij}n_{kl}}. \label{eq:mle_diff}
\end{align}
Since for any $t\in(0,1)$, $\phi_t(p, q) < 1$ for any two distributions $p \neq q$, the approach is to bound a subset of the terms to provide a uniform bound. We consider the following two disjoint cases: 

\vskip4pt
\textbf{Case 1:} There exist $a, b, b'$ with $b \neq b'$ such that $n_{ab}, n_{ab'} \geq {\varepsilon \varepsilon_0 \log n}/{k}$. Take $z \in Z$ such that $p_{bz}\neq p_{b'z}$, which is guaranteed to exist (or otherwise the two classes $b$ and $b'$ are indistinguishable). Let $y\in Z$ be such that $n_{yz}\geq {|u: x(u)=z|}/{k}$, which is guaranteed to exist by the pigeonhole principle. Since $x \in X^\star_1(\varepsilon)$ and $\varepsilon \le \pi_{\min}/3$, we have
\[ \frac{|u: x(u)=z|}{k} \geq \frac{(\pi_z-\varepsilon)\varepsilon_0 \log n}{k} \geq \frac{\varepsilon \varepsilon_0 \log n}{ k}. \]
Since $p_{bz} \neq p_{b'z}$, we must have either $p_{ay}\neq p_{bz}$ or $p_{ay}\neq p_{b'z}$. Substituting into \eqref{eq:mle_diff}, it follows that
\begin{align}
         \mathbb{P}( \ell_0(G, x) - \ell_0(G, x^\star) > 0 ) &\leq \big(\min\{ \phi_t(p_{ay}, p_{bz}), \phi_t(p_{ay}, p_{b'z}) \}\big)^{\frac{\varepsilon^2 \varepsilon_0^2 \log^2 n}{k^2}} \label{eq:case1_bound}
\end{align}

Recall the definition of $\phi_t(\cdot)$ in \eqref{eq:define-phi-t} and $\Phi_t$ in \eqref{eqn:Phi_t}. We obtain
\begin{align}
    \min\{ \phi_t(p_{ay}, p_{bz}), \phi_t(p_{ay}, p_{b'z}) \} &\leq \Phi_t <1.
    \label{eq:case1_min_bound}
\end{align}
Combining \eqref{eq:case1_bound} and (\ref{eq:case1_min_bound}) implies there is a constant $C' \in (0,1)$ such that $\mathbb{P}( \ell_0(G, x) - \ell_0(G, x^\star) > 0 ) \leq C'^{\log^2 n}$ for all $x \in X^\star_1(\varepsilon)$ in Case 1.

\vskip4pt
\textbf{Case 2:} For any community $i$, there exists at most one $j$ such that $n_{ij}\geq {\varepsilon \varepsilon_0 \log n}/{k}$. This condition implies with high probability, for all community $i$, there exists exactly one community $j$ such that $n_{ij}\geq {\varepsilon \varepsilon_0 \log n}/{k}$.  Otherwise, since $x^\star \in X^\star_1(\varepsilon)$ with high probability, we obtain
\[  (\pi_i-\varepsilon)\varepsilon_0 \log n \leq |\{u: x^\star(u)=i\}| = \sum_{j \in Z}n_{ij} < \sum_{j \in Z}\frac{\varepsilon \varepsilon_0 \log n}{k} = \varepsilon \varepsilon_0 \log n.\]
Rearranging terms shows $\pi_i < 2\varepsilon$, which is a contradiction since $\varepsilon \leq \pi_i / 3$ by assumption. 

Let $\omega: Z \to Z$ be a mapping such that $n_{i\omega(i)}\geq {\varepsilon \varepsilon_0 \log n}/{k}$ for all $i$. We claim $\omega$ is a permutation. By way of contradiction, suppose there is some $j$ for which there is no $i$ satisfying $\omega(i)=j$, implying $n_{ij} < {\varepsilon \varepsilon_0 \log n}/{k}$ for all $i$. Then, since $x \in X^\star_1(\varepsilon)$, we obtain
\[ (\pi_j-\varepsilon)\varepsilon_0 \log n \leq  |\{u: x(u)=j\}| =  \sum_{i\in Z}n_{ij} < \sum_{i\in Z}\frac{\varepsilon \varepsilon_0 \log n}{k} = \varepsilon \varepsilon_0 \log n, \]
which again yields a contradiction.

Next, we claim that $\omega$ is not permissible with high probability. By way of contradiction, suppose that $\omega$ is permissible, which implies that its inverse, $\omega^{-1}$, is also permissible. Then we have
\begin{align*}
    \text{DISC}(x, x^\star) \leq d_H(\omega^{-1} \circ x, x^\star) &= \sum_{i}\sum_{j\colon j\neq \omega(i)} n_{ij} \\
    &\leq \sum_{i} \Big( (k-1)\frac{\varepsilon \varepsilon_0 \log n}{k} \Big) \\
    &= (k-1) \varepsilon \varepsilon_0 \log n \\
    &\leq c\log n.
\end{align*}
where the last inequality is due to $\varepsilon \leq c/[(k-1)\varepsilon_0]$ as defined in \eqref{eq:epsilon}. This contradicts with the high discrepancy of $x$; hence, $\omega$ is not permissible. 

We proceed by showing there must exist $(a, b)$ for which $p_{ab}\neq p_{\omega(a), \omega(b)}$. If $\pi_i=\pi_j$ for all $i, j\in Z$, then this condition clearly follows because $\omega$ is not permissible. Therefore, consider the case where $\{\pi_i\}$ are not all equal. In this case, suppose there exists $i$ such that $\pi_{\omega(i)} < \pi_i$.  Since $\sum_{r\in Z} n_{ir} \geq (\pi_i-\varepsilon)\varepsilon_0 \log n$ with high probability, and $n_{ir} < {\varepsilon \varepsilon_0 \log n}/{k}$ for any $r\neq \omega(i)$, we have
\begin{align}
    n_{i\omega(i)} > (\pi_i - \varepsilon) \varepsilon_0 \log n-(k-1)\frac{\varepsilon \varepsilon_0 \log n}{k}  > (\pi_i - 2 \varepsilon) \varepsilon_0 \log n. \label{eq:n_omega_lower_bound}
\end{align}
Then, combining \eqref{eq:n_omega_lower_bound} and \eqref{eq:restricted_labels} yields
\[ (\pi_i - 2\varepsilon)\varepsilon_0\log n < \sum_{r}n_{r\omega(i)} \leq (\pi_{\omega(i)}+\varepsilon)\varepsilon_0 \log n. \]
Rearranging terms yields $\varepsilon >  \frac{1}{3}  (\pi_i - \pi_{\omega(i)})$, which is a contradiction because $\varepsilon$ was defined to satisfy $\varepsilon \leq \frac{1}{3} |\pi_i - \pi_{\omega(i)}|$. Hence, for all $i \in Z$, we must have $\pi_{\omega(i)}\geq \pi_i$, which implies $\pi_{\omega(i)}=\pi_i$ because $\{\pi_i\}$ sums to $1$. Since $\omega$ is not permissible, there must exist $(a, b)$ for which $p_{ab}\neq p_{\omega(a), \omega(b)}$. Then we obtain for some constant $0<C''<1$, 
\begin{align*}
         \mathbb{P}( \ell_0(G, x) - \ell_0(G, x^\star) > 0 ) &\leq \big( \phi_t(p_{ab}, p_{\omega(a),\omega(b)}) \big)^{n_{a,\omega(a)}n_{b,\omega(b)}} \\
         &\leq \big( \phi_t(p_{ab}, p_{\omega(a),\omega(b)})\big)^{\varepsilon_0^2(\pi_a-2\varepsilon)(\pi_b-2\varepsilon)\log^2 n}\\
         &\leq \Phi_t^{\varepsilon_0^2(\pi_a-2\varepsilon)(\pi_b-2\varepsilon)\log^2 n}\\
         &\leq C''^{\log^2 n}.
\end{align*}
Note that the above bound is uniform for all $x$ satisfying Case 2. 

Combining both cases implies there exists a constant $0<C<1$ such that 
\[ \mathbb{P}( \ell_0(G, x) - \ell_0(G, x^\star) > 0)  \leq C^{\log^2 n}\]
for all $x$ such that $x\in X^\star_0(\varepsilon)$ and $\textsc{DISC}(x, x^\star) > c\log n$. Taking a union bound over all possible labelings $x$ yields 
\begin{align*}
    \mathbb{P}\Big[ \bigcup_{x\in X_0^\star: \textsc{DISC}(x, x^\star) > c\log n} \big\{\ell_0(G, x) - \ell_0(G, x^\star) \geq 0\big\} \Big] &\leq k^{\epsilon_0\log n}C^{\log^2 n} = o(1). 
\end{align*} 
\end{proof}

\subsection{Propagating Labels Among Occupied Blocks}
\label{subsec:propagation}
We now demonstrate that the \texttt{Propagate} subroutine (Algorithm \ref{alg:propagation}) ensures that all occupied blocks are classified with at most $M$ mistakes, for a suitable constant $M$.
 
We introduce a vector $m=(m_1, \cdots, m_{(n/(\chi\log n))})\in \mathbb{Z}_{+}^{(n/(\chi\log n))}$ and define the event \[\cV(m)=\big\{|V_i| = m_i \text{ for } i\in[n/(\chi\log n)]\big\}.\]  
Each $\cV(m)$ corresponds to a specific $(\chi,\delta)$-visibility graph $H$. Thus, conditioning on $\cV(m)$ and $H$ being connected determines the occupied block set $V^\dagger$ and the propagation ordering over $V^\dagger$. To simplify the analysis, we fix the vector $m$ in what follows, and condition on the event $\{\mathcal{V}(m) \subset \mathcal{H}\}$, recalling that $\mathcal{H} = \{H \text{ is connected}\}$. We write $\mathbb{P}_m(\cdot) = \mathbb{P}( \cdot \given \cV(m))$ as a reminder. Note that conditioned on $\mathcal{V}(m)$, the labels of vertices are independent, and the edges are independent conditioned on the vertex labels.
  
Let $\omega^\star\in \Omega_{\pi, P}$ be the permissible relabeling corresponding to $\widehat{x}_0$, that is, $\widehat{x}_0 = \omega^\star \circ x^\star$. We condition on the event that $\omega^\star$ exists, which occurs with high probability due to Theorem \ref{thm:initial_block}. Let vector $z = (z(j,+), z(j,-))_{j\in Z}\in \Z_{+}^{2k}$ be the \emph{configuration} of a block, where $z(j,+)$ is the number of vertices correctly labeled as $j$ and $z(j,-)$ is the number of vertices incorrectly labeled as $j$ by $\widehat{x}$. For $i\in V^\dagger$, the event $\cC_{i}(z)$ signifies that the occupied block $V_{i}$ possesses a configuration $z$ such that  
  \$
  &|\{u\in V_{i}, \widehat x(u) = j, \omega^\star \circ x^\star(u) =j\}| = z(j,+)\\
  &|\{u\in V_{i}, \widehat x(u) = j, \omega^\star \circ x^\star(u) \neq j\}| = z(j,-).
  \$
  Consider $i \in V^{\dagger} \setminus \{i_1\}$ and a configuration $z \in \mathbb{Z}_+^{2k}$. The key observation is that because the labels $\{\widehat x(u) \colon u \in V_i\}$ are determined using disjoint sets of edges, they are independent conditioned on $\cC_{p(i)}$.

  To ensure successful propagation, we assume that for any community $j$, any two distributions  $P(j,r)$ and $P(j,s)$ with $r\neq s$ are distinct, as stated in Assumption \ref{ass:distinguishable}. Fix any $t\in (0,1)$. For any $j$ and any distinct $r\neq s$, the distinctness assumption ensures that $\phi_t(p_{jr}, p_{js}) < \Phi_t < 1$, given $\phi_t$ and $\Phi_t$ defined in \eqref{eq:define-phi-t} and \eqref{eqn:Phi_t}. We define positive constants 
  \begin{align}
      c_2 = \delta\log (1/\Phi_t)/k,  \quad
      M = 5/(4c_2), \quad
      \eta_2 = k^2(e^\eta/\Phi_t)^M. \label{eqn:eta_2}
  \end{align}
  Let $\cA_{i}$ be the event that $\widehat x$ makes at most $M$ mistakes on $V_{i}$; that is: 
  \#\label{eq:define-A}
  \cA_{i} = \big\{ |\{u \in V_{i} : \widehat{x}(u) \neq \omega^\star \circ x^\star(u) \} | \leq M\big\}.
  \#
Let $\ell(u,x,S; j,r)$ be the log-likelihood used in Algorithm \ref{alg:propagation} at any $u\in S^\prime \subset V$, $S\subset V$, $x\colon S\to Z$, and $j, r\in Z$,
\#\label{eq:def-log-likelihood}
\ell(u,x,S; j,r) = \sum_{v\in S\setminus\{u\}, 
x(v)=j} \log p_{jr}(y_{uv}).
\#
We define the largest community labeled by $\widehat x$ in $V_{p(i)}$ as
\begin{align}
    J_i = \argmax_{a\in Z} |\{v\in V_{p(i)}\colon \widehat x(v)=a\}|. \label{eqn:J_i}
\end{align}
Then Algorithm \ref{alg:propagation} labels a vertex $u\in V_i$ by 
\$
\widehat x(u) = \argmax_{r\in Z}\ell(u,\widehat x, V_{p(i)}; J_i,r).
\$
The following lemma bounds the probability of misclassifying a vertex $u \in V_i$ using Algorithm \ref{alg:propagation}. The bound is uniform over all possible configurations $z$ restricted to $V_{p(i)}$ that is consistent with the total occupancy $m_{p(i)}$ and the maximal number of mistakes $M$.
\begin{lemma}\label{lem:degree-profile-error}
Suppose that Assumptions \ref{ass:bounded-ratio} and \ref{ass:distinguishable} hold. Fix $i\in V^\dagger\setminus\{i_1\}$. Fix $z \in \mathbb{Z}_+^{2k}$ such that $\sum_{j\in Z}(z(j,+) + z(j,-)) = m_{p(i)}$ and $\sum_{j\in Z}z(j,-) \leq M$ (so that $\cC_{p(i)}(z) \subset \cV(m) \cap \cA_{p(i)} $). Then for any $u\in V_i$, we have \$
\pr_m\big(\widehat x(u) \neq \omega^\star \circ x^\star(u)\biggiven \cC_{p(i)}(z)\big) \le \eta_2n^{-c_2}.
\$
\end{lemma}
    
\begin{proof} 
Fix $i\in V^\dagger\setminus\{i_1\}$. Recall that $J_i$ in \eqref{eqn:J_i} is the largest community labeled by $\widehat{x}$ in $V_{p(i)}$. We fix $j\in Z$ and first study the case when $J_i=j$. Then the pigeonhole principle yields $z(j,+) + z(j,-) \ge m_{p(i)}/k$. Moreover, the condition that $\sum_{j\in Z}z(j,-) \leq M$ ensures $z(j,-) \le M$ and thus $z(j, +) \ge m_{p(i)}/k - M \ge \delta\log n/k - M$. 
    
We fix any $u\in V_i$ and write $\mathbb{P}_{m, z, j, r}(\cdot) = \mathbb{P}( \cdot \given \cV(m), \cC_{p(i)}(z), J_i = j, \omega^\star \circ x^\star(u) = r)$. For any $t\in(0,1)$ and any label $s\neq r$, the definition of $\ell$ in \eqref{eq:def-log-likelihood} and the Chernoff bound give that
    \$
    & \pr_{m,z,j,r}\big(\ell(u, \widehat x, V_{p(i)}; j,s) \ge \ell(u, \widehat x, V_{p(i)}; j,r)\big) \\
    & \qquad = \pr_{m,z,j,r}\bigg(\sum_{\substack{v\in V_{p(i)}\\ v\colon \widehat x(v)=j}} \log \Big(\frac{p_{js} (y_{uv})}{p_{jr} (y_{uv})}\Big) \ge 0 \bigg) \\
    &\qquad \le \E_{m,z,j,r}\bigg[\exp\bigg(t \sum_{\substack{v\in V_{p(i)}\\ v\colon \widehat x(v)=j}} \log \Big(\frac{p_{js}(y_{uv})}{p_{j r}(y_{uv})}\Big)\bigg) \bigg] \\
    &\qquad = \prod_{\substack{v\in V_{p(i)}\\ v\colon \widehat x(v)=j}}  \E_{m,z,j,r}\bigg[\exp\bigg(t  \log \Big(\frac{p_{j s}(y_{uv})}{p_{j r}(y_{uv})}\Big)\bigg) \bigg] \\
    &\qquad = \prod_{\substack{v\in V_{p(i)}\\ \widehat x(v)=j \\  \omega^\star \circ x^\star(v) = j}}  \E_{m,z,j,r}\bigg[\exp\bigg(t  \log \Big(\frac{p_{j s}(y_{uv})}{p_{j r}(y_{uv})}\Big)\bigg)\bigg] \\
    &\qquad \qquad \cdot \prod_{\substack{v\in V_{p(i)}\\ \widehat x(v)=j \\  \omega^\star \circ x^\star(v) \neq j}}  \E_{m,z,j,r}\bigg[\exp\bigg(t  \log \Big(\frac{p_{j s}(y_{uv})}{p_{j r}(y_{uv})}\Big)\bigg) \bigg] \\
    & \qquad \le \prod_{\substack{v\in V_{p(i)}\\ \widehat x(v)=j \\  \omega^\star \circ x^\star(v) = j}} \phi_t(p_{jr}, p_{js}) \cdot \prod_{\substack{v\in V_{p(i)}\\ \widehat x(v)=j \\ \omega^\star \circ x^\star(v) \neq j}} \exp(\eta),
    \$
where the first product of the last inequality holds since  $\omega^\star$ is permissible, and thus conditioned on $\{\omega^\star \circ x^\star(u) = r\}$, we have $y_{uv}\sim P_{x^\star(u), x^\star(v)} = P_{r j} = P_{j r}$ for 
$v$ such that $\omega^\star \circ x^\star(v) = j$; and the second product follows from Assumption \ref{ass:bounded-ratio} that $\log(p_{ij}(y)/p_{ab}(y))<\eta$ for any $i,j,a,b$ and any $y$. Recall that the distinctness in Assumption \ref{ass:distinguishable} gives $P_{jr} \neq P_{js}$ and thus $\phi_t(p_{j r}, p_{j s}) < \Phi_t < 1$ for any $s\neq r$. Thus, the above gives
\$
& \pr_{m,z,j,r}\big(\ell(u, \widehat x, V_{p(i)}; j,s) \ge \ell(u, \widehat x, V_{p(i)}; j,r) \big) \\
&\qquad \le \prod_{\substack{v\in V_{p(i)}\\ \widehat x(v)=j \\  \omega^\star \circ x^\star(v) = j}} \phi_t(p_{jr}, p_{js}) \cdot \prod_{\substack{v\in V_{p(i)}\\ \widehat x(v)=j \\ \omega^\star \circ x^\star(v) \neq j}} \exp(\eta) \\
&\qquad \le \Phi_t^{z(j,+)} \cdot e^{\eta z(j,-)} \\
&\qquad \le \Phi_t^{\delta\log n/k - M} \cdot e^{\eta M} \\
&\qquad = (e^\eta/\Phi_t)^M \cdot n^{-\delta\log (1/\Phi_t)/k} = (\eta_2/k^2) \cdot n^{-c_2},
\$
where the third inequality follows from $z(j,-) \le M$ and $z(j, +) \ge m_{p(i)}/k - M \ge \delta\log n/k - M$.
Thus, conditioned on $\{\omega^{\star} \circ x^{\star}(u) = r\}$, a union bound over all $s\neq r$ gives that
\$
\pr_{m,z, j, r}\big(\widehat x(u) \neq \omega^\star \circ x^\star(u)\big) & = \pr_{m,z,j, r}\Big(\bigcup_{s\neq r \in Z} \big\{\ell(u, \widehat x, V_{p(i)}; j, s) \ge \ell(u, \widehat x, V_{p(i)}; j, r) \big\}\Big) \\
& \le \eta_2 n^{-c_2}.
\$
Since the above bound is uniform for all $j\in Z$ and $r\in Z$, the law of total probability completes the proof. 
\end{proof}

The following lemma concludes that Phase I makes few mistakes on occupied blocks during the propagation. Let $\eta_3 = (e\eta_2\Delta/M)^M$.
  
\begin{lemma}\label{prop:phase1-occupied} Let $G \sim \text{GHCM}(\lambda, n,\pi, P,d)$ and $\widehat{x} : V \to Z$ be the output of Phase I in Algorithm \ref{alg:almost-exact} on input $G$. 
  Suppose $m$ is such that $\cV(m)\subset\cI\cap\cH$. Lines \ref{line:step3}-\ref{line:step3-end} of Algorithm \ref{alg:almost-exact} ensure that 
   \$\pr_m\Big(\bigcap_{i\in V^\dagger}\cA_{i}\Big) \geq \big(1- o(1) \big) \Big(1 - \frac{\eta_3 n^{-\frac{1}{8}}}{\chi \log n} \Big).\$
  \end{lemma}
  \begin{proof}
  Consider $i_j \in V^{\dagger}$ for $2 \leq j \leq |V^{\dagger}|$, and fix $z \in \mathbb{Z}_+^{2k}$ such that 
  \begin{equation}
  \sum_{j\in Z}\big(z(j,+) + z(j,-)\big) = m_{p(i)} \text{ and } \sum_{j\in Z}z(j,-) \leq M. \label{eq:z-conditions}
  \end{equation}
  Observe that the events that $u \in V_{i_j}$ is mislabeled by $\widehat x$ are mutually independent conditioned on $\cC_{p(i_j)}(z)$. Lemma \ref{lem:degree-profile-error} shows that each individual vertex in $V_{i_j}$ is misclassified with probability at most $\eta_2n^{-c_2}$, conditioned on $\cC_{p(i_j)}(z)$. It follows that conditioned on $\cC_{p(i_j)}(z)$,
  \[|\{u \in V_{i_j} : \widehat{x}(u) \neq \omega^\star \circ x^\star(u) \}| \stleq \text{Bin}\left(\Delta\log n, \eta_2n^{-c_2} \right) =: \xi.\]
  Let $\mu_{\xi} = \mathbb{E}[\xi] = \eta_2\Delta n^{-c_2}\log n$.
  Using the Chernoff bound (Lemma \ref{lem:Chernoff-binomial}), we obtain
  \begin{align}
  \pr_m\big(\cA_{i_j}^c \biggiven \cC_{p(i_j)}(z) \big) &= \pr_m\big(|\{u \in V_{i_j} : \widehat{x}(u) \neq \omega^\star \circ x^\star(u)\} | > M \biggiven \cC_{p(i_j)}(z)\big) \nonumber \\
      & \le \pr(\xi > M)\nonumber \\
      & = \pr\big(\xi - \mu_\xi> (M/\mu_\xi- 1)\mu_\xi \big)\nonumber\\
  & \le e^{M-\mu_\xi}(\mu_\xi/M)^{M} \nonumber\\
  &\le (e\eta_2\Delta/M)^M(\log n)^M n^{-c_2M}\nonumber \\
  & \le \eta_3 n^{-9/8}.\label{eq:mistake-bound}    
  \end{align}
  The last inequality holds since $c_2M=5/4$ by definition and $(\log n)^M\le n^{1/8}$ for large enough $n$. Since $\cA_{i_j}$ is independent of $\{\cA_{i_k}: k < j, k \neq p(i_j)\}$ conditioned on $\cC_{p(i_j)}$, \eqref{eq:mistake-bound} implies 
  \begin{equation*}
  \pr_m\Big(\cA_{i_j}^c \biggiven \cC_{p(i_j)}(z), \bigcap_{k < j: i_k \neq p(i_j)}\cA_{i_k} \Big) \leq \eta_3 n^{-9/8}. 
  \end{equation*}
  Furthermore, since \eqref{eq:mistake-bound} is a uniform bound over all $z$ satisfying \eqref{eq:z-conditions}, it follows that
  \begin{equation*}
  \pr_m\Big(\cA_{i_j}^c \biggiven  \bigcap_{k < j}\cA_{i_k} \Big) \leq \eta_3 n^{-9/8}. 
  \end{equation*}
  
  Thus, combining Theorem \ref{thm:initial_block} with the preceding bound, we have 
  \$
  \pr_m\Big(\bigcap_{i\in V^\dagger}\cA_{i}\Big) &= \pr_m\big( \cA_{i_1}\big)\cdot\prod_{j=2}^{|V^{\dagger}|}\pr_m\big(\cA_{i_j} \biggiven \cA_{i_1}, \cdots, \cA_{i_{j-1}} \big) \\
  & \ge \big(1- o(1) \big) \Big(1 - \eta_3 n^{-\frac{9}{8}} \Big)^{|V^{\dagger}|-1}\\
  &\geq \big(1- o(1) \big) \Big(1 - \eta_3 n^{-\frac{9}{8}} \Big)^{\frac{n}{\chi \log n}}\\
  &\geq \big(1- o(1) \big) \Big(1 - \frac{\eta_3 n^{-\frac{1}{8}}}{\chi \log n} \Big), 
  \$
  where we use the fact that there are ${n}/{\chi \log n}$ blocks in total along with Bernoulli's inequality.
  \end{proof}

Combining the above results, we now prove the success of Phase I in Theorem \ref{thm:phase1-summary}. We highlight that since $\eta_1 > 0$ is arbitrary, the following equation \eqref{eq:almost-exact-recovery} implies Theorem \ref{theorem:almost-exact-recovery}.
\begin{theorem}\label{thm:phase1-summary}
  Let $G\sim\mathrm{GHCM}(\lambda, n, \pi, P, d)$ under Assumptions \ref{ass:bounded-ratio} and \ref{ass:distinguishable}. Fix any $\eta_1 > 0$. Let $\kappa = \nu_d(1+\sqrt{d}\chi^{1/d})^d/\chi$. Let $\widehat x$ be the labeling obtained from Phase I with $\chi>0$ satisfying \eqref{eq:chi-fomula} and $\delta>0$ satisfying \eqref{eq:delta-formula} and $\delta<\eta_1/\kappa$, respectively. 
  Then there exists a constant $M$ such that $\widehat x$ makes at most $M$ mistakes on every occupied block with high probability, i.e.,
  \begin{equation}
  \mathbb{P}\Big(\bigcap_{i \in V^{\dagger}} \big\{|\{v \in V_i\colon \widehat x(v) \neq \omega^\star \circ x^\star(v)\}| \leq M\big\} \Big) = 1-o(1). \label{eq:occupied-block-mistakes}
  \end{equation}
  Moreover, it follows that
  \begin{equation}
  \mathbb{P}\big(|\{v \in V \colon \widehat x(v) \neq \omega^\star \circ x^\star(v)\} | \leq \eta_1 n/(\chi\kappa) \big) = 1-o(1)   \label{eq:almost-exact-recovery} 
  \end{equation}
  and
  \begin{equation}
  \pr\Big(\bigcap_{u\in V}\big\{|v\in\cN(u)\colon \widehat x(v) \neq \omega^\star \circ x^\star(v)|\le \eta_1\log n\big\}\Big) = 1-o(1). \label{eq:dispersed-errors} 
  \end{equation}
  \end{theorem}
Observe \eqref{eq:almost-exact-recovery} bounds the total number of errors and \eqref{eq:dispersed-errors} bounds the number of errors in the visible radius around any vertex. 
  \begin{proof} Fixing any $\eta_1>0$, we consider $\chi>0$ satisfying \eqref{eq:chi-fomula} and $\delta>0$ satisfying \eqref{eq:delta-formula} and $\delta<\eta_1/\kappa$, respectively.
  Given any $m$ such that $\cV(m)\subset\cI\cap\cH$, for occupied blocks,
  Proposition \ref{prop:phase1-occupied} yields the existence of a constant $M>0$ such that \$\pr_m\Big(\bigcap_{i\in V^\dagger}\{ |\{v \in V_i \colon \widehat x(v) \neq \omega^\star \circ x^\star(v)\} | \leq M \}\Big) \geq \big(1- o(1) \big) \Big(1 - \frac{\eta_3 n^{-\frac{1}{8}}}{\chi \log n} \Big).\$ 
  Since the above bound is uniform over all $m$ such that $\cV(m) \subset \cI \cap \mathcal{H}$, we have 
  \$
  &\pr\Big(\bigcap_{i\in V^{\dagger}}\big\{|v\in V_i\colon \widehat x(v) \neq \omega^\star \circ x^\star(v)|\le M\big\}\Big) \\
  &\quad \ge \sum_{m\colon \cV(m) \subset \cI \cap \cH} \pr_m\Big(\bigcap_{i\in V^{\dagger}}\big\{|v\in V_i\colon \widehat x(v) \neq \omega^\star \circ x^\star(v)|\le M\big\}\Big) \cdot \pr\big(\cV(m)\big) \\
  &\quad \ge \big(1- o(1) \big) \Big(1 - \frac{\eta_3 n^{-\frac{1}{8}}}{\chi \log n} \Big)\cdot  \pr\big(\cI\cap\cH\big) = 1-o(1), 
  \$
  where the last step holds by Propositions \ref{lem:visibility-d1-small-lambda} and \ref{lemma:connectivity}, and Lemma \ref{lem:B-upper-bound}. Thus, we have proven \eqref{eq:occupied-block-mistakes}.
  
  Since $\delta \log n > M$ for $n$ large enough, it follows that
  \begin{equation}
  \pr\Big(\bigcap_{i \in [n/\chi\log n]}\{ |\{v \in V_i \colon \widehat x(v) \neq \omega^\star \circ x^\star(v)\} | \leq \delta \log n \}\Big) = 1-o(1).   \label{eq:block-mistakes} 
  \end{equation}

  Observe the intersection in \eqref{eq:block-mistakes} contains all blocks in $\mathcal{S}_{d, n}$, including the unoccupied blocks, because the unoccupied blocks contain at most $\delta \log n$ vertices. If $\widehat x$ makes fewer than $\delta\log n$ mistakes on $V_i$ for all $i\in[n/(\chi\log n)]$, then $\widehat x$ makes fewer than $\delta\log n \cdot (n/(\chi\log n)) = \delta n/\chi \le \eta_1 n/(\chi\kappa)$ mistakes in $\mathcal{S}_{d,n}$. Thus, \eqref{eq:almost-exact-recovery} follows from \eqref{eq:block-mistakes}. Moreover, if $\widehat x$ makes fewer than $\delta\log n$ mistakes on $V_i$ for all $i\in[n/(\chi\log n)]$, then there will be fewer than $\delta\kappa\log n\le \eta_1\log n$ mistakes in all vertices' neighborhood since each neighborhood $\cN(u)$ intersects at most $\kappa$ blocks. Thus, \eqref{eq:dispersed-errors} also follows from \eqref{eq:block-mistakes}.
  \end{proof}

\xn{
\subsubsection{Alternative Propagate Subroutine for Algorithm \ref{alg:propagation}} \label{sect:propagation-alt}
Finally, we remark that choosing the most frequent label in Algorithm \ref{alg:propagation} as the reference is primarily for convenience. An alternative approach is to label nodes relative to all other nodes in the seed region. For completeness, we also present this alternative approach in Algorithm \ref{alg:propagation-alt}, along with Lemma \ref{lem:degree-profile-error-alt}, which serves as an analogue of Lemma \ref{lem:degree-profile-error}. 
\begin{algorithm}
    \caption{\texttt{Propagate (Alternative)}} \label{alg:propagation-alt}
    \begin{algorithmic}[1]
    \Require{ Graph $G = (V,E)$, mutually visible vertex sets $S, S' \subset V$, $S \cap S' = \emptyset$, and $\widehat{x}_S\colon S\to Z$.}
    \Ensure{An estimated labeling $\widehat x_{S'} \colon S'\to Z$.} 
    \For{$u \in S'$}
        \$
        \widehat x_{S'}(u) = \argmax_{r\in Z}  \sum_{v\in S }
        \log p_{\widehat x_S(v),r}(y_{uv}).
        \$
    \EndFor
    \end{algorithmic}
\end{algorithm}

Let $\ell(u,x,S; r)$ be the log-likelihood used in Algorithm \ref{alg:propagation-alt} at any $u\in S^\prime \subset V$, $S\subset V$, $x\colon S\to Z$, and $r\in Z$; i.e.,
\#\label{eq:def-log-likelihood-alt}
\ell(u,x,S; r) = \sum_{v\in S\setminus\{u\}} \log p_{x(v),r}(y_{uv}).
\#
Then Algorithm \ref{alg:propagation-alt} labels a vertex $u\in V_i$ by 
\$
\widehat x(u) = \argmax_{r\in Z}\ell(u,\widehat x, V_{p(i)}; r).
\$

\begin{lemma}\label{lem:degree-profile-error-alt}
Suppose that Assumptions \ref{ass:bounded-ratio} and \ref{ass:distinguishable} hold. Fix $i\in V^\dagger\setminus\{i_1\}$. Fix $z \in \mathbb{Z}_+^{2k}$ such that $\sum_{j\in Z}(z(j,+) + z(j,-)) = m_{p(i)}$ and $\sum_{j\in Z}z(j,-) \leq M$ (so that $\cC_{p(i)}(z) \subset \cV(m) \cap \cA_{p(i)} $). For $\widehat x(\cdot)$ given by Algorithm \ref{alg:propagation-alt}, and any $u\in V_i$,  we have
\$
\pr_m\big(\widehat x(u) \neq \omega^\star \circ x^\star(u)\biggiven \cC_{p(i)}(z)\big) \le \eta_2n^{-kc_2}.
\$
\end{lemma}    
\begin{proof} 

Fix $i\in V^\dagger\setminus\{i_1\}$ and $u\in V_i$, and write $\mathbb{P}_{m, z, r}(\cdot) = \mathbb{P}( \cdot \given \cV(m), \cC_{p(i)}(z), \omega^\star \circ x^\star(u) = r)$. For any $t\in(0,1)$ and any label $s\neq r$, the definition of $\ell$ in \eqref{eq:def-log-likelihood-alt} and the Chernoff bound give that
    \$
    & \pr_{m,z,r}\big(\ell(u, \widehat x, V_{p(i)}; s) \ge \ell(u, \widehat x, V_{p(i)}; r)\big) \\
    & \qquad = \pr_{m,z,r}\bigg(\sum_{\substack{v\in V_{p(i)}}} \log \Big(\frac{p_{\widehat x(v),s} (y_{uv})}{p_{\widehat x(v),r} (y_{uv})}\Big) \ge 0 \bigg) \\
    &\qquad \le \E_{m,z,r}\bigg[\exp\bigg(t \sum_{\substack{v\in V_{p(i)}}} \log \Big(\frac{p_{\widehat x(v),s}(y_{uv})}{p_{\widehat x(v), r}(y_{uv})}\Big)\bigg) \bigg] \\
    &\qquad = \prod_{\substack{v\in V_{p(i)}}}  \E_{m,z,r}\bigg[\exp\bigg(t  \log \Big(\frac{p_{\widehat x(v), s}(y_{uv})}{p_{\widehat x(v), r}(y_{uv})}\Big)\bigg) \bigg] \\
    &\qquad = \prod_{j\in Z}\prod_{\substack{v\in V_{p(i)}\\ \widehat x(v)=j \\  \omega^\star \circ x^\star(v) = j}}  \E_{m,z,r}\bigg[\exp\bigg(t  \log \Big(\frac{p_{j s}(y_{uv})}{p_{j r}(y_{uv})}\Big)\bigg)\bigg] \\
    &\qquad \qquad \qquad \qquad \cdot \prod_{\substack{v\in V_{p(i)}\\ \widehat x(v)=j \\  \omega^\star \circ x^\star(v) \neq j}}  \E_{m,z,r}\bigg[\exp\bigg(t  \log \Big(\frac{p_{j s}(y_{uv})}{p_{j r}(y_{uv})}\Big)\bigg) \bigg] \\
    & \qquad \le \prod_{j\in Z}\prod_{\substack{v\in V_{p(i)}\\ \widehat x(v)=j \\  \omega^\star \circ x^\star(v) = j}} \phi_t(p_{jr}, p_{js}) \cdot \prod_{\substack{v\in V_{p(i)}\\ \widehat x(v)=j \\ \omega^\star \circ x^\star(v) \neq j}} \exp(\eta),
    \$
where the first product of the last inequality holds since  $\omega^\star$ is permissible, and thus conditioned on $\{\omega^\star \circ x^\star(u) = r\}$, we have $y_{uv}\sim P_{x^\star(u), x^\star(v)} = P_{r j} = P_{j r}$ for 
$v$ such that $\omega^\star \circ x^\star(v) = j$; and the second product follows from Assumption \ref{ass:bounded-ratio} that $\log(p_{ij}(y)/p_{ab}(y))<\eta$ for any $i,j,a,b$ and any $y$. Recall that the distinctness in Assumption \ref{ass:distinguishable} gives $P_{jr} \neq P_{js}$ and thus $\phi_t(p_{j r}, p_{j s}) < \Phi_t < 1$ for any $s\neq r$. Thus, the above gives
\$
& \pr_{m,z,r}\big(\ell(u, \widehat x, V_{p(i)}; s) \ge \ell(u, \widehat x, V_{p(i)}; r) \big) \\
&\qquad \le \prod_{j\in Z} \prod_{\substack{v\in V_{p(i)}\\ \widehat x(v)=j \\  \omega^\star \circ x^\star(v) = j}} \phi_t(p_{jr}, p_{js}) \cdot \prod_{\substack{v\in V_{p(i)}\\ \widehat x(v)=j \\ \omega^\star \circ x^\star(v) \neq j}} \exp(\eta) \\
&\qquad \le \Phi_t^{\sum_{j\in Z}z(j,+) } \cdot e^{\eta \sum_{j\in Z}z(j,-) } \\
&\qquad \le \Phi_t^{\delta\log n - M} \cdot e^{\eta M} \\
&\qquad = (e^\eta/\Phi_t)^M \cdot n^{-\delta\log (1/\Phi_t)} = (\eta_2/k^2) \cdot n^{-kc_2},
\$
where the third inequality follows from $\sum_{j\in Z}z(j,-) \le M$ and $\sum_{j\in Z}z(j, +) \ge m_{p(i)} - M \ge \delta\log n - M$.
Thus, conditioned on $\{\omega^{\star} \circ x^{\star}(u) = r\}$, a union bound over all $s\neq r$ gives that
\$
\pr_{m,z, r}\big(\widehat x(u) \neq \omega^\star \circ x^\star(u)\big) & = \pr_{m,z, r}\Big(\bigcup_{s\neq r \in Z} \big\{\ell(u, \widehat x, V_{p(i)}; s) \ge \ell(u, \widehat x, V_{p(i)}; r) \big\}\Big) \\
& \le \eta_2 n^{-kc_2}. 
\$
\end{proof}
}

\subsection{Almost Exact Recovery When $d = 1$ and $|\Omega_{\pi, P}| = 1$}

In the one-dimensional $\text{GHCM}$, we partition $\cS_{1, n}$ into blocks of length $\log(n)/2$ so that any pair of adjacent blocks are mutually visible. Recall, as in the impossibility result of Theorem \ref{theorem:impossibility} when $d=1, |\Omega_{\pi, P}|\geq 2$, that a connected component in the visibility graph corresponds to a contiguous segment of occupied blocks. In this section, we show that almost exact recovery is possible when $|\Omega_{\pi, P}| = 1$ for any value of $\lambda$. We refine the definition of 
segments to consist of contiguous blocks that are all $\delta$-occupied.
Observe that there may exist multiple segments in $\cS_{1, n}$. The number of connected components is upper bounded by the number of blocks which contain fewer than $\delta \log n$ vertices in $\cS_{1,n}$, which we show is $O(n^{1-\frac{\lambda}{4}})$

\begin{lemma} \label{lem:1d_N}
    When $d=1$, the number of blocks which contain fewer than $\delta \log n$ vertices is $O(n^{1-\frac{\lambda}{4}})$ with high probability. 
\end{lemma}
\begin{proof}
    Let $X\sim \mathrm{Pois}(\mu)$. A Chernoff bound yields, for $\kappa \in (0, 1)$,
    \[ \mathbb{P}\big(X \leq (1-\kappa)\mu \big) \leq \exp\big(-\mu[\kappa +(1-\kappa)\log(1-\kappa)]\big). \]
    Since $\mu = \lambda \log(n) / 2$, there exists $\kappa$ sufficiently close to $1$ (and thus $\delta$ sufficiently small) such that the RHS is $n^{-\lambda/4}$.

    Let $N$ be the number of blocks with fewer than $\delta \log n$ vertices in $\cS_{1,n}$. Observe that $N$ is stochastically upper-bounded by $\mathrm{Bin}(2n/\log(n), n^{-\lambda / 4})$. It follows that $N = O(n^{1-\lambda/4})$ with high probability by Markov's inequality. 
\end{proof}


To achieve almost exact recovery in the one-dimensional case, we apply our algorithm to each segment, i.e., conduct an initial labeling of a subset of the "leftmost" interval in each segment, then propagate.

\begin{algorithm}
      \caption{Almost exact recovery for the GHCM ($d=1$)} \label{alg:almost-exact-1d}
      \begin{algorithmic}[1]
      \Require{$G \sim \text{GHCM}(\lambda, n, \pi, P, 1)$.}
      \Ensure{An estimated community labeling $\widetilde{x}: V \to Z$.}
      \vspace{5pt}
      \State{{\bf Phase I:}} 
      \State{Partition the interval $[-n/2, n/2]$ into $2n/\log n$ blocks of volume $\log n/2$ each. Let $l$ be the number of segments, and $n_r$ be the number of blocks in the $r$-th segment. Let $B_i^r$ be the $i$-th block in the $r$-th segment, and $V_i^r$ be the set of vertices in $B_i^r$ for $i\in [n_r]$.} 
      \State Set $\varepsilon_0 \leq \min(\frac{\lambda}{8\log k}, \delta)$. 
      \For{$r=1,\cdots, l$} 
      \State Sample $V_{0}^r \subset V_{1}^r $ such that $|V_{0}^r| = \varepsilon_0 \log n$ 
      \State Set $V_{1}^r \leftarrow V_{1}^r \setminus V_{0}^r$ 
      \State Apply \texttt{Maximum a Posteriori} (Algorithm \ref{alg:initial-block}) on input $G$ and $V_{0}^r$ to obtain a labeling $\widehat{x}$ of $V_{0}^r$. \label{line:MAP-1d}
      \For{$i=1, \cdots, n_r$}
      \State Apply \texttt{Propagate} (Algorithm \ref{alg:propagation}) on input $G, V_{i-1}^r, V_i^r$ to determine the labeling $\widehat{x}$ on $V_i^r$. \label{line:propagate-1d}
      \EndFor
      \EndFor 
      \end{algorithmic}
  \end{algorithm}

\subsubsection{Initial Block Labeling}

Recall that when $d = 1$ and $|\Omega_{\pi, P}| \geq 2$, the value of $\lambda$ needed to be sufficiently high to ensure the connectivity of the block visibility graph; otherwise, there would be multiple segments and the MAP would not be able to select a consistent permissible relabeling across the components. On the other hand, when $|\Omega_{\pi, P}| = 1$, this ambiguity among labelings no longer exists, and thus we are able to consider smaller values of $\lambda$ that lead to disconnected visibility graphs. The presence of multiple segments necessitates the labeling of multiple initial blocks, potentially a diverging number. The existence of multiple initial blocks brings a new barrier to analysis, namely the ground truth $x^\star$ does not have guaranteed concentration around expected community sizes over all initial blocks. Thus, we cannot simply apply the result of Lemma \ref{thm:initial_block_restricted_MLE} using the restricted MLE. Instead, we add a stronger distinctness assumption (Assumption \ref{ass:distinguishable-strong}) that imposes distinctness of distributions for every $i, j\in Z$. With this assumption, we show that the MLE succeeds in labeling all but a vanishing fraction of vertices, in \emph{each} initial block, with high probability.

To ensure a linear runtime of $O(n \log n)$, Line \ref{line:MAP-1d} of Algorithm \ref{alg:almost-exact-1d} labels only $\varepsilon_0\log n$ points, where \[\varepsilon_0 = \min\left(\delta, \frac{\lambda}{8 \log k}\right).\] With high probability, there are at most $n^{1-\frac{\lambda}{4}}$ initial blocks to label  due to Lemma \eqref{lem:1d_N}. For each initial block, recall there are $k^{\varepsilon_0\log n}$ possible labelings, and evaluating the posterior probability of one labeling has runtime $O\left(\epsilon_0 \log n + \binom{\epsilon_0 \log n}{2} \right)$. Thus, for sufficiently large $n$, the runtime is 
\begin{align*}
    O\left(n^{1-\frac{\lambda}{4}} \log^2(n) \cdot k^{\epsilon_0 \log n} \right) &= O\left(n^{1-\frac{\lambda}{4}} \log^2(n) \cdot e^{\epsilon_0 \log(n) \log(k)}\right) \\
    &= O\left(n^{1-\frac{\lambda}{4}}\log^2(n) \cdot e^{\frac{\lambda}{8} \log(n) }\right) \\
    &= O\left(n^{1-\frac\lambda8}\log^2(n) \right) = O(n \log n). 
\end{align*}

Next, we show that on one initial block, the output of the MLE will make at most $c\log n$ mistakes, for $c>0$ arbitrarily small, with probability $1-O(n^{-\log n})$.

\begin{lemma}\label{lem:1d_initial_block}
    Let $d=1, |\Omega_{\pi, P}| = 1$, and any $c>0$. For any $x:V_{0}\to Z$ such that $\textup{DISC}(x, x^\star) = d_H(x, x^\star)> c\log n$,
    \[ \mathbb{P}( \ell_0(G, x) - \ell_0(G, x^\star) > 0 ) = n^{-C\log n} \]
    for some constant $C$ which is a function of $c.$
\end{lemma}
\begin{proof}
Fix $x$ such that $\textup{DISC}(x, x^\star) > c\log n$. Recall we defined 
\[ n_{ij} = |\{ u\in V_{0}: x^\star(u)=i, x(u)=j \}|. \] 
As in the high-discrepancy proof of Proposition \ref{prop:high-discrepancy}, we have
\begin{align}
    \mathbb{P}( \ell_1(G, x) - \ell_1(G, x^\star) > 0 ) &\leq \mathbb{E}\big[ e^{t(\ell_1(G, x) - \ell_1(G, x^\star)} \big] \nonumber \\
    &= \prod_{\substack{i,j,l,m \in Z \\ p_{il}\neq p_{jm}}} \big( \phi_t(p_{jm}, p_{il}) \big)^{n_{ij}n_{lm}}. \label{eq:1d_high_disc_bound}
\end{align}
It is sufficient to bound a subset of the terms in \eqref{eq:1d_high_disc_bound}. 
By the pigeonhole principle, there exists $i\neq j\in Z$ such that $n_{ij} \geq \frac{c\log n}{k(k-1)}$. Then, fix $l\neq \{i, j\}$ and $\varepsilon < 1$. Observe $\sum_{r\in Z} n_{lr}$, the number of points in the initial block whose ground truth label is $l$, is distributed as $\text{Bin}(\varepsilon_0\log n, \pi_l)$. Hence, Lemma \ref{lem:Chernoff-binomial-lower} yields $\sum_{r\in Z} n_{lr} \geq (\pi_l - \varepsilon)\varepsilon_0 \log n \geq (\pi_{\text{min}} - \varepsilon)\varepsilon_0 \log n$ with high probability. Then, the pigeonhole principle guarantees that there exists $m\in Z$ such that $n_{\ell m} \geq \frac{1}{k}(\pi_{\text{min}} - \epsilon) \epsilon_0 \log n$. The strong distinctness assumption (Assumption \ref{ass:distinguishable-strong}) implies $p_{il}\neq p_{jm}$. Therefore we obtain for some constant $C<1$
\begin{align}
    \mathbb{P}( \ell_1(G, x) - \ell_1(G, x^\star) > 0 ) &\leq \phi_t(p_{i\ell}, p_{jm})^{n_{ij}n_{l m}} \nonumber \\
    &\leq \Phi_t^{\frac{c\log n}{k(k-1)} \frac{(\pi_{\text{min}} - \epsilon) \epsilon_0 \log n}{k} } \nonumber \\
    &\leq C^{\log^2 n} \nonumber \\
    &= n^{-\Omega(\log n)}. \nonumber
\end{align}
Union bounding over all $k^{\Omega(\log n)}$ high-discrepancy labelings shows the error probability is $n^{-\Omega(\log n)}$.
\end{proof}

\begin{lemma} \label{lem:initial-block-1d}
    When $d=1, |\Omega_{\pi, P}| = 1$, for any $c>0$, labeling all initial blocks of connected components yields at most $c\log n$ mistakes per block with high probability. 
\end{lemma}
\begin{proof}
    The number of initial blocks is at most $o(n^{1-\frac{\lambda}{2}})$ by Lemma \ref{lem:1d_N}, with high probability. Since the probability that a single initial block contains more than $c\log n$ errors is $n^{-\Omega(\log n)}$ by Lemma \ref{lem:1d_initial_block}, it follows that the total error probability is
    \[ o(n^{1-\frac{\lambda}{2}})n^{-\Omega(\log n)} = o(1). \]
\end{proof}

\subsubsection{Propagation}
Line \ref{line:propagate-1d} of Algorithm \ref{alg:almost-exact-1d} applies the \texttt{Propagate} subroutine (Algorithm \ref{alg:propagation}) to each segment. Recall from Algorithm \ref{alg:almost-exact-1d} that $l$ is the number of segments and $n_r$ is the number of blocks in the $r$-th segment, for $r\in[l]$. Given the $r$-th segment, denote the vertices within each of its blocks as $V_i^r$, for $i\in[n_r]$. Adapting the notation from Section \ref{subsec:propagation}, we define the vector $m^r=(m^r_1, \cdots, m^r_{n_r})\in \mathbb{Z}_{+}^{n_r}$ and the event \[\cV(m^r)=\{|V^r_i| = m^r_i \text{ for } i\in[n_r]\}.\] 
We fix the vector $m^r$, and condition on the event $\mathcal{V}(m^r)$. We denote $\pr_m(\cdot) = \mathbb{P}(\cdot | m^1, \dots m^l)$. Recall the configuration of a block as a vector 
  $z = (z(j,+), z(j,-))_{j\in Z}\in \Z_{+}^{2k}$ to count the correct labels of the label $\widehat x$. For each $i\in[n_r]$, the event $\cC_i^r(z)$ implies the block $V_i^r$ is configured such that 
  \$
  &z(j,+) = |\{u\in V^r_{i}, \widehat x(u) = r, x^\star(u) =j\}|\\
  &z(j,-) = |\{u\in V^r_{i}, \widehat x(u) = r, x^\star(u) \neq j\}|.
  \$
Fix $c\in(0, 1)$ such that all initial blocks have at most $c\log n$ mistakes, as in Theorem \ref{lem:1d_initial_block}. Let $\cA_i^r$ be the event that $\widehat x$ makes at most $c\log n$ mistakes on $V^r_{i}$; that is, 
  \$
  \cA_{i}^r = \big\{ |\{u \in V^r_{i} : \widehat{x}(u) \neq  x^\star(u) \} | \leq c \log n\big\}.
  \$

Denote $p(i; r)$ as the parent block of $V_i^r$. We adapt Lemma \ref{lem:degree-profile-error} so that we tolerate at most $c\log n$ mistakes, as opposed to $M$ mistakes, when propagating. Recall from \eqref{eqn:eta_2} that $c_2 = \delta\log(1/\Phi_t) / k$. 

  \begin{corollary}\label{cor:degree-profile-error-1d}
    Let $d=1$ and $|\Omega_{\pi, P}| = 1$. Suppose that Assumptions \ref{ass:bounded-ratio} and \ref{ass:distinguishable} hold. Fix $i$ and $r$. Set 
    \begin{equation}
        c < \frac{c_2}{\log(e^\eta / \Phi_t)} < 1, \label{eq:c_1d}
    \end{equation}
    and fix $z \in \mathbb{Z}_+^{2k}$ such that $\sum_{j\in Z}(z(j,+) + z(j,-)) = m_{p(i; r)}$ and $\sum_{j\in Z}z(j,-) \leq c\log n$ (so that $\cC_{p(i; r)}(z) \subset \cV(m^r) \cap \cA_{p(i; r)}^r $). Then for any $u\in V_i^r$, there exists a constant $c_3>0$ such that 
    \$
    \pr_{m}\big(\widehat x(u) \neq x^\star(u)\biggiven \cC_{p(i; r)}(z)\big) \le (k-1)n^{-c_3}.
    \$
    \end{corollary}

\begin{proof}
    Fix an $r$-th segment and its $i$-th block. Let $J_i^r$ be the largest community labeled by $\widehat x$ in $V_{p(i;r)}^r$. For any label $s\neq x^\star(u)$, the similar reasoning in Lemma \ref{lem:degree-profile-error} yields
    \$
    & \pr_{m}\big(\ell(u, \widehat x, V_{p(i;r)}^r; J_i^r,s) \ge \ell(u, \widehat x, V_{p(i;r)}^r; J_i^r,x^\star(u)) \biggiven \cC_{p(i;r)}^r(z)\big) \\
    &\qquad \leq (e^\eta/\Phi_t)^{c\log n} \cdot n^{-\delta\log (1/\Phi_t)/k} \\
    &\qquad = (e^\eta/\Phi_t)^{c\log n} \cdot n^{-c_2} \\
    &\qquad =  n^{-c_2} \exp\left( c\log(\eta/\Phi_t)\log(n) \right) \\
    &\qquad = n^{c\log(\eta/\Phi_t) - c_2} \\
    &\qquad = n^{-c_3},
    \$
    where the last line is due to the choice of $c$ in \eqref{eq:c_1d}. Taking a union bound over all $s\neq x^\star(u)$ yields the desired result.     
\end{proof}

Finally, we show that the almost exact recovery algorithm (Algorithm \ref{alg:almost-exact-1d}) makes at most $c \log n$ errors on occupied blocks. We define 
\begin{equation}
    c_4 = \left(\frac{e\Delta(k-1)}{c}\right)^c, \quad c_5 = c\cdot c_3. \label{eq:c4_c5}
\end{equation}

\begin{corollary}\label{cor:phase1-occupied-1d} Let $d=1$ $|\Omega_{\pi, P}| = 1$. Let $G \sim \text{GHCM}(\lambda, n,\pi, P,d)$ and $\widehat{x} : V \to Z$ be the output of Algorithm \ref{alg:almost-exact-1d} on input $G$. Let $l\leq n^{1-\lambda/2}$. Then, under Assumptions \ref{ass:bounded-ratio} and \ref{ass:distinguishable-strong}, 
   \$\pr_{m}\Big(\bigcap_{r=1}^l \bigcap_{i\in n_r}\cA_{i}^r\Big) \geq \big(1- o(1) \big) \Big(1 - \frac{2c_4^{\log n}n^{1 - c_5\log n}}{\log n} \Big) \geq (1-o(1)) \left(1 - n^{-\frac{c_5}{2}\log n}\right).\$
  \end{corollary}

  \begin{proof}
     Consider $V_i^r$ and fix $z \in \mathbb{Z}_+^{2k}$ such that 
  \begin{equation}
  \sum_{j\in Z}\big(z(j,+) + z(j,-)\big) = m^r_{p(i; r)} \text{ and } \sum_{j\in Z}z(j,-) \leq c\log n. \label{eq:z_conditions_1d}
  \end{equation}
  The events that $u \in V_{i}^r$ is mislabeled by $\widehat x$ are mutually independent conditioned on $\cC_{p(i;r)}(z)$. Corollary \ref{cor:degree-profile-error-1d} shows that each individual vertex in $V_i^r$ is misclassified with probability at most $(k-1)n^{-c_3}$, conditioned on $\cC_{p(i; r)}(z)$. It follows that conditioned on $\cC_{p(i; r)}(z)$,
  \[|\{u \in V_{i}^r : \widehat{x}(u) \neq x^\star(u) \}| \stleq \text{Bin}\left(\Delta\log n, (k-1)n^{-c_3} \right) =: \xi.\]

  Let $\mu_{\xi} = \mathbb{E}[\xi] = (k-1)\Delta n^{-c_3}\log n$. Using the Chernoff bound (Lemma \ref{lem:Chernoff-binomial}), we obtain
  \begin{align}
  \pr_m\big(\cA_{i}^{rc} \biggiven \cC_{p(i;r)}(z) \big) &= \pr_m\big(|\{u \in V_{i}^r : \widehat{x}(u) \neq  x^\star(u)\} | > c\log n \biggiven \cC_{p(i;r)}(z)\big) \nonumber \\
      & \le \pr(\xi > c\log n)\nonumber \\
      & = \pr\big(\xi - \mu_\xi> (c\log n/\mu_\xi- 1)\mu_\xi \big)\nonumber\\
  & \le e^{c\log n -\mu_\xi}\left(\frac{\mu_\xi}{c\log n}\right)^{c\log n} \nonumber\\
  & = \left( \frac{e\mu_\xi}{c\log n} \right)^{c\log n} e^{-\mu_\xi} \nonumber\\
  & \leq \left( \frac{e\mu_\xi}{c\log n} \right)^{c\log n} \label{eqn:mu_xi} \\
  & = \left( \frac{e (k-1)\Delta n^{-c_3}\log n }{c\log n} \right)^{c\log n} \nonumber \\
  & =  \left( \frac{e (k-1)\Delta }{c} \right)^{c\log n} n^{-c_3\cdot c \log n} \nonumber \\
  &= c_4^{\log n} n^{-c_5\log n}. \label{eq:mistake-bound-1d}    
  \end{align}
 The inequality in \eqref{eqn:mu_xi} holds because $e^{-\mu_\xi} \leq 1$, and \eqref{eq:mistake-bound-1d} follows from the definitions in \eqref{eq:c4_c5}. Since $\cA_{i}^r$ is independent of $\{\cA_{k}^r: k < i, k \neq p(i; r)\}$ conditioned on $\cC_{p(i; r)}$, \eqref{eq:mistake-bound-1d} implies 
  \begin{equation*}
  \pr_m\Big(\cA_{i}^{rc} \biggiven \cC_{p(i; r)}(z), \bigcap_{k < i: k \neq p(i;r)}\cA_{k}^r \Big) \leq  c_4^{\log n} n^{-c_5\log n}. 
  \end{equation*}
  Furthermore, since \eqref{eq:mistake-bound-1d} is a uniform bound over all $z$ satisfying \eqref{eq:z_conditions_1d}, it follows that
  \begin{equation*}
  \pr_m\Big(\cA_{i}^{rc} \biggiven  \bigcap_{k < j}\cA_{k}^r \Big) \leq c_4^{\log n} n^{-c_5\log n}.
  \end{equation*}
Next, since $l \leq n^{1-\lambda/2}$, combining Theorem \ref{lem:initial-block-1d} with the preceding bound yields 
  \$
  \pr_{m}\Big(\bigcap_{r=1}^l \bigcap_{i\in n_r}\cA_{i}^r\Big) &= \prod_{r=1}^l \pr_{m}\Big(\bigcap_{i\in n_r}\cA_{i}^r\Big) \\
  &= \prod_{r=1}^l \pr_m\big( \cA_{0}^r\big)\cdot\prod_{i=1}^{n_r}\pr_m\big(\cA_{i}^r \biggiven \cA_{0}^r, \cdots, \cA_{i-1}^r \big)  \\
  &\geq \left( \prod_{r=1}^l \pr_m\big( \cA_{0}^r\big) \right) \left( 1 -  c_4^{\log n} n^{-c_5\log n} \right)^{\sum_r n_r} \\
  &\geq \big(1- o(1) \big) \left( 1 -  c_4^{\log n} n^{-c_5\log n} \right)^{\frac{2n}{\log n}} \\
  &\geq \big(1- o(1) \big) \Big(1 - \frac{2c_4^{\log n}n^{1 - c_5\log n}}{\log n} \Big), 
  \$
  where we use the fact that there are at most ${2n}/{\log n}$ blocks in total along with Bernoulli's inequality.
  \end{proof}

\begin{theorem}
  Let $d=1$ and $|\Omega_{\pi, P}| = 1$. Let $G\sim\text{GHCM}(\lambda, n, \pi, P, d)$ under Assumptions \ref{ass:bounded-ratio} and \ref{ass:distinguishable-strong}. Let $\widehat x$ be the labeling obtained Algorithm \ref{alg:almost-exact-1d} with 
  $c < c_2/{\log(e^\eta / \Phi_t)} < 1$
  and $\delta>0$ satisfying \eqref{eq:delta-formula} with $\chi=1/2$. 
  Then $\widehat x$ makes at most $c \log n$ mistakes on every occupied block with high probability, i.e.,
  \begin{equation}
  \mathbb{P}\Big(\bigcap_{r} \bigcap_{i \in n_r} \big\{|\{v \in V_i^r\colon \widehat x(v) \neq x^\star(v)\}| \leq c \log n\big\} \Big) = 1-o(1). \label{eq:occupied-block-mistakes-1d}
  \end{equation}
  \end{theorem}

  \begin{proof} 
  Fix $l \leq n^{1-\lambda/2}$. Define $\mathcal{I}_l$ as the event that GHCM consists of $l$ segments, which occurs with high probability due to Lemma \ref{lem:1d_N}. Let $m^r$ be the occupancies of the $r$-th connected component for $r \in [l]$, and condition on a sequence $m^1, m^2, \dots m^l$ which is consistent with $\mathcal{I}_l$. Corollary \ref{cor:phase1-occupied-1d} yields
  \$\pr_m\Big(\bigcap_{r} \bigcap_{i \in n_r}\{ |\{v \in V_i^r \colon \widehat x(v) \neq x^\star(v)\} | \leq c\log n \}\Big) \geq  (1-o(1)) \left(1 - n^{-\frac{c_5}{2}\log n}\right).\$ 
  Observe the above bound is uniform over all $m^r$ such that $\bigcap_r \cV(m^r) \subset \cI_l$ and there are at most $n / \log n$ possible segments. Therefore, we have 
  \begin{align*}
      &\pr\Big(\bigcap_{r} \bigcap_{i \in n_r} \big\{|v\in V_i\colon \widehat x(v) \neq x^\star(v)|\le c\log n\big\}\Big) \nonumber \\
  &\quad =  \sum_{l=1}^{n / \log n} \sum_{m\colon \bigcap_r \cV(m^r) \subset \cI_l} \pr_m\Big(\bigcap_{r} \bigcap_{i \in n_r} \big\{|v\in V_i\colon \widehat x(v) \neq x^\star(v)|\le c\log n\big\}\Big) \cdot \pr\big(\bigcap_r \cV(m^r)\big) \nonumber \\
  &\quad \ge  \sum_{l=1}^{n^{1-\lambda / 2}} \sum_{m\colon \bigcap_r \cV(m^r) \subset \cI_l} \pr_m\Big(\bigcap_{r} \bigcap_{i \in n_r} \big\{|v\in V_i\colon \widehat x(v) \neq x^\star(v)|\le c\log n\big\}\Big) \cdot \pr\big(\bigcap_r \cV(m^r)\big) \nonumber \\
  &\quad \ge  (1-o(1)) \left(1 - n^{-\frac{c_5}{2}\log n}\right) \cdot   \sum_{l=1}^{n^{1-\lambda/2}} \pr\big(\cI_l \big) \\
  &\quad = 1-o(1),
  \end{align*}
  where the last two lines are due to Corollary \ref{cor:phase1-occupied-1d} and $\sum_{l=1}^{n^{1-\lambda/2}} \pr\big(\cI_l \big)=1-o(1)$ via Lemma \ref{lem:1d_N}. Thus, we have proven \eqref{eq:occupied-block-mistakes-1d}. 
  \end{proof}

\section{Proof of Exact Recovery}
\label{sec:exact_recovery}

In this section, we prove that Algorithm \ref{alg:almost-exact} achieves exact recovery under the assumptions of Theorem \ref{theorem:exact-recovery}. Recall that $\Lambda(t) = \log \E[\exp(tX)]$ is the CGF of a random variable $X$. The following lemma characterizes the sum of a random number of random variables. 
\begin{lemma}\label{lem:cumulant-generating}
  Let $Y$ be a random variable, $X_1, \cdots, X_Y$ be a number of $Y$ i.i.d. copies of a random variable $X$, independent of $Y$,  and $S = \sum_{i=1}^Y X_i$. Then it holds that $\Lambda_S(t) = \Lambda_Y(\Lambda_X(t))$.
\end{lemma}
We first show the following lemma.
\begin{lemma}\label{lem:exact-bnd}
    If $\lambda\nu_d \min_{i\neq j} D_+(\theta_i,\theta_j)>1$, then for a fixed $0<\epsilon \le (\lambda \nu_d \min_{i\neq j} D_+(\theta_i, \theta_j) - 1)/2$, given any $u$ with $x^\star(u)=i$ and $j\neq i$, it holds that
    \$
    \pr\big(\ell_i(u, x^\star)-\ell_j(u, x^\star) \le \epsilon\log n \given x^\star(u)=i \big) = n^{-(1 + \Omega(1))}.
    \$
\end{lemma}
\begin{proof}
    Fix any $u$ with $x^\star(u)=i$. We consider the random variable $S$ defined as follows,
    \$
    S & := \ell_j(u, x^\star)-\ell_i(u, x^\star) = \sum_{v\in V\setminus\{u\}, v\sim u} \log \Big(\frac{p_{j,x^\star(v)} (Y_{uv})}{p_{i,x^\star(v)} (Y_{uv})}\Big) \\
    & = \sum_{r\in Z}\Big[\sum_{v\in V\setminus\{u\}, v\sim u, x^\star(v)=r} \log \Big(\frac{p_{jr} (Y_{uv})}{p_{ir} (Y_{uv})}\Big)\Big],
    \$
    where $Y_{uv}\sim P_{ir}$ for any $v\in V\setminus\{u\}$ such that $v\sim u$ and $x^\star(v)=r$. Thus, the Chernoff bound gives that 
    \#\label{eq:chernoff-S}
    \pr\big(\ell_i(u, x^\star)-\ell_j(u, x^\star) \le \epsilon\log n \given x^\star(u)=i \big) &= \pr(S\ge -\epsilon \log n) \notag\\
    & \le \inf_{t>0}\E[\exp(tS)]\cdot \exp(t\epsilon \log n).
    \#
    We now compute the MGF $\E[\exp(tS)]$.  
    We define $\zeta_r \sim \text{Poisson}(\lambda \nu_d \pi_r \log n )$ for $r\in Z$, $\xi_{r,s}\sim P_{ir}$ for $r\in Z$ and $s\in[\zeta_r]$, and 
    \$
    \widehat S:= \sum_{r\in Z}\Big[\sum_{s=1}^{\zeta_r} \log \Big(\frac{p_{jr} (\xi_{r,s})}{p_{ir} (\xi_{r,s})}\Big)\Big].
    \$
    Since $|\{v\colon v\in V\setminus\{u\}, v\sim u, x^\star(v)=r\}|\sim \text{Poisson}(\lambda \nu_d \pi_r \log n )$ in the definition of $S$, $S$ and $\widehat S$ share identical distributions.   Let $\Lambda_r(t)$ be the CGF of the random variable $\sum_{s=1}^{\zeta_r} \log ({p_{jr} (\xi_{r,s})}/{p_{ir} (\xi_{r,s})})$. Then for any $t>0$, we have
    \#\label{eq:mgf-S}
    \E[\exp(tS)] &= \E[\exp(t\widehat S)] = \prod_{r\in Z} \exp(\Lambda_r(t)).
    \#
    We now compute $\Lambda_r(t)$. For $\zeta_r \sim \text{Poisson}(\lambda \nu_d \pi_r \log n )$, we recall that $\Lambda_{\zeta_r}(t) = \lambda \nu_d \pi_r \log n (e^t - 1)$. Moreover, we define $\tau_r := \log (p_{jr} (\xi)/p_{ir} (\xi))$ with $\xi\sim P_{ir}$ and recall the definition in \eqref{eq:define-phi-t} that 
    \$
    \E[\exp(t\tau_r)]  = \phi_t(p_{jr}, p_{ir}).
    \$ 
    Thus, by Lemma \ref{lem:cumulant-generating}, we have 
    \$
    \Lambda_r(t) &= \Lambda_{\zeta_r}(\Lambda_{\tau_r}(t)) \\
    & = \lambda \nu_d \pi_r \log n \cdot \big[\exp(\log \E[\exp(t\tau_r)])- 1\big] \\
    & =  \lambda \nu_d \pi_r \log n \cdot \big[\E[\exp(t\tau_r)]- 1\big] \\
    & = -\lambda \nu_d \pi_r\big[1-\phi_t(p_{jr}, p_{ir})\big] \cdot\log n.
    \$
Substituting the above equation into \eqref{eq:mgf-S}, we obtain the MGF of $S$ as follows,
    \$
    \E[\exp(tS)] = \prod_{r\in Z} \exp(\Lambda_r(t)) = n^{-\lambda \nu_d [1-\sum_{r\in Z}\pi_r\phi_t(p_{jr}, p_{ir})]}.
    \$
    By substituting the above into \eqref{eq:chernoff-S}, we have for $t \in (0,1]$,
    \$
    \pr\big(\ell_i(u, x^\star)-\ell_j(u, x^\star) \le \epsilon\log n \given x^\star(u)=i \big)
    & \le \inf_{t>0}\E[\exp(tS)]\cdot \exp(t\epsilon \log n) \\
    &  \le \inf_{t\in(0,1]}  n^{-\lambda \nu_d [1-\sum_{r\in Z}\pi_r\phi_t(p_{jr}, p_{ir})] + t\epsilon} \\
    &  =  n^{-\lambda \nu_d [1-\inf_{t\in(0,1]}\sum_{r\in Z}\pi_r\phi_t(p_{jr}, p_{ir})] + \epsilon} \\
    & \le  n^{-\lambda \nu_d D_+(\theta_i, \theta_j) + \epsilon} = n^{-(1 + \Omega(1))}, 
    \$
where the second equality follows from \eqref{eq:CH-phi} and the last one holds when taking $\epsilon = (\lambda \nu_d \min_{i\neq j} D_+(\theta_i, \theta_j) - 1)/2$.
\end{proof}

\begin{theorem}\label{thm:exact-reovery-result}
    Under Assumptions \ref{ass:bounded-ratio} and \ref{ass:distinguishable}, if $\lambda\nu_d \min_{i\neq j} D_+(\theta_i, \theta_j)>1$, then $\widetilde x$, the output of Algorithm \ref{alg:almost-exact}, achieves exact recovery.
\end{theorem}
\begin{proof}
We first fix a constant $c>\lambda$ and let $\cE_0= \{|V| < c n\}$. Since $|V|\sim\text{Poisson}(\lambda n)$, the Chernoff bound in Lemma \ref{lem:Chernoff-poisson} gives that
\$
  \pr(\cE_0^c) = \pr(|V|> c n) \le \exp\big(-\frac{(c-\lambda)^2n}{2c}\big) = o(1).
\$
For $\beta > 0$ to be determined, let $\cE_1$ be the event that $\widehat x$ makes at most $\beta \log n$ mistakes in the neighborhood for all vertices (Phase I succeeds); that is,
\$
\cE_1=\bigcap_{u\in V}\big\{|v\in\cN(u)\colon \widehat x(v)\neq \omega^\star \circ x^\star(v)|\le \beta\log n\big\}.
\$
Theorem \ref{thm:phase1-summary} ensures that $\pr(\cE_1)=1-o(1)$. Let $\cE_2'$ be the event that Algorithm \ref{alg:almost-exact} achieves exact recovery and $\cE_2$ be the event that all vertices are labeled correctly relative to $\omega^\star$; that is,
\[
  \mathcal{E}_2' = \bigcup_{\omega \in \Omega_{\pi, P}} \bigcap_{u \in V}\{\widetilde x(u) = \omega \circ x^\star(u)\},\quad\mathcal{E}_2 = \bigcap_{u \in V}\{\widetilde x(u) = \omega^\star \circ x^\star(u)\}.
\]
Then we have $\pr(\cE_2')\ge\pr(\cE_2)$. Since $\mathbb{P}(\mathcal{E}_0), \mathbb{P}(\mathcal{E}_1) = 1-o(1)$, it follows that
\begin{equation}
  \pr(\cE_2^c) \leq \pr(\cE_2^c \cap \cE_1 \cap \cE_0) + \pr(\cE_1^c) + \pr(\cE_0^c) = \pr(\cE_2^c \cap \cE_1 \cap \cE_0) + o(1). \label{eq:phase-II-failure}
\end{equation}
  
Note that we analyze $\pr(\cE_2^c \cap \cE_1 \cap \cE_0)$ rather than $\pr(\cE_2^c \given \cE_1, \cE_0)$, in order to preserve the data distribution. Next, we would like to show that the probability of misclassifying a vertex $v$ is $o(1/n)$, and conclude that the probability of misclassifying \emph{any} vertex is $o(1)$. To formalize such an argument, sample $N \sim \text{Poisson}(\lambda n)$, and generate $\max\{N, cn\}$ points in the region $\mathcal{S}_{d,n}$ uniformly at random. Note that on the event $\mathcal{E}_0$, we have $\max\{N, cn\} = cn$. In this way, the first $N$ points form a Poisson point process with intensity $\lambda$. We can simulate Algorithm \ref{alg:almost-exact} on the first $N$ points. To bound the failure probability of Phase II, we can assume that any $v\in \{N+1, \dots, cn\}$ must also be classified (by thresholding $\ell_i(v,x)$, computed only using edge/non-edge observations between $v$ and $u\in [N]$)).
  For $v \in [cn]$, let
  \[\mathcal{E}_2(v) = \{\widetilde{x}(v) = \omega^\star\circ x^\star(v)\}.\]
  Then
  \[\mathcal{E}_2^c \cap \mathcal{E}_1 \cap \mathcal{E}_0 \subseteq \bigcup_{v=1}^{cn} \{\mathcal{E}_2(v)^c \cap \mathcal{E}_1 \cap \mathcal{E}_0\} \subseteq \bigcup_{v=1}^{cn} \{\mathcal{E}_2(v)^c \cap \mathcal{E}_1\},\]
  so that a union bound yields
  \begin{equation}\mathbb{P}\left(\mathcal{E}_2^c \cap \mathcal{E}_1 \cap \mathcal{E}_0 \right) \leq \sum_{v=1}^{cn} \mathbb{P}\left(\mathcal{E}_2(v)^c \cap \mathcal{E}_1 \right). \label{eq:union-bound}
  \end{equation}
  Fix $v \in [cn]$. In order to bound $\mathbb{P}\left(\mathcal{E}_2(v)^c \cap \mathcal{E}_1 \right)$, we classify $v$ according to running the \texttt{Refine} algorithm with respect to edge/non-edge observations between $v$ and $u \in [N]$. Analyzing $\mathcal{E}_2(v)^c \cap \mathcal{E}_1$ now reduces to analyzing robust likelihood testing.
  Let $W(v) = \{x : \mathcal{N}(v) \to Z\cup\{*\}\}$ and $d_H$ be the Hamming distance. 
  We define the set of all estimators that differ from $\omega^\star\circ x^\star$ on at most $\beta \log n$ vertices in $\cN(v)$ as
  \$
  W'(v; \beta) &= \{x\in W(v)\colon d_H(x(\cdot), \omega^\star\circ x^\star(\cdot)) \le \beta\log n \}. 
  \$ Then $\cE_1 = \{\widehat x\in W'(v; \beta)\}$. Under Assumption \ref{ass:bounded-ratio}, for any $i\in Z$ and any $x\in W'(v; \beta)$, we have \$
\big|\ell_i(u, \omega^\star\circ x^\star) - \ell_i(u, x) \big|= \Big|\sum_{v\in V\setminus\{u\}, v\sim u}\log \frac{p_{i, \omega^\star\circ x^\star(v)} (y_{u v})}{p_{i, x(v)} (y_{u v})}\Big| \le \eta\beta\log n.
\$
We also note that $\ell_{\omega(i)}(u, \omega\circ x) = \ell_i(u, x)$ for the definition in \eqref{eq:ell_global} for any $\omega\in \Omega_{\pi, P}$, $i\in Z$, $u\in V$, and $x\colon V\to Z$.
Let $\cE_v$ be the event that there exists $x \in W'(v; \beta)$ such that the likelihood testing with respect to $x$ fails on vertex $v$; that is,
\$
  \mathcal{E}_v &= \bigcup_{i\in Z} \Big[\big\{\omega^\star\circ x^\star(v) = i \big\} \bigcap \Big(\bigcup_{x \in W'(v; \beta)} \bigcup_{j\neq i} \big\{\ell_i (v,  x) \leq \ell_j(v, x) \big\}\Big)\Big].
  \$
  It follows that \begin{equation}\mathbb{P}\left(\mathcal{E}_2(v)^c \cap \mathcal{E}_1\right) \leq \mathbb{P}(\cE_v). \label{eq:v-failure}
  \end{equation}
  
  We now aim to show that for $\beta > 0$ sufficiently small, $\mathbb{P}(\cE_v) = n^{-(1 + \Omega(1))}$. We first have
  \#\label{eq:prob_Ev}
  \mathbb{P}(\cE_v) &= \sum_{i\in Z} \pi_i \mathbb{P}(\cE_v \given \omega^{\star}\circ x^\star(v) = i).    
  \#
  For any $i\in Z$, we have
  \$
\mathbb{P}(\cE_v \given \omega^{\star}\circ x^\star(v) = i) &= \pr \Big(\bigcup_{x \in W'(v; \beta)} \bigcup_{j\neq i} \big\{\ell_i (v,  x) \leq \ell_j(v, x) \big\} \Biggiven \omega^{\star}\circ x^\star(v) = i \Big) \\
& \le \sum_{j\in Z\colon j\neq i} \pr \Big(\bigcup_{x \in W'(v; \beta)} \big\{\ell_i (v,  x) \leq \ell_j(v, x) \big\} \Biggiven \omega^{\star}\circ x^\star(v) = i \Big).
  \$
 For convenience, let $\omega^{-1}\colon Z\to Z$ be the inverse mapping of $\omega^\star$. Then for any fixed $j\neq i$, we have
  \$
& \pr \Big(\bigcup_{x \in W'(v; \beta)} \big\{\ell_i (v,  x) \leq \ell_j(v, x) \big\} \Biggiven \omega^{\star}\circ x^\star(v) = i \Big) \\
& \qquad = \pr \Big(\min_{x \in W'(v; \beta)} \big\{\ell_i (v,  x) - \ell_j(v, x)\big\}\leq 0 \Biggiven \omega^{\star}\circ x^\star(v) = i \Big) \\
& \qquad \le \pr \Big( \ell_i (v, \omega^{\star}\circ x^\star) - \ell_j(v, \omega^{\star}\circ x^\star) \leq 2\eta\beta\log n \Biggiven \omega^{\star}\circ x^\star(v) = i \Big) \\
& \qquad = \pr \Big( \ell_{\omega^{-1}(i)} (v, x^\star) - \ell_{\omega^{-1}(j)}(v, x^\star) \leq 2\eta\beta\log n \Biggiven x^\star(v) = \omega^{-1}(i) \Big) \\
& \qquad = n^{-(1+\Omega(1))},
  \$
where the last equality follows from Lemma \ref{lem:exact-bnd} when taking $0<\beta \le (\lambda \nu_d \min_{i\neq j} D_+(\theta_i, \theta_j) - 1)/(4\eta)$. By substituting the preceding two equations into \eqref{eq:prob_Ev}, we have $\mathbb{P}(\cE_v) = o(n^{-(1+\Omega(1))})$. Finally, combining \eqref{eq:phase-II-failure}, \eqref{eq:union-bound}, and \eqref{eq:v-failure}, we have $\pr(\cE_2^c) = o(1)$ and conclude the proof.
  \end{proof}

\section{Robustness Under Monotone Adversaries}
\label{sec:monotone}
We now study the robustness of our algorithm to monotone corruptions, in the GSBM. First, we show that when $k\ge 3$, the monotone adversary can readily break the distinctness assumption from Assumption \ref{ass:distinguishable}. In particular, let $a=\max_{i\in Z} a_{ii}$ and $b = \min_{i,j\in Z, i\neq j} a_{ij}$. Assumption \ref{ass:monotone} yields that $a> b$. Let $P^\prime$ be such that $P^\prime_{ii} = \text{Bern}(a)$ for any $i\in Z$ and $P^\prime_{ij} = \text{Bern}(b)$ for any $i\neq j\in Z$. The adversary can selectively add or remove edges on $\text{GSBM}(\lambda, n, \pi, P, d)$ to simulate $\text{GSBM}(\lambda, n, \pi, P^\prime, d)$ independently across edges:
\begin{enumerate}
    \item A non-edge within community $i$ is added as an edge with probability $\frac{a - a_{ii}}{1 - a_{ii}}$.
\item An edge between communities $i$ and $j$ is removed with probability $\frac{b}{a_{ij}}$.
\end{enumerate}
Thus a given pair of vertices within community $i$ is connected with probability
\[a_{ii} + (1 - a_{ii}) \cdot \frac{a - a_{ii}}{1 - a_{ii}} = a \]
and a given pair of vertices, one in community $i$ and the other in community $j$, is connected with probability
\[a_{ij} \cdot \frac{b}{a_{ij}} = b.\]

For the rest of the section, we study the $k=2$ case. In particular, we study the two-community semi-random GSBM$(\lambda, n, \pi, P, d)$, 
where $Z = \{\pm 1\}$, $\pi= (\pi_{-1}, \pi_{1})$, $P_{11}=\text{Bern}(a_{1})$, $P_{-1,-1}=\text{Bern}(a_{-1})$, and $P_{-1,1}= P_{1,-1}=\text{Bern}(b)$. Assumption \ref{ass:monotone} gives that $a_1>b$ and $a_{-1} >b$. We clarify our recovery goal as follows. When $\pi_{-1} = \pi_{1} = 1/2$, we can only hope to recover the correct partition, but not the labeling, since the adversary can add random edges to simulate $\text{GSBM}(\lambda, n, \pi, P^\prime,d)$ with $P_{-1,-1}^\prime = P_{11}^\prime$. 
When $\pi_{-1} \neq \pi_{1}$, the recovery goal is to find the correct labeling. 

We now show that when $k=2$, our algorithm achieves the recovery goal under monotone adversaries. Intuitively, the monotone adversary reinforces the local similarity within communities by adding intra-community edges and removing inter-community ones. Since the locations of vertices are unaffected by the adversary, our local labeling strategy succeeds up to the same threshold when faced with an adversary. 
We adapt Algorithm \ref{alg:almost-exact} with the following changes to label $G^\prime$. For the initial block in Line \ref{line:step2}, we use Algorithm 3 from \cite{Liu2022} to handle monotone adversaries. For the propagation in Lines \ref{line:step3}-\ref{line:step3-end}, we apply Algorithm \ref{alg:propagation-monotone} specialized for the case $k=2$. 
\begin{algorithm}
      \caption{\texttt{Propagate (two communities)}} \label{alg:propagation-monotone}
      \begin{algorithmic}[1]
      \Require{Graph $G^\prime = (V,E^\prime)$, mutually visible vertex sets $S, S' \subset V$, $S \cap S' = \emptyset$, and $\widehat{x}_S\colon S\to Z$.}
      \State{Take $j=\argmax_{r\in \{-1,1\}} |\{v\in S\colon \widehat x_S(v)=r\}|$, so that $j$ represents the largest community in $S$.}
      \For{$u \in S'$}
      \If{$|\{v \in S \colon \widehat{x}_S(v) = j, (u,v)\in E^\prime\}| \ge (a+b)\cdot|\{v \in S \colon \widehat{x}(v) = j\}|/2$}
      \State Set $\widehat{x}(u) = j$.
      \Else
      \State Set $\widehat{x}(u) = -j$.
      \EndIf
      \EndFor
      \end{algorithmic}
\end{algorithm}
We analyze each step of the algorithm in detail. We skip the connectivity of the visibility graph since it remains the same as Section \ref{sec:connectivity}.

\subsection{Phase I: Labeling the Initial Block}
We label the initial block using techniques (Algorithm 3) from \cite{Liu2022}. The following theorem guarantees the success of their algorithm on a semi-random SBM. 

\begin{theorem}\cite[Theorem 3.3]{Liu2022}\label{thm:ref-adversarial}
    There is a polynomial time algorithm that when run on a semi-random SBM with $n$ vertices and monotone changes \footnote{\cite{Liu2022} handled additional corruptions, where in addition to the monotone changes, an adversary may pick up to $\epsilon n$ nodes and modify their incident edges arbitrarily.}, edge probabilities $b'/n < a'/n \leq 1/2$ and $k = 2$ communities of sizes between ${\alpha n}/{2}$ and $n/(2\alpha)$, outputs a labeling that has expected error rate at most
    \[e^{-C/2 + O(\alpha^{-9}\sqrt{C} \log C)} + \frac{e^{-\sqrt{\log n}}}{n},\]
    where $C = (\sqrt{a'} - \sqrt{b'})^2$.
    The constants hidden in the $O(\cdot)$'s are universal constants.
\end{theorem}
We apply the stated theorem to label the initial block of $G^\prime$. Let $(V_1, E_1^\prime)$ with $E_1^\prime = \{(u, v)\colon u,v\in V_1, (u, v)\in E^\prime \}$ be the initial block. Since the theorem requires $a',b'=o(n)$, we subsample $E_1^\prime$ 
by retaining each edge with probability $s/|V_1|$ and obtain $\widehat E_1$, with $s=o(|V_1|)$ to be chosen in \eqref{eq:choose-s}. Let $a = \min\{a_{-1}, a_1\}$, and SBM$(n, \pi, W)$ be the two-community SBM with $n$ vertices, partition probability $\pi$, and edge probabilities $W$. Then $(V_1, \widehat E_1)$ can be simulated by SBM$(|V_1|, \pi, \frac{1}{|V_1|}\big(\begin{smallmatrix}as & bs\\ bs & as\end{smallmatrix}\big))$ (SBM$_1$) with monotone adversaries. 
More specifically, within the visibility radius, $(V_1, \widehat E_1)$ is equivalent to SBM$(|V_1|, \pi, \frac{1}{|V_1|}\big(\begin{smallmatrix}a_{-1}s & bs\\ bs & a_1s\end{smallmatrix}\big))$ with monotone adversaries, and if there is $a_i>a$, the later can be obtained when the adversary selectively adds edges within community $i$ to SBM$_1$, with probability $(a_i - a)/(1-a)$. Thus, we apply Theorem \ref{thm:ref-adversarial} on $(V_1, \widehat E_1)$ to label $V_1$. Let $N_{\textsf{err}}$ be the number of errors made by the algorithm. 
\begin{theorem}
For any $\delta_1>0$ and $\delta_2>0$, there is a polylogarithmic time algorithm 
that when run on the initial block of the two-community GSBM with monotone adversaries, outputs a labeling with errors 
\$
\pr(N_{\textsf{err}}\ge \delta_1 \log n) \le \delta_2.
\$
\end{theorem}
\begin{proof}
We fix $\delta\log n \le m_1\le\Delta\log n$.
Let $\alpha = \min\{\pi_1, 1/(3\pi_1), \pi_2, 1/(3\pi_2)\}$ in Theorem \ref{thm:ref-adversarial} and $\mathcal{F} = \bigcap_{i\in\{1,2\}} \{\alpha |V_1|/2 \le |u\colon x^\star(u)=i| \le |V_1|/(2\alpha)\}$ be the event that the community sizes are concentrated. We have $\alpha/2 - \pi_1 \le -\pi_1/2$, $1/(2\alpha)-\pi_1 \ge \pi_1/2$, and similarly for $\pi_2$. Thus, we obtain that, by taking $t = \pi_1/2$, 
\$
& \pr_{m_1}\big(\alpha |V_1|/2 \le |u\colon x^\star(u)=1| \le |V_1|/(2\alpha) \big) \\
&\qquad \ge \pr_{m_1}\big(-\pi_1|V_1|/2 \le |u\colon x^\star(u)=1| - \pi_1 |V_1| \le \pi_1|V_1|/2 \big)\\
&\qquad  = 1- 
\pr_{m_1}\big(\big||u\colon x^\star(u)=1| - \pi_1 |V_1|\big| \ge t|V_1| \big) \\
&\qquad \ge 1 - 2\exp(-2t^2m_1) \ge 1 - \exp(-{2t^2\delta\log n}) = 1-o(1),
\$
where Hoeffding's inequality gives the second inequality.
A similar statement holds for $|u\colon x^\star(u)=2|$. Thus, we have $\pr(\mathcal{F})=1-o(1)$. The following analysis will condition on $\mathcal{F}$. Since $(V_1, \widehat E_1)$ can be simulated by SBM$_1$ with monotone adversaries, we will apply Theorem \ref{thm:ref-adversarial} with 
$a'=as$, $b'=bs$, and $C = s(\sqrt{a} - \sqrt{b})^2$. We choose $s$ and thus $C$ sufficiently large so that
\#\label{eq:choose-s}
e^{-C/2 + O(\alpha^{-9}\sqrt{C} \log C)} \le \frac{\delta_1\delta_2}{2\Delta}.
\#
We take $n_0$ such that $e^{-\sqrt{\delta\log n_0}}/(\delta\log n_0)\le {\delta_1\delta_2}/{(2\Delta)}$. 
Then for $m_1 \geq \delta \log n$ and $n\ge n_0$, 
we have
\$
\frac{e^{-\sqrt{\log m_1}}}{m_1} \le \frac{e^{-\sqrt{\delta\log n}}}{\delta\log n} \le \frac{e^{-\sqrt{\delta\log n_0}}}{\delta\log n_0} \le \frac{\delta_1\delta_2}{2\Delta}.
\$
Conditioned on $|V_1|=m_1$ and $\mathcal{F}$, Theorem \ref{thm:ref-adversarial} ensures that there is a polylogarithmic time (polynomial time in terms of $m_1 = O(\log n)$) algorithm on $(V_1, \widehat E_1)$ that outputs a labeling with expected error rate
\$
\E_{m_1,\mathcal{F}}[N_{\textsf{err}}/m_1] \le e^{-C/2 + O(\alpha^{-9}\sqrt{C} \log C)} + \frac{e^{-\sqrt{\log m_1}}}{m_1} \le \delta_1\delta_2/\Delta.
\$
Thus, Markov's inequality gives that 
\$
\pr_{m_1,\mathcal{F}}(N_{\textsf{err}} \ge \delta_1\log n) \le \pr_{m_1,\mathcal{F}}(N_{\textsf{err}} \ge \delta_1m_1/\Delta) \le \frac{\E_{m_1,\mathcal{F}}[N_{\textsf{err}}/m_1]}{\delta_1/\Delta} \le \delta_2.
\$
Finally, the law of total probability concludes the proof. 
\end{proof}

\subsection{Phase I: Propagating Labels Among Occupied Blocks}
We now show the successful propagation. The difference between the analysis in Section \ref{subsec:propagation} and the discussion here mainly lies in Lemmas \ref{lem:degree-profile-error} and \ref{lem:degree-profile-monotone}. Here we need to additionally handle more errors from the initial block labeling ($\delta_1 \log n$ errors instead of $M$), and the monotone changes on the edges. Let $\omega^\star\in \Omega_{\pi, P}$ be the permissible relabeling for the initial block labeling.  
We define the degree profile for $G$ and $G^\prime$ as follows.
  \begin{definition}[Degree profile]
      Given $G\sim \text{GSBM}(\lambda, n, \pi, P, d)$ and $G^\prime$ after monotone changes, the \emph{degree profiles} of a vertex $u \in V$ for a reference set $S \subset V$, a labeling $x \colon S \to \{-1, 1\}$, and $r\in\{-1, 1\}$ is given by, 
      \$
      &d_{r}^+(u, x, S; G) = |\{v \in S \colon \omega^\star\circ x(v) = r, (u,v) \in E \}|, \\
      &d_{r}^+(u, x, S; G^\prime) = |\{v \in S \colon \omega^\star\circ x(v) = r, (u,v) \in E^\prime \}|.
      \$
  \end{definition}
Recall that $a=\min\{a_1, a_{-1}\}$. Without loss of generality, we assume $a_1 = a_{-1}$ in this subsection, since $G^\prime$ can be simulated by a semi-random GSBM$(\lambda, n, \pi, \widehat P, d)$ with $\widehat P_{11}=\text{Bern}(a)$ and $\widehat P_{-1,-1}=\text{Bern}(a)$. Fix a block $i \in V^{\dagger} \setminus \{i_1\}$.
Let $J_i$ be the larger community labeled by $\widehat x$ in $V_{p(i)}$:
\begin{align}
    J_i = \argmax_{r\in \{-1, 1\}} |\{v\in V_{p(i)}\colon \widehat x(v)=r\}|. \label{eqn:J_i-monotone}
\end{align}
Then Algorithm \ref{alg:propagation-monotone} labels a vertex $u\in V_i$ by thresholding the degree profile relative to the larger community: 
\$
\widehat x(u) = J_i \cdot \left(2\mathds{1}\left\{d_{J_i}^+(u, \widehat x, V_{p(i)}; G^\prime) \ge \frac{(a+b)}{2} \cdot |v\in V_{p(i)} \colon \widehat x(v) = J_i|\right\}-1\right).
\$
Let the \emph{configuration} of a block be a vector $z = (z(1,1), z(1,-1), z(-1,-1), z(-1,1))\in Z_{+}^4$, where each entry represents the count of vertices labeled as $1$ or $-1$ by $\omega^\star\circ x^\star$ and $\widehat{x}$. For example, $z(1,-1)$ is the number of vertices labeled as $1$ by $\omega^\star\circ x^\star$  but labeled as $-1$ by $\widehat{x}$. 
For $i\in V^\dagger$, the event $\cC_{i}(z)$ signifies that the occupied block $V_{i}$ possesses a configuration $z$ such that for $j,r\in \{-1,1\}$,
  \begin{align*}
  &|\{u\in V_{i}, \omega^\star\circ x^\star(u) = j, \widehat x(u) = r\}| = z(j,r).
  \end{align*}
We define positive constants $\delta_1$ and $c_2$ as follows:
\#\label{eq:define-c2}
0<\delta_1 \le \min\Big\{\frac{(a-b)\delta}{8a},  \frac{(a-b)\delta}{8(1-b)}\Big\}, \quad c_2 = \frac{(a-b)^2\delta^2}{32\Delta}.
\#
Let $\cA_{i}^{\delta_1}$ be the event that $\widehat x$ makes at most $\delta_1\log n$ mistakes on $V_{i}$:
\$
  \cA_{i}^{\delta_1} = \big\{ |\{u \in V_{i} : \widehat{x}(u) \neq \omega^\star \circ x^\star(u)\} | \leq \delta_1 \log n\big\}.
  \$
The following lemma bounds the probability of misclassifying a given vertex using Algorithm \ref{alg:propagation-monotone}. 
  \begin{lemma}\label{lem:degree-profile-monotone}
  Suppose that Assumption \ref{ass:monotone} holds. Choose constants $\delta_1$ and $c_2$ from \eqref{eq:define-c2}. Fix $z \in \mathbb{Z}_+^4$ such that $z(1,1) + z(1,-1) + z(-1,-1) + z(-1,1) = m_{p(i)}$ and $z(1,-1)  +z(-1,1) \leq \delta_1 \log n$ (so that $\cC_{p(i)}(z) \subset \cA_{p(i)}^{\delta_1}$). Then for any $u\in V_i$, we have \$
  \pr_m\big(\widehat x(u)\neq\omega^\star\circ x^\star(u)\biggiven \cC_{p(i)}(z)\big) \le n^{-c_2}.
  \$
  \end{lemma}
\begin{proof} 
  We first study the case when $J_i = 1$ for $J_i$ in \eqref{eqn:J_i-monotone}, where Algorithm \ref{alg:propagation-monotone} uses $d_{1}^+(u, \widehat x, V_{p(i)}; G^\prime)$ to label a fixed $u\in V_i$. Conditioned on any $\cC_{p(i)}(z)$, we have $|\{v \in V_{p(i)} \colon \widehat{x}(v) = 1 \}| = z(1,1) + z(-1,1)$. Among these vertices $v\in V_{p(i)}$ with $\widehat{x}(v)=1$, $z(1,1)$ vertices have ground truth label $\omega^\star\circ x^\star(u) = 1$ and $z(-1,1)$ of them have $\omega^\star\circ x^\star(u) = -1$. We now bound the probability of making a mistake, namely $\widehat x(u)\neq\omega^\star\circ x^\star(u)$.

  Recall the degree profile $d_1^+(u, \widehat x, V_{p(i)}; G) = |\{v \in V_{p(i)} : \widehat x(v) = 1, \{u,v\} \in E\}|$ for $G$. Conditioned on $\{\cC_{p(i)}(z), \omega^\star\circ x^\star(u) = 1\}$, $d_1^+(u, \widehat x, V_{p(i)}; G)$ has the same distribution as $Z=\sum_{r=1}^{z(1,1)}X_r + \sum_{r=1}^{z(-1,1)}Y_r$, where random variables $\{X_r\}_{r=1}^{z(1,1)}$ and $\{Y_r\}_{r=1}^{z(-1,1)}$ are independent, $X_r\sim\text{Bern}(a)$, and $Y_r\sim\text{Bern}(b)$. Then after the monotone changes, $d_1^+(u, \widehat x, V_{p(i)}; G^\prime)$ has the same distribution as $Z^\prime=\sum_{r=1}^{z(1,1)}X_r^\prime + \sum_{r=1}^{z(-1,1)}Y_r^\prime$, where $\{X_r^\prime\}_{r=1}^{z(1,1)}$ and $\{Y_r^\prime\}_{r=1}^{z(-1,1)}$ are independent with values on $\{0,1\}$, $X_r^\prime\ge X_r$, and $Y_r^\prime \le Y_r$. In particular, $X_r^\prime = X_r$ and $Y_r^\prime = Y_r$ for those edges in $E$, $X_r^\prime = 1$ for the added inter-community edges, and $Y_r^\prime = 0$ for the removed intra-community ones. We note that $Z^\prime\ge \sum_{r=1}^{z(1,1)}X_r^\prime \ge \sum_{r=1}^{z(1,1)}X_r$, and then obtain
\$
  &\pr_m\big(\widehat x(u)\neq\omega^\star\circ x^\star(u)\biggiven \cC_{p(i)}(z), J_i=1, \omega^\star\circ x^\star(u) = 1\big) \\
  & \qquad = \pr_m\Big(d_{1}^+(u, \widehat{x}, V_{p(i)}; G^\prime)< (a+b)|\{v\in V_{p(i)}: \widehat x(v) = 1\}|/2 \\
  & \qquad \qquad \qquad\Biggiven  \cC_{p(i)}(z), J_i=1, \omega^\star\circ x^\star(u) = 1\Big)  \\
  & \qquad = \pr_m\Big(Z^\prime < (a +b)\big(z(1,1) + z(-1,1)\big)/2 \Big) \\
  & \qquad \le \pr_m\Big(\sum_{r=1}^{z(1,1)}X_r < (a +b)\big(z(1,1) + z(-1,1)\big)/2\Big)\\
  & \qquad = \pr_m\Big(\sum_{r=1}^{z(1,1)}X_r -az(1,1)< -(a -b)z(1,1)/2 + (a +b)z(-1,1)/2\Big). 
  \$
Recall that $z(1,-1) + z(-1,1) \leq \delta_1\log n$, and $J_{i}= 1$ implies $|\{v \in V_{p(i)} : \widehat x(v) = 1 \}|\ge|V_{p(i)}|/2\ge\delta\log n/2$. It follows that $z(-1,1) \leq \delta_1 \log n$, $z(1,1) + z(-1,1)\ge\delta\log n/2$, and $z(1,1) \geq (\delta/2 - \delta_1)\log n$. 
Thus, the above gives 
\#\label{eq:propagation-mistake-monotone-1}
  &\pr_m\big(\widehat x(u)\neq\omega^\star\circ x^\star(u)\biggiven \cC_{p(i)}(z), J_i=1, \omega^\star\circ x^\star(u) = 1\big) \notag\\
  & \quad \le \pr_m\Big(\sum_{r=1}^{z(1,1)}X_r -az(1,1)< -(a -b)z(1,1)/2 + (a +b)z(-1,1)/2\Big) \notag\\
  & \quad \le \pr_m\big(\sum_{r=1}^{z(1,1)}X_r - az(1,1) < -(a-b)(\delta/2 - \delta_1)\log n/2 + (a +b)\delta_1 \log n/2\big) \notag\\
  & \quad \le \pr_m\big(\sum_{r=1}^{z(1,1)}X_r - az(1,1) < -(a-b)\delta \log n/8\big) \notag\\
  & \quad \le \exp\big(- \frac{(a-b)^2\delta^2\log^2n}{32z(1,1)}\big) \notag\\
  & \quad \le \exp\big(- \frac{(a-b)^2\delta^2\log n}{32\Delta}\big) = n^{-c_2 },
  \#
  where the third inequality holds when taking $\delta_1$ small enough so that $\delta_1 \le (a-b)\delta/(8a)$, the fourth one follows from Hoeffding's inequality, and the last one holds since $z(1,1)\le m_{i}\le \Delta\log n$.
  
  Similarly, conditioned on $\{\cC_{p(i)}(z), \omega^\star\circ x^\star(u) = -1\}$, $d_1^+(u, \widehat x, V_{p(i)}; G)\equiv\sum_{r=1}^{z(1,1)}Y_r + \sum_{r=1}^{z(-1,1)}X_r$, where $\{X_r\}_{r=1}^{z(-1,1)}$ and $\{Y_r\}_{r=1}^{z(1,1)}$ are independent, $X_r\sim\text{Bern}(a)$, and $Y_r\sim\text{Bern}(b)$. After the monotone changes, $d_1^+(u, \widehat x, V_{p(i)}; G^\prime)$ has the same distribution as $Z^\prime=\sum_{r=1}^{z(1,1)}Y_r^\prime + \sum_{r=1}^{z(-1,1)}X_r^\prime$, where $\{X_r^\prime\}_{r=1}^{z(-1,1)}$ and $\{Y_r^\prime\}_{r=1}^{z(1,1)}$ are independent on $\{0,1\}$, $X_r^\prime\ge X_r$, and $Y_r^\prime \le Y_r$. Thus, we observe $Z^\prime \le \sum_{r=1}^{z(1,1)}Y_r + z(-1,1)$ and have
  \$
  &\pr_m\big(\widehat x(u)\neq\omega^\star\circ x^\star(u)\biggiven \cC_{p(i)}(z), J_i=1, \omega^\star\circ x^\star(u) = -1\big) \\
  & \qquad = \pr_m\Big(d_{1}^+(u, \widehat{x}, V_{p(i)}; G^\prime) \ge (a+b)|\{v\in V_{p(i)}: \widehat x(v) = 1\}|/2 \\
  & \qquad \qquad \qquad \Biggiven  \cC_{p(i)}(z), J_i=1, \omega^\star\circ x^\star(u) = -1\Big)  \\
  & \qquad = \pr_m\Big(Z^\prime \ge (a +b)\big(z(1,1) + z(-1,1)\big)/2 \Big) \\
  & \qquad \le \pr_m\Big(\sum_{r=1}^{z(1,1)}Y_r + z(-1, 1) \ge (a +b)\big(z(1,1) + z(-1,1)\big)/2\Big)\\
  & \qquad \le \pr_m\Big(\sum_{r=1}^{z(1,1)}Y_r - bz(1,1) \ge (a -b)z(1,1)/2 - [1-(a +b)/2]z(-1,1)\Big). 
  \$

Since $z(-1,1) \leq \delta_1 \log n$ and $z(1,1) \geq (\delta/2 - \delta_1)\log n$, the above gives 
\#\label{eq:propagation-mistake-monotone--1}
  &\pr_m\big(\widehat x(u)\neq\omega^\star\circ x^\star(u)\biggiven \cC_{p(i)}(z), J_i=1, \omega^\star\circ x^\star(u) = -1\big) \notag\\
  & \quad \le \pr_m\Big(\sum_{r=1}^{z(1,1)}Y_r - bz(1,1) \ge (a -b)z(1,1)/2 + [(a +b)/2 - 1]z(-1,1)\Big) \notag\\
  & \quad \le \pr_m\Big(\sum_{r=1}^{z(1,1)}Y_r - bz(1,1) \ge (a -b)(\delta/2 - \delta_1)\log n/2 - [1-(a +b)/2]\delta_1\log n\Big) \notag\\
  & \quad \le \pr_m\Big(\sum_{r=1}^{z(1,1)}Y_r - bz(1,1) \ge (a -b)\delta\log n/8\Big) \notag\\
  & \quad \le \exp\big(- \frac{(a-b)^2\delta^2\log^2n}{32z(1,1)}\big) \le n^{-c_2 },
  \#
where the third inequality holds when taking $\delta_1$ small enough so that $\delta_1 \le (a-b)\delta/[8(1-b)]$, 
the fourth one follows from Hoeffding's inequality, and the last one holds since $z(1,1)\le \Delta\log n$. The bounds \eqref{eq:propagation-mistake-monotone-1} and \eqref{eq:propagation-mistake-monotone--1} together imply
\$
\pr_m\big(\widehat x(u)\neq\omega^\star\circ x^\star(u)\biggiven \cC_{p(i)}(z), J_i=1\big) \le n^{-c_2}.
\$
We can derive symmetric analysis for the case when $J_i = -1$, where Algorithm \ref{alg:propagation-monotone} uses $d_{-1}^+(u, \widehat x, V_{p(i)}; G^\prime)$ to label $u\in V_i$. This completes the proof of the lemma.
\end{proof} 

The above lemma bounds the probability of misclassifying a single vertex at $n^{-c_2}$, the same rate as Lemma \ref{lem:degree-profile-error}. Note that $\cA_{i}^{\delta_1}$ allows more errors in each block and thus holds with higher probability than $\cA_{i}$ defined in \eqref{eq:define-A}. Thus, it is easy to check that Lemma \ref{prop:phase1-occupied} holds by replacing $\cA_{i}$ with $\cA_{i}^{\delta_1}$, so does the following statement that with high probability $\widehat x$ makes at most $\delta_1\log n$ mistakes on every occupied block,
\#\label{eq:almost-exact-monotone}
\pr\Big(\bigcap_{i \in V^{\dagger}} \big\{|\{v \in V_i\colon \widehat x(v) \neq \omega^\star \circ x^\star(v)\}| \leq \delta_1\log n\big\} \Big) = 1-o(1).
\#
Note that $\delta_1<\delta$ from \eqref{eq:define-c2}. Then the proof of Theorem \ref{thm:phase1-summary} works here following from \eqref{eq:almost-exact-monotone}. Thus, $\widehat x$ achieves almost exact recovery with dispersed errors, as stated in \eqref{eq:almost-exact-recovery} and \eqref{eq:dispersed-errors}.

\subsection{Phase II: Refining the Labels}
Finally, we prove Theorem \ref{thm:monotone} that Phase II achieves exact recovery by refining the almost-exact labeling $\widehat{x}$ obtained from Phase I. 
We define the likelihood function of class $i$ with reference labeling $x$ at a vertex $u\in V$ as
\$
\ell_i(u, x; G^\prime) = \sum_{v\in V\setminus\{u\}, v\sim u}\log p_{i, x(v)} (y_{u v}^\prime).
\$
Our refinement procedure assigns, for any $u\in V$,
\$
\widetilde x(u) = \argmax_{i\in \{-1,1 \}}\ell_i(u,\widehat x; G^\prime).
\$
The following lemma states a similar result as Lemma \ref{lem:exact-bnd}.
\begin{lemma}\label{lem:exact-bnd-monotone}
    If $\lambda\nu_d D_+(\theta_{-1},\theta_1)>1$, then for a fixed $0<\epsilon \le (\lambda \nu_d D_+(\theta_{-1},\theta_1) - 1)/2$, given any $u$ with $x^\star(u)=i \in \{\pm 1\}$, it holds that 
    \$
    \pr\big(\ell_i(u, x^\star; G^\prime)-\ell_{-i}(u, x^\star; G^\prime) \le \epsilon\log n \given x^\star(u)=i \big) = n^{-(1 + \Omega(1))}.
    \$
\end{lemma}
\begin{proof}
    Fix any $u$ with $x^\star(u)=i$. For any $v\in V\setminus\{u\}$ with $v\sim u$, we recall that in graph $G$, $Y_{uv}\sim \text{Bern}(a_i)$ if $x^\star(v)=i$, and $Y_{uv}\sim \text{Bern}(b)$ if $x^\star(v)=-i$. For $r\in\{-1, 1\}$, we define degree profiles in $G$ as 
    \$
    & D_r^{+} = |\{v\in V\colon v\neq u, v\sim u, x^\star(v) = r, Y_{uv} = 1\}|, \\
    & D_r^{-} = |\{v\in V\colon v\neq u, v\sim u, x^\star(v) = r, Y_{uv} = 0\}|.
    \$
    Let $Y_{uv}^\prime$, $(D_r^{+})^\prime$, and $(D_r^{-})^\prime$ be the random variables with respect to $G^\prime$, after the monotone changes. 
    The monotonicity yields that $Y_{uv}^\prime \ge Y_{uv}$ if $x^\star(v)=i$, and $Y_{uv}^\prime \le Y_{uv}$ if $x^\star(v)=-i$. In addition, we have $(D_i^{+})^\prime \ge D_{i}^+$, $(D_i^{-})^\prime \le D_{i}^-$, $(D_{-i}^{+})^\prime \le D_{-i}^+$, and $(D_{-i}^{-})^\prime \ge D_{i}^-$.
    Now we consider the random variable $S^\prime$ defined as follows,
    \$
    S^\prime & := \ell_{-i}(u, x^\star; G^\prime)-\ell_i(u, x^\star; G^\prime) = \sum_{v\in V\setminus\{u\}, v\sim u} \log \Big(\frac{p_{-i,x^\star(v)} (Y_{uv}^\prime)}{p_{i,x^\star(v)} (Y_{uv}^\prime)}\Big) \\
    & = \sum_{v\in V\setminus\{u\}, v\sim u, x^\star(v)=i} \log \Big(\frac{p_{-i,i} (Y_{uv}^\prime)}{p_{ii} (Y_{uv}^\prime)}\Big) +  \sum_{v\in V\setminus\{u\}, v\sim u, x^\star(v)=-i} \log \Big(\frac{p_{-i,-i} (Y_{uv}^\prime)}{p_{i,-i} (Y_{uv}^\prime)}\Big) \\
    & = -\log\Big(\frac{a_i}{b} \Big)(D_{i}^+)^\prime -\log\Big(\frac{1-a_i}{1-b} \Big)(D_{i}^-)^\prime + \log\Big(\frac{a_{-i}}{b} \Big)(D_{-i}^+)^\prime +\log\Big(\frac{1-a_{-i}}{1-b} \Big)(D_{-i}^-)^\prime \\
    & \le -\log\Big(\frac{a_i}{b} \Big) D_{i}^+ -\log\Big(\frac{1-a_i}{1-b} \Big) D_{i}^- + \log\Big(\frac{a_{-i}}{b} \Big) D_{-i}^+ +\log\Big(\frac{1-a_{-i}}{1-b} \Big) D_{-i}^-,
    \$
    where the last two lines follow from the definition and properties of the degree profiles. On the other hand, when studying the original graph $G$, we have
    \$
    S & := \ell_j(u, x^\star)-\ell_i(u, x^\star) = \sum_{v\in V\setminus\{u\}, v\sim u} \log \Big(\frac{p_{j,x^\star(v)} (Y_{uv})}{p_{i,x^\star(v)} (Y_{uv})}\Big) \\
    & = \sum_{v\in V\setminus\{u\}, v\sim u, x^\star(v)=i} \log \Big(\frac{p_{-i,i} (Y_{uv})}{p_{ii} (Y_{uv})}\Big) +  \sum_{v\in V\setminus\{u\}, v\sim u, x^\star(v)=-i} \log \Big(\frac{p_{-i,-i} (Y_{uv})}{p_{i,-i} (Y_{uv})}\Big) \\
    & = -\log\Big(\frac{a_i}{b} \Big) D_{i}^+ -\log\Big(\frac{1-a_i}{1-b} \Big) D_{i}^- + \log\Big(\frac{a_{-i}}{b} \Big) D_{-i}^+ +\log\Big(\frac{1-a_{-i}}{1-b} \Big) D_{-i}^-.
    \$
Combining the preceding statements, we conclude that $S^\prime \le S$ and obtain
\$
& \pr\big(\ell_i(u, x^\star; G^\prime)-\ell_{-i}(u, x^\star; G^\prime) \le \epsilon\log n \given x^\star(u)=i \big) \\
&\qquad = \pr(S^\prime\ge -\epsilon \log n) \\
&\qquad \le \pr(S\ge -\epsilon \log n) \notag\\
    &\qquad = \pr\big(\ell_i(u, x^\star)-\ell_{-i}(u, x^\star) \le \epsilon\log n \given x^\star(u)=i \big).
    \$
We now apply Lemma \ref{lem:exact-bnd} and complete the proof of the lemma.
\end{proof}
Finally, by replacing $\ell_i(u, x^\star)$ with $\ell_i(u, x^\star;G^\prime)$ and using Lemma \ref{lem:exact-bnd-monotone} in the proof of Theorem \ref{thm:exact-reovery-result}, we complete the proof of Theorem \ref{thm:monotone}. 

\section{Exact Recovery for Gaussian models}
\label{sec:gaussian_proofs}

We prove a stronger claim than Proposition \ref{cor:gsync} by showing our algorithm achieves exact recovery for any Gaussian GHCM with an arbitrary number of communities, as long as the distinctness assumption is satisfied. For clarity and ease of computation, we let all distributions to have the same variance, though the case with different variances can be handled using similar techniques.

\begin{proposition}
    Consider $\textnormal{GHCM}(\lambda, n, \pi, P, d)$ with $Z=[k]$, for some fixed $k$, and  $P_{ij} = \text{N}(\mu_{ij}, 1)$. If $P_{ij}$'s satisfy Assumption \ref{ass:distinguishable},
    \begin{enumerate}
    \item any estimator fails at exact recovery with high probability whenever \[ \lambda \nu_d \min_{i \neq j}  D_+(\theta_i, \theta_j) < 1,\]
    or whenever $d=1$, $\lambda < 1$, and $|\Omega_{\pi, P}|\ge 2$.
    \item there exists a polynomial-time algorithm achieving exact recovery whenever
    \[ \lambda \nu_d \min_{i \neq j}  D_+(\theta_i, \theta_j) > 1, \]
    and either (1) $d\geq 2$; or (2) $d=1$ and $\lambda > 1$. 
\end{enumerate}
\end{proposition}

Observe the impossibility result is a direct consequence of Theorem \ref{theorem:impossibility}, which applies without either Assumption \ref{ass:bounded-ratio} or Assumption \ref{ass:distinguishable}. For achievability, we relax Assumption \ref{ass:bounded-ratio} by allowing Gaussian distributions. The log-likelihood ratio bound $\eta$ of Assumption \ref{ass:bounded-ratio} in Assumption was used in three parts:

\begin{enumerate}
    \item In Proposition \ref{prop:low-discrepancy}, which shows low-discrepancy labels yield lower likelihoods than $x^\star$ with high probability.
    \item In Lemma \ref{lem:degree-profile-error}, which bounds the error probability of labeling one vertex of the propagation procedure.
    \item In Theorem \ref{thm:exact-reovery-result}, which shows the refinement procedure achieves exact recovery.
\end{enumerate}

We modify the proofs of above results to show that the desired results still hold when the distributions are Gaussian, with different constant factors. We use the Gaussian tail bound throughout the analysis of this section.
\begin{lemma}
    \label{lem:gaussian_tail_bound}
    For $X \sim N(0, 1)$, $\pr(X > t) \leq  \frac{1}{t}e^{-t^2/2}$.
\end{lemma}
Moreover, the analysis for the 1-dimensional case can be handled using similar techniques. As a result, the almost exact recovery result of Phase I and exact recovery result of Phase II hold for the Gaussian case. 

\subsection{Initial block labeling}
\label{subsec:gaus_initial}
We denote the following constants:

\begin{align}
&a = \max_{i, j, k, l \in Z} \bigg\{ (\mu_{ij} - \mu_{il} + \mu_{kl} - \mu_{kj})\mu_{ij} - \frac{1}{2} (\mu_{ij}^2 - \mu_{il}^2 + \mu_{kl}^2 - \mu_{kj}^2) \bigg\}, \label{eq:gaussian_a} \\
&a_+ = \max(a, 0), \quad b = \max_{i, j, k, l \in Z} (\mu_{ij} - \mu_{il} + \mu_{kl} - \mu_{kj})^2, \quad c_1 \le {\varepsilon_0\pi_{\min}\log (1/ \rho)}/{3} \label{eq:gaussian_b} \\
& c_1 > \frac{a_+ c +\sqrt{2bc}}{2}, \quad    \varepsilon \le \frac{1}{3}\min\Big\{ \pi_{\min}, \min_{i,j\in Z, \pi_i \neq \pi_j}\{ |\pi_i - \pi_j| \}, \frac{c}{(k-1)\varepsilon_0 } \Big\}. \label{eq:epsilon_gaussian}
\end{align}

\begin{proposition}[Low discrepancy]\label{prop:low-discrepancy-gaussian}
Suppose $P_{ij} \sim N(\mu_{ij}, 1)$ satisfying Assumption \ref{ass:distinguishable}. Then, with $\varepsilon$ defined in \eqref{eq:epsilon_gaussian}, there exists $c\in(0, 1)$ satisfying the condition in \eqref{eq:epsilon_gaussian} such that  with high
probability
\$
\forall x\colon V_{i_0}\to Z \text{ such that } 0<\textup{DISC}(x, x^\star) <c\log n, \text{ it holds that } \ell_0(G, x) < \ell_0(G, x^\star).
\$
\end{proposition}

\begin{proof}
    Consider the same decomposition of $\ell_0(G, x) - \ell(G, x^\star) = A(x) + B_1(x) + B_2(x)$ as done in the proof of Proposition \ref{prop:low-discrepancy}. The bounds on $A(x)$ and $B_1(x)$ still hold, meaning $A(x)=B_1(x) \leq -d_H(x, x^\star)c_1\log n$ with high probability. We seek to bound 
    \begin{align*}
        B_2(x) &:= \sum_{\substack{v\in V_{i_0}\\ v\colon x(v)\neq x^\star(v)}}\sum_{\substack{u\neq v\in V_{i_0} \\ u\colon x(u)\neq x^\star(u) }} \bigg[\log \Big(\frac{p_{x^\star(u), x^\star(v)}(y_{uv})}{p_{x^\star(u), x(v)}(y_{uv})}\Big) + \log \Big( \frac{p_{x(u), x(v)}(y_{uv})}{p_{x(u), x^\star(v)}(y_{uv})}\Big) \bigg] \\
        &= -\frac{1}{2} \sum_{\substack{v\in V_{i_0}\\ v\colon x(v)\neq x^\star(v)}}\sum_{\substack{u\neq v\in V_{i_0} \\ u\colon x(u)\neq x^\star(u) }} \bigg[  \Big(y_{uv} - \mu_{x^\star(u), x^\star(v)}\Big)^2
        - \Big(y_{uv} - \mu_{x^\star(u), x(v)}\Big)^2 \\
        &\quad\quad + \Big(y_{uv} - \mu_{x(u), x(v)}\Big)^2 - \Big(y_{uv} - \mu_{x(u), x^\star(v)}\Big)^2  \bigg] \\
        &=-\frac{1}{2} \sum_{\substack{v\in V_{i_0}\\ v\colon x(v)\neq x^\star(v)}}\sum_{\substack{u\neq v\in V_{i_0} \\ u\colon x(u)\neq x^\star(u) }} \bigg[ -2a_{uv} y_{uv} +  b_{uv}   \bigg] \\
       &= \sum_{\substack{v\in V_{i_0}\\ v\colon x(v)\neq x^\star(v)}}\sum_{\substack{u\neq v\in V_{i_0} \\ u\colon x(u)\neq x^\star(u) }} \bigg[ a_{uv} y_{uv} -  \frac{1}{2}b_{uv}   \bigg] ,
    \end{align*}
where $a_{uv} = \mu_{x^\star(u), x^\star(v)} - \mu_{x^\star(u), x(v)} +\mu_{x(u), x(v)} - \mu_{x(u), x^\star(v)}$ and $b_{uv} =  \mu_{x^\star(u), x^\star(v)}^2 - \mu_{x^\star(u), x(v)}^2 +\mu_{x(u), x(v)}^2 - \mu_{x(u), x^\star(v)}^2$. 

Observe, since $y_{uv} \sim N(\mu_{x^\star(u), x^\star(v)}, 1)$, we have that $B_2(x) \sim N(a_x, b_x^2)$, where
\begin{align*}
    a_x &=  \sum_{\substack{v\in V_{i_0}\\ v\colon x(v)\neq x^\star(v)}}\sum_{\substack{u\neq v\in V_{i_0} \\ u\colon x(u)\neq x^\star(u) }} a_{uv}\mu_{x^\star(u), x^\star(v)} -\frac{1}{2}b_{uv}, \\
    b_x &=  \sum_{\substack{v\in V_{i_0}\\ v\colon x(v)\neq x^\star(v)}}\sum_{\substack{u\neq v\in V_{i_0} \\ u\colon x(u)\neq x^\star(u) }} a_{uv}^2.
\end{align*}

By Lemma \ref{lem:gaussian_tail_bound}, we have that
\begin{align}
    \pr\bigg( \frac{B_2(x) - a_x}{\sqrt{b_x}} > t \bigg) \leq \frac{1}{t} e^{-t^2/2}. \nonumber
\end{align}
Fix $t=2\sqrt{c}\log n$. As a result, $B_2(x) > 2\sqrt{b_xc} \log n + a_x$ holds with probability at most $\frac{1}{2\sqrt{c}\log n}n^{-2c\log n}.$ Denote $d_H$ as shorthand for $d_H(x, x^\star)$.
Then, note $b_x \leq d_H(d_H-1)b$ and $a_x \leq d_H (d_H-1)a$ for any $x$, as defined in \eqref{eq:gaussian_a} and \eqref{eq:gaussian_b}. Therefore, union bounding over all low-discrepancy $x$ yields that, for any low-discrepancy $x$, $B_2(x) > \sqrt{d_H (d_H-1)b} \sqrt{2c}\log n + d_H (d_H-1)a$ holds with probability at most 
\begin{align*}
    &c\log n \binom{n}{c \log n}(k-1)^{c\log n}\frac{1}{2\sqrt{c}\log n}n^{-2c\log n}\\
    &\qquad \leq c\log n \big(\frac{en}{c\log n} \big)^{c\log n} (k-1)^{c\log n} \frac{1}{2\sqrt{c}\log n}n^{-2c\log n} \\ 
    &\qquad = o(1),
\end{align*}
where the factor of $c\log n$ is due to $x$ having at most $c\log n$ mistakes and the $\binom{n}{c \log n}(k-1)^{c\log n}$ factor is an upper bound on the total number of distinct labels with some fixed $k$ mistakes, for $0 < k \leq c\log n$. Therefore, with high probability, 
\begin{align}
    \ell_0(G, x) - \ell(G, x^\star) &= A(x) + B_1(x) + B_2(x) \nonumber \\
    &\leq -2 d_H c_1 \log n + \sqrt{d_H (d_H-1)b} \sqrt{2c}\log n + d_H (d_H-1)a \nonumber \\
    &\leq -2 d_H c_1 \log n + d_H\sqrt{b} \sqrt{2c}\log n + d_H^2 a_+ \label{eq:tail}
\end{align}
Since $d_H \leq c\log n$, \eqref{eq:tail} is further bounded by
\begin{equation}
    \log^2 n (-2cc_1+c\sqrt{2bc} + a_+ c^2). \label{eq:tail2}
\end{equation}

The condition \eqref{eq:epsilon_gaussian} ensures \eqref{eq:tail2} is negative, which ensures $\ell_0(G, x) - \ell(G, x^\star) < 0$ with high probability.
\end{proof}

\subsection{Propagation}
\label{subsec:gaus_prop}
Recall, from Appendix \ref{subsec:propagation}, the events $\cV(m)$, $H$, and $z$, which indicate the number of vertices in all blocks, connectivity of the visibility graph, and the configuration of all blocks, respectively. Further recall the constants $\phi_t$, $\Phi_t$, $c_2$, and $M$ in \eqref{eq:define-phi-t}, \eqref{eqn:Phi_t}, and \eqref{eqn:eta_2}. We re-define $\eta_2$ such that
\begin{align}
    & a = \max_{\substack{i, j, r, s\in Z, \\ i\neq j, r\neq s}} 2|\mu_{ir}| \cdot |\mu_{js}-\mu_{jr}|, \quad b = \max_{\substack{j, r, s\in Z, \\  r\neq s}} |\mu_{js}^2 - \mu_{jr}^2|, \label{eq:ab_gaussian} \\
    &c = \max_{\substack{j, r, s\in Z, \\  r\neq s}}2(\mu_{js}-\mu_{jr})^2, \quad \eta' = a+b+c, \quad \eta_2 = k(e^{\eta'}/\Phi_t)^M. \label{eq:ceta_gaussian} 
\end{align}

It follows that the misclassification probability of a vertex using Algorithm $\ref{alg:propagation}$ follows a similar upper bound in the Gaussian case as Lemma \ref{lem:degree-profile-error}. As a result, combining Proposition \ref{prop:low-discrepancy-gaussian} and Lemma \ref{lem:gaussian-degree-profile-error} with the remaining analysis in Appendix \ref{alg:almost-exact} implies the labeling produced by Algorithm \ref{alg:almost-exact-1d} achieves almost exact recovery in the Gaussian GHCM.

  \begin{lemma}\label{lem:gaussian-degree-profile-error}
    Suppose $P_{ij} \sim N(\mu_{ij}, 1)$ satisfy Assumption \ref{ass:distinguishable}. Fix $i\in V^\dagger\setminus\{i_1\}$. Fix $z \in \mathbb{Z}_+^{2k}$ such that $\sum_{j\in Z}(z(j,+) + z(j,-)) = m_{p(i)}$ and $\sum_{j\in Z}z(j,-) \leq M$. Then for any $u\in V_i$, we have \$
    \pr_m\big(\widehat x(u) \neq \omega^\star \circ x^\star(u)\biggiven \cC_{p(i)}(z)\big) \le \eta_2n^{-c_2}.
    \$
    \end{lemma}

\begin{proof}
    Fix $i\in V^\dagger\setminus\{i_1\}$ and recall that $J_i$ in \eqref{eqn:J_i} is the largest community labeled by $\widehat{x}$ in $V_{p(i)}$. The proof of Lemma \ref{lem:degree-profile-error} yields the bounds  $z(j,-) \le M$ and $z(j, +) \ge m_{p(i)}/k - M \ge \delta\log n/k - M$ for some $j \in Z$. Let $J_i=j$.

    We fix any $u\in V_i$ and recall $\mathbb{P}_{m, z, j, r}(\cdot) = \mathbb{P}( \cdot \given \cV(m), \cC_{p(i)}(z), J_i = j, \omega^\star \circ x^\star(u) = r)$. From Lemma \ref{lem:degree-profile-error}, for any $t\in (0, 1)$ an any $s\neq r$, we have
    \begin{align}
        & \pr_{m,z,j,r}\big(\ell(u, \widehat x, V_{p(i)}; j,s) \ge \ell(u, \widehat x, V_{p(i)}; j,r)\big) \nonumber \\
    & \qquad \le \Phi_t^{z(j,+)} \cdot  \prod_{\substack{v\in V_{p(i)}\\ \widehat x(v)=j \\  \omega^\star \circ x^\star(v) \neq j}}   \E_{m,z,j,r}\bigg[\exp\bigg(t  \log \Big(\frac{p_{j s}(y_{uv})}{p_{j r}(y_{uv})}\Big)\bigg) \bigg]. \label{eq:mgf_gaus}
    \end{align}

    Next, we bound the expectation term of \eqref{eq:mgf_gaus}. For each $v$ in the product, observe there is a $j'\neq j$ with $\omega^\star \circ x^\star(v) = j'$ such that $y_{uv} \sim N(\mu_{j'r}, 1)$, which implies
    \begin{align*}
         \log \Big(\frac{p_{j s}(y_{uv})}{p_{j r}(y_{uv})}\Big) &= 2 y_{uv}(\mu_{js}-\mu_{jr}) - (\mu_{js}^2 - \mu_{jr}^2) \\
         &\sim N(2\mu_{j'r}(\mu_{js}-\mu_{jr})- (\mu_{js}^2 - \mu_{jr}^2), 4(\mu_{js}-\mu_{jr})^2).
    \end{align*}
    Each expectation is, therefore, the MGF of a Gaussian with mean $2\mu_{j'r}(\mu_{js}-\mu_{jr})- (\mu_{js}^2 - \mu_{jr}^2)$ and variance $ 4(\mu_{js}-\mu_{jr})^2$. Recalling the MGF of $N(\mu, \sigma^2)$ is $\exp(\mu t + \sigma^2t^2 /2)$ and using the constants defined in \eqref{eq:ab_gaussian} and \eqref{eq:ceta_gaussian}, we have
    \begin{align*}
        & \E_{m,z,j,r}\bigg[\exp\bigg(t  \log \Big(\frac{p_{j s}(y_{uv})}{p_{j r}(y_{uv})}\Big)\bigg) \bigg] \\
        &\qquad = \exp\bigg( t\big(2\mu_{j'r}(\mu_{js}-\mu_{jr}) - (\mu_{js}^2 - \mu_{jr}^2)\big) + 2t^2 (\mu_{js}-\mu_{jr})^2 \bigg) \\
        &\qquad \leq \exp(t(a + b)+t^2c) \\
        &\qquad \leq \exp(\eta').
    \end{align*}
    It follows that, by the same argument in Lemma \ref{lem:degree-profile-error},
    \begin{align*}
        & \pr_{m,z,j,r}\big(\ell(u, \widehat x, V_{p(i)}; j,s) \ge \ell(u, \widehat x, V_{p(i)}; j,r)\big) \le \Phi_t^{z(j,+)} e^{\eta' z(j, -)} \leq (\eta_2 /k ) \cdot n^{-c_2}.
    \end{align*}
Union bounding over all $\neq r$ concludes the proof.
\end{proof}

\subsection{Refinement} In this section, we prove the \texttt{Refine} procedure achieves exact recovery in the Gaussian GHCM.

\begin{theorem}\label{thm:exact-reovery-result-gaussian}
    Suppose $P_{ij}\sim N(\mu_{ij}, 1)$ satisfy Assumption \ref{ass:distinguishable}. If $\lambda\nu_d \min_{i\neq j} D_+(\theta_i, \theta_j)>1$, then $\widetilde x$, the output of Algorithm \ref{alg:almost-exact}, achieves exact recovery.
\end{theorem}

\begin{proof}
    Recall from the proof of Theorem \ref{thm:exact-reovery-result} that we condition on the event $|V|\leq cn$ and that $\beta>0$ is a sufficiently small constant, to be determined, such that $\cE_1$ is the event that the Phase I label $\widehat x$ makes at most $\beta \log n$ mistakes in the neighborhood for all vertices. We seek to show the probability of misclassifying a vertex $v$ in Phase II is $o(1/n)$, so that the misclassification probability of any vertex is $o(1)$. Recall that $W'(v; \beta)$ is the set of all of estimators that differ from $w^\star \circ x^\star$ on at most $\beta \log n$ vertices in $\cN(v)$ and that it is sufficient to show $\mathbb{P}(\cE_v) = o(1)$, where $\cE_v$ is the event that there exists $x\in W'(v; \beta)$ such that the likelihood testing with respect to $x$ fails on $v$.

     The analysis of Theorem \ref{thm:exact-reovery-result} shows 
    \begin{align*}
        \mathbb{P}(\cE_v) &= \sum_{i\in Z} \pi_i \mathbb{P}(\cE_v \given \omega^{\star}\circ x^\star(v) = i) \\
        &\leq \sum_{i\in Z} \pi_i \sum_{j\in Z\colon j\neq i} \pr \Big(\bigcup_{x \in W'(v; \beta)} \big\{\ell_i (v,  x) \leq \ell_j(v, x) \big\} \Biggiven \omega^{\star}\circ x^\star(v) = i \Big).
    \end{align*}
    Next, we uniformly bound the above probability. Denote $\Delta_i(u, x):=\ell_i(u, \omega^\star\circ x^\star) - \ell_i(u, x)$. Observe
  \begin{align}
      & \pr \Big(\bigcup_{x \in W'(v; \beta)} \big\{\ell_i (v,  x) \leq \ell_j(v, x) \big\} \Biggiven \omega^{\star}\circ x^\star(v) = i \Big)  \label{eq:gaus_cond_prob} \\ 
& \qquad = \pr \Big(\bigcup_{x \in W'(v; \beta)} \big\{\ell_i (v, \omega^{\star}\circ x^\star) - \ell_j(v, \omega^{\star}\circ x^\star) \leq \Delta_i(v, x) - \Delta_j(v, x) \big\} \Biggiven \omega^{\star}\circ x^\star(v) = i \Big). \nonumber
  \end{align}

   For any $i,j\in Z$ where $i\neq j$ and any $x\in W'(v; \beta)$, we have
    \begin{align}
        \Delta_i(v, x) - \Delta_j(v, x)  &= \sum_{v\in V\setminus\{u\}, v\sim u}\log \frac{p_{i, \omega^\star\circ x^\star(v)} (y_{u v})}{p_{i, x(v)} (y_{u v})} - \log \frac{p_{j, \omega^\star\circ x^\star(v)} (y_{u v})}{p_{j, x(v)} (y_{u v})} \label{eq:delta_gaus} \\
        &= \sum_{v\in V\setminus\{u\}, v\sim u} 2y_{uv}(\mu_{i, \omega^\star\circ x^\star(v)} - \mu_{i, x(v)} - \mu_{j, \omega^\star\circ x^\star(v)} + \mu_{j, x(v)}) \nonumber \\
        &\qquad -(\mu_{i, \omega^\star\circ x^\star(v)}^2 - \mu_{i, x(v)}^2 -\mu_{j, \omega^\star\circ x^\star(v)}^2 + \mu_{j, x(v)}^2) \nonumber \\
        &\sim N( a_{x, ij}, b_{x, ij}),  \nonumber
    \end{align}
    where
    \begin{align*}
        &a_{x, ij} = \sum_{v\in V\setminus\{u\}, v\sim u} 2\mu_{x^\star(u), x^\star(v)}(\mu_{i, \omega^\star\circ x^\star(v)} - \mu_{i, x(v)} - \mu_{j, \omega^\star\circ x^\star(v)} + \mu_{j, x(v)}) \\
        &\qquad \qquad -(\mu_{i, \omega^\star\circ x^\star(v)}^2 - \mu_{i, x(v)}^2-\mu_{j, \omega^\star\circ x^\star(v)}^2 + \mu_{j, x(v)}^2), \\
        &b_{x, ij} =  \sum_{v\in V\setminus\{u\}, v\sim u}4(\mu_{i, \omega^\star\circ x^\star(v)} - \mu_{i, x(v)} - \mu_{j, \omega^\star\circ x^\star(v)} + \mu_{j, x(v)})^2.
    \end{align*}

By Lemma \ref{lem:gaussian_tail_bound}, we have that \[ \frac{\Delta_i(v, x) - \Delta_j(v, x) - a_{x, ij}}{\sqrt{b_{x, ij}}} > t \]
    holds with probability at most $e^{-t^2/2}/t$ for any $t>0.$ We set $t = c' \sqrt{\log(n)}$, where $c'>0$ is a constant, and define the constants
    \begin{align*}
        &a = \max_{i, j, k, l, m \in Z, i
        \neq j}\bigg\{ 2\big|\mu_{kl}\big|\cdot \big|(\mu_{i, \omega^\star(l)} - \mu_{im} - \mu_{j, \omega^\star(l)} + \mu_{jm}\big| + \big|\mu_{i, \omega^\star(l)}^2 - \mu_{im}^2-\mu_{j, \omega^\star(l)}^2 + \mu_{jm}^2\big| \bigg\}, \\
        &b =  \max_{i, j, k, l \in Z, i\neq j}4(\mu_{ik} - \mu_{il} - \mu_{jk} + \mu_{jl})^2.
    \end{align*}
    Observe $a_{x, ij} \leq a \beta \log n $ and $b_{x, ij} \leq b \beta \log n $ for any $x$ with at most $\beta \log n $ mistakes and any $i, j\in Z, i\neq j$. Therefore, we have that
    \begin{equation}
        \mathbb{P}(\Delta_i(v, x) - \Delta_j(v, x) > (a \beta  + c' \sqrt{b \beta }) \log n) \leq \frac{1}{c'\sqrt{\log(n)}}e^{-c'^2\log(n)/2}. \label{eq:X_x_tail}
    \end{equation} 

   For any fixed $j\neq i$, we upper bound \eqref{eq:gaus_cond_prob} using the law of total probability to yield

  \begin{align}
      &\pr \Big(\bigcup_{x \in W'(v; \beta)} \big\{\ell_i (v, \omega^{\star}\circ x^\star) - \ell_j(v, \omega^{\star}\circ x^\star) \leq \Delta_i(v, x) - \Delta_j(v, x) \big\} \Biggiven \omega^{\star}\circ x^\star(v) = i \Big) \nonumber \\
      &\qquad \leq \pr \Big( \ell_i (v, \omega^{\star}\circ x^\star) - \ell_j(v, \omega^{\star}\circ x^\star) \leq (a \beta  + c' \sqrt{b \beta }) \log n \Biggiven \omega^{\star}\circ x^\star(v) = i \Big)  \label{eq:gaus_test1} \\
      &\qquad \qquad + \pr\Big(\bigcup_{x \in W'(v; \beta)} \big\{\Delta_i(v, x) - \Delta_j(v, x) > (a \beta  + c' \sqrt{b \beta }) \log n\big\}\Big). \label{eq:gaus_test2} 
  \end{align}

We bound \eqref{eq:gaus_test1} using Lemma \ref{lem:exact-bnd}.
Recall that $\ell_{\omega(i)}(u, \omega\circ x) = \ell_i(u, x)$ for the definition in \eqref{eq:ell_global} for any $\omega\in \Omega_{\pi, P}$, $i\in Z$, $u\in V$, and $x\colon V\to Z$. Letting $\omega^{-1}\colon Z\to Z$ be the inverse mapping of $\omega^\star$, \eqref{eq:gaus_test1} reduces to 
\begin{align*}
     &\pr \Big( \ell_i (v, \omega^{\star}\circ x^\star) - \ell_j(v, \omega^{\star}\circ x^\star) \leq (a \beta  + c' \sqrt{b \beta }) \log n  \Biggiven \omega^{\star}\circ x^\star(v) = i \Big)  \\
     &\qquad = \pr \Big( \ell_{\omega^{-1}(i)} (v, x^\star) - \ell_{\omega^{-1}(j)}(v, x^\star) \leq (a \beta  + c' \sqrt{b \beta }) \log n \Biggiven x^\star(v) = \omega^{-1}(i) \Big) \\
& \qquad = n^{-(1+\Omega(1))},
\end{align*}
where the last equality follows from Lemma \ref{lem:exact-bnd} by taking \[ a \beta  + c' \sqrt{b \beta } < (\lambda \nu_d \min_{i\neq j} D_+(\theta_i, \theta_j) - 1)/2.\]

We bound \eqref{eq:gaus_test2} using \eqref{eq:X_x_tail} and union bounding over $x$. By similar reasoning as in Appendix \ref{subsec:gaus_initial}, the number of labels with at most $\beta \log n$ mistakes is upper bounded by
\[ \beta \log n \cdot \binom{c\log n}{\beta \log n}(k-1)^{\beta \log n} \leq \beta \log n \bigg( \frac{ec(k-1)}{\beta} \bigg)^{\beta\log n}\]

Union bounding yields
    \begin{align}
        &\pr\Big(\bigcup_{x \in W'(v; \beta)} \big\{\Delta_i(v, x) - \Delta_j(v, x) > (a \beta  + c' \sqrt{b \beta }) \log n \big\}\Big) \nonumber \\
        &\qquad \leq \frac{\beta \sqrt{\log n}}{c'}\exp\big( -c'^2\log(n) / 2 + \beta\log(n) + \beta \log(c(k-1)/\beta) \log(n)\big). \label{eq:gaus_refine_bound}\end{align} 

    Therefore, setting 
    \[ \frac{c'^2}{2} > 1 + \beta +\beta \log(c(k-1)/\beta), \]
    ensures \eqref{eq:gaus_refine_bound} is $n^{-(1+\Omega(1))}$.

    Combining the two bounds, we choose sufficiently small $\beta$ and a corresponding constant $c'$ such that
    \begin{align}
        a \beta  + c' \sqrt{b \beta } < \frac{\lambda \nu_d \min_{i\neq j} D_+(\theta_i, \theta_j) - 1}{2} \label{eq:gaus_cond1_refine} \\
         \frac{c'^2}{2} > 1 + \beta +\beta \log(c(k-1)/\beta). \label{eq:gaus_cond2_refine}
    \end{align}
    Taking $c'$ to be any constant greater than $\sqrt{2}$ allows $\beta$ to be driven sufficiently small to satisfy \eqref{eq:gaus_cond1_refine} and \eqref{eq:gaus_cond2_refine}. Then, we have $\pr(\cE_v) = o(n^{-(1+\Omega(1))})$, which concludes the proof.
\end{proof}

\end{document}